\documentclass[12pt]{article}
\usepackage{latexsym, amsmath, amsthm, amsfonts}
\usepackage{bbm}
\usepackage{enumerate, amssymb, xspace}
\usepackage[mathscr]{eucal}
\usepackage{graphicx}
\usepackage{rotating}
\usepackage{hyperref}
\usepackage{color,colordvi}

\usepackage{mathtools}   

\usepackage{tikz}
\usepackage{tikz-cd}
\usetikzlibrary{matrix,tqft,calc,decorations,arrows,shapes}
\usetikzlibrary{decorations.markings}
\usetikzlibrary{decorations.pathmorphing}
\usetikzlibrary{arrows.meta}  
\usepackage{braids}  
\usepackage{makeidx}
\usepackage{pifont}
\usepackage{psfrag}
\usepackage{stackrel}

\def\StdwidthBdy       {1.3pt}  
\def\StdwidthBdyfield  {0.88}   
\def\StdwidthEllipseX  {0.56}   
\def\StdwidthEllipseY  {0.31}   
\def\StdwidthFieldline {1.5pt}  
\def\StdwidthSepline   {0.8pt}  
\def\widthboxline      {1.5pt}
\def\ResetToStdWidths{%
    \def\widthBdy       {\StdwidthBdy}       
    \def\widthBdyfield  {\StdwidthBdyfield}  
    \def\widthEllipseX  {\StdwidthEllipseX}  
    \def\widthEllipseY  {\StdwidthEllipseY}  
    \def\widthFieldline {\StdwidthFieldline} 
    \def\widthSepline   {\StdwidthSepline}   
    }
\ResetToStdWidths

\def\widthObj       {1.9pt}  
\def\colorBdy       {blue!81!black}
\def\colorBdyfieldB {brown!71!blue}
\def\colorBdyfield  {brown!81}
\def\colorBdyLite   {blue!26}
\def\colorBulkfield {brown!81}
\def\colorBulkfieldDark {brown!71!black}
\def\colorC         {green!61!black}

\def\colorDefect    {purple!80}
\def\colorDeffield  {brown!81}
\def\colorFace      {yellow!77!brown}
\def\colorFaceDark  {yellow!64!brown}
\def\colorFaceLite  {yellow!93!brown}
\def\colorHole      {yellow!11}
\def\colorHoleDark  {yellow!41}
\def\colorHoleLite  {yellow!36}
\def\colorM         {blue!81!black}
\def\colorSepline   {yellow!72!black}
\def\colorZ         {green!47!black}
\def\arrowDefect    {stealth'} 
\newcommand\scopeArrow[2] {\begin{scope}[decoration={markings,mark=at position #1
                 with \arrow{#2}}]}  

\newcommand\Bdyfield[3]{%
     \begin{scope}[shift={(#1,#2)}]
     \draw[line width=\widthFieldline,\colorBulkfield,rotate=#3]
     (-0.5*\widthBdyfield,0) -- (0.5*\widthBdyfield,0) ; \end{scope}}
\newcommand\BdyfieldBroad[3]{%
     \begin{scope}[shift={(#1,#2)}]
     \draw[line width=\widthFieldline,\colorBulkfield,rotate=#3]
     (-0.8*\widthBdyfield,0) -- (0.8*\widthBdyfield,0) ; \end{scope}}
\newcommand\BdyfieldLabeled[5]{%
     \begin{scope}[shift={(#1,#2)}]
     \draw[line width=\widthFieldline,\colorBulkfield,rotate=#3]
     (-0.5*\widthBdyfield,0) -- (0.5*\widthBdyfield,0) ;
     \node[rotate=#3,yshift=8pt,\colorBulkfieldDark] {$\BBx^{#4,#5}$} ;
     \end{scope}}
\newcommand\BdyfieldLabeledBottom[5]{%
     \begin{scope}[shift={(#1,#2)}]
     \draw[line width=\widthFieldline,\colorBulkfield,rotate=#3]
     (-0.5*\widthBdyfield,0) -- (0.5*\widthBdyfield,0) ;
     \node[rotate=#3,yshift=-8pt,\colorBulkfieldDark] {$\BBx^{#4,#5}$} ;
     \end{scope}}
\newcommand\Bulkfield[3]{%
     \begin{scope}[shift={(#1,#2)}]
     \filldraw[fill=\colorHole,line width=\widthFieldline,draw=\colorBulkfield,rotate=#3]
     (0,0) circle (\widthEllipseX cm and \widthEllipseY cm) ; \end{scope}}
\newcommand\BulkfieldLite[3]{%
     \begin{scope}[shift={(#1,#2)}]
     \filldraw[fill=\colorFaceLite,line width=\widthFieldline,draw=\colorBulkfield,rotate=#3]
     (0,0) circle (\widthEllipseX cm and \widthEllipseY cm) ; \end{scope}}
\newcommand\BulkfieldLabeled[3]{%
     \begin{scope}[shift={(#1,#2)}]
     \filldraw[fill=\colorHole,line width=\widthFieldline,draw=\colorBulkfield,rotate=#3]
     (0,0) circle (\widthEllipseX cm and \widthEllipseY cm) node {$\FFx$} ; \end{scope}}
\newcommand\DeffieldLabeled[5]{%
     \begin{scope}[shift={(#1,#2)}]
     \filldraw[fill=\colorHole,line width=\widthFieldline,draw=\colorBulkfield,rotate=#3]
     (0,0) circle (\widthEllipseX cm and \widthEllipseY cm)
     node [rotate=#3,yshift=-17pt,\colorBulkfieldDark] {$\DDx^{#4,#5}_{}$} ; \end{scope}}
\newcommand\DeffieldLabeledTop[5]{%
     \begin{scope}[shift={(#1,#2)}]
     \filldraw[fill=\colorHole,line width=\widthFieldline,draw=\colorBulkfield,rotate=#3]
     (0,0) circle (\widthEllipseX cm and \widthEllipseY cm)
     node [rotate=#3,yshift=13pt,\colorBulkfieldDark] {$\DDx^{#4,#5}_{}$} ; \end{scope}}

\newcommand\ThreeBdyFieldsOnDiskTopLabeled[4]{ 
  \def\BBx{#4}
  \fill[\colorFace]
     (-\locpX+0.38*\widthBdyfield,\locpY+0.33*\widthBdyfield)
     [out=-55,in=-125] to node[sloped,midway,yshift=-8pt,color=\colorM] {$#3$}
     (\locpX-0.38*\widthBdyfield,\locpY+0.33*\widthBdyfield)
     -- (\locpX+0.35*\widthBdyfield,\locpY-0.35*\widthBdyfield)
     [out=-135,in=95] to node[sloped,midway,yshift=8pt,color=\colorM] {$#1$}
     (0.5*\widthBdyfield,0) -- (-0.5*\widthBdyfield,0)
     [out=85,in=-45] to node[sloped,midway,yshift=8pt,color=\colorM] {$#2$}
     (-\locpX-0.35*\widthBdyfield,\locpY-0.35*\widthBdyfield) -- cycle ;
  \draw[line width=\widthBdy,\colorBdy]
     (-\locpX+0.38*\widthBdyfield,\locpY+0.33*\widthBdyfield)
          [out=-55,in=-125] to (\locpX-0.38*\widthBdyfield,\locpY+0.33*\widthBdyfield)
     (0.48*\widthBdyfield,0) [out=95,in=-135] to (\locpX+0.34*\widthBdyfield,\locpY-0.33*\widthBdyfield)
     (-0.48*\widthBdyfield,0) [out=85,in=-45] to (-\locpX-0.34*\widthBdyfield,\locpY-0.33*\widthBdyfield) ;
  \BdyfieldLabeled 0 0 0 {#2} {#1}
  \BdyfieldLabeledBottom {-\locpX} {\locpY} {40} {#3} {#2}
  \BdyfieldLabeledBottom {\locpX} {\locpY} {-40} {#1} {#3}
  }
\def\ThreeBdyFieldsOnDisk{ 
  \fill[\colorFace]
     (-\locpX+0.38*\widthBdyfield,\locpY+0.33*\widthBdyfield)
     [out=-55,in=-125] to (\locpX-0.38*\widthBdyfield,\locpY+0.33*\widthBdyfield)
     -- (\locpX+0.35*\widthBdyfield,\locpY-0.35*\widthBdyfield)
     [out=-135,in=95] to (0.5*\widthBdyfield,0) -- (-0.5*\widthBdyfield,0)
     [out=85,in=-45] to (-\locpX-0.35*\widthBdyfield,\locpY-0.35*\widthBdyfield) -- cycle ;
  \draw[line width=\widthBdy,\colorBdy]
     (-\locpX+0.38*\widthBdyfield,\locpY+0.33*\widthBdyfield)
          [out=-55,in=-125] to (\locpX-0.38*\widthBdyfield,\locpY+0.33*\widthBdyfield)
     (0.48*\widthBdyfield,0) [out=95,in=-135] to (\locpX+0.34*\widthBdyfield,\locpY-0.33*\widthBdyfield)
     (-0.48*\widthBdyfield,0) [out=85,in=-45] to (-\locpX-0.34*\widthBdyfield,\locpY-0.33*\widthBdyfield) ;
  \Bdyfield 0 0 0
  \Bdyfield {-\locpX} {\locpY} {40}
  \Bdyfield {\locpX} {\locpY} {-40}
  }
\def\TwoplusonebroadBdyFieldsOnDisk{ 
  \fill[\colorFace]
     (-\locpX+0.38*\widthBdyfield,\locpY+0.33*\widthBdyfield)
     [out=-55,in=-125] to (\locpX-0.38*\widthBdyfield,\locpY+0.33*\widthBdyfield)
     -- (\locpX+0.35*\widthBdyfield,\locpY-0.35*\widthBdyfield)
     [out=-135,in=95] to (0.8*\widthBdyfield,0) -- (-0.8*\widthBdyfield,0)
     [out=85,in=-45] to (-\locpX-0.35*\widthBdyfield,\locpY-0.35*\widthBdyfield) -- cycle ;
  \draw[line width=\widthBdy,\colorBdy]
     (-\locpX+0.38*\widthBdyfield,\locpY+0.33*\widthBdyfield)
          [out=-55,in=-125] to (\locpX-0.38*\widthBdyfield,\locpY+0.33*\widthBdyfield)
     (0.78*\widthBdyfield,0) [out=95,in=-135] to (\locpX+0.34*\widthBdyfield,\locpY-0.33*\widthBdyfield)
     (-0.78*\widthBdyfield,0) [out=85,in=-45] to (-\locpX-0.34*\widthBdyfield,\locpY-0.33*\widthBdyfield) ;
  \BdyfieldBroad 0 0 0
  \Bdyfield {-\locpX} {\locpY} {40}
  \Bdyfield {\locpX} {\locpY} {-40}
  }
\def\TwoplusoneBdyFieldsOnDisk{ 
  \fill[\colorFace]
     (-\locpX,\locpY+0.5*\widthBdyfield) [out=-35,in=-145] to (\locpX,\locpY+0.5*\widthBdyfield)
     -- (\locpX,\locpY-0.5*\widthBdyfield) [out=-185,in=95] to (0.5*\widthBdyfield,0)
     -- (-0.5*\widthBdyfield,0) [out=85,in=5] to (-\locpX,\locpY-0.5*\widthBdyfield) -- cycle ;
  \draw[line width=\widthBdy,\colorBdy]
     (-\locpX,\locpY+0.5*\widthBdyfield) [out=-35,in=-145] to (\locpX,\locpY+0.5*\widthBdyfield)
     (0.48*\widthBdyfield,0) [out=95,in=-185] to (\locpX,\locpY-0.5*\widthBdyfield)
     (-0.48*\widthBdyfield,0) [out=85,in=5] to (-\locpX,\locpY-0.5*\widthBdyfield) ;
  \Bdyfield 0 0 0
  \Bdyfield {-\locpX} {\locpY} {90}
  \Bdyfield {\locpX} {\locpY} {-90}
  }
\def\ThreeBulkFieldsOnSphere{ 
  \shade[shading=ball,left color=\colorFaceDark]
     (-\locpx+0.99*\widthEllipseX,\locpy+0.39*\widthEllipseY)
     [out=280,in=260] to (\locpx-0.99*\widthEllipseX,\locpy+0.39*\widthEllipseY)
     -- (\locpx+0.85*\widthEllipseX,\locpy-0.95*\widthEllipseY)
     [out=215,in=95] to (\widthEllipseX,0) -- (-\widthEllipseX,0)
    [out=85,in=-35] to (-\locpx-0.85*\widthEllipseX,\locpy-0.95*\widthEllipseY) -- cycle;
  \draw[line width=\widthSepline,color=\colorSepline]
     (-\locpx+1.02*\widthEllipseX,\locpy+0.43*\widthEllipseY)
         [out=280,in=260] to (\locpx-1.02*\widthEllipseX,\locpy+0.43*\widthEllipseY)
     (1.05*\widthEllipseX,0) [out=95,in=215] to (\locpx+0.87*\widthEllipseX,\locpy-0.95*\widthEllipseY)
     (-1.05*\widthEllipseX,0) [out=85,in=-35] to (-\locpx-0.87*\widthEllipseX,\locpy-0.95*\widthEllipseY);
  \Bulkfield 0 0 0
  \Bulkfield {-\locpx} {\locpy} {20}
  \Bulkfield {\locpx} {\locpy} {-20}
  }
\newcommand\FourDefectFieldsOnSphereLabeled[5]{ 
         \def\locpa {31}   
         \def\locpb {27}   
         \def\locpCX {0.55*\widthEllipseX}
         \def\locpCY {0.53*\widthEllipseY}
         \def\locpDX {0.35*\widthEllipseX}
         \def\locpDY {1.07*\widthEllipseY}
         \def\locpEX {0.69*\widthEllipseX}
         \def\locpEY {1.33*\widthEllipseY}
  \def\DDx{#5}
  \shade[shading=ball,left color=\colorFaceDark]
     (-\locpx-\locpEX,-\locpy+\locpEY)
     [out=\locpa,in=-\locpa] to (-\locpx-\locpEX,\locpy-\locpEY)
     -- (-\locpx+\locpEX,\locpy+\locpEY)
     [out=-90+\locpa,in=-90-\locpa] to (\locpx-\locpEX,\locpy+\locpEY)
     -- (\locpx+\locpEX,\locpy-\locpEY)
     [out=180+\locpa,in=180-\locpa] to (\locpx+\locpEX,-\locpy+\locpEY)
     -- (\locpx-\locpEX,-\locpy-\locpEY)
     [out=90+\locpa,in=90-\locpa] to (-\locpx+\locpEX,-\locpy-\locpEY) -- cycle ;
  \draw[line width=\widthSepline,color=\colorSepline]
     (-\locpx-\locpEX,-\locpy+\locpEY)
     [out=\locpa,in=-\locpa] to (-\locpx-\locpEX,\locpy-\locpEY)
     (-\locpx+\locpEX,\locpy+\locpEY)
     [out=-90+\locpa,in=-90-\locpa] to (\locpx-\locpEX,\locpy+\locpEY)
     (\locpx+\locpEX,\locpy-\locpEY)
     [out=180+\locpa,in=180-\locpa] to (\locpx+\locpEX,-\locpy+\locpEY)
     (\locpx-\locpEX,-\locpy-\locpEY)
     [out=90+\locpa,in=90-\locpa] to (-\locpx+\locpEX,-\locpy-\locpEY) ;
  \draw[line width=\widthFieldline,color=\colorDefect]
     (-\locpx-\locpDX,-\locpy+\locpDY) [out=\locpb,in=-\locpb]
     to node[midway,sloped,below=-1pt] {$#1$} (-\locpx-\locpDX,\locpy-\locpDY)
     (-\locpx+\locpCX,\locpy+\locpCY) [out=-90+\locpb,in=-90-\locpb]
     to node[midway,sloped,below=-1pt] {$#4$} (\locpx-\locpCX,\locpy+\locpCY)
     (\locpx+\locpDX,\locpy-\locpDY) [out=180+\locpa,in=180-\locpa]
     to node[midway,sloped,below=-1pt] {$#3$} (\locpx+\locpDX,-\locpy+\locpDY)
     (\locpx-\locpCX,-\locpy-\locpCY) [out=90+\locpa,in=90-\locpa]
     to node[midway,sloped,above=-2pt] {$#2$} (-\locpx+\locpCX,-\locpy-\locpCY) ;
  \DeffieldLabeled {-\locpx} {-\locpy} {-45} {#1}{#2}
  \DeffieldLabeled {\locpx} {-\locpy} {45} {#2}{#3}
  \DeffieldLabeledTop {\locpx} {\locpy} {-45} {#3}{#4}
  \DeffieldLabeledTop {-\locpx} {\locpy} {45} {#4}{#1}
  }

\def\FourBulkFieldsOnSphere{ 
         \def\locpa {31}   
         \def\locpb {27}   
         \def\locpCX {0.55*\widthEllipseX}
         \def\locpCY {0.53*\widthEllipseY}
         \def\locpDX {0.35*\widthEllipseX}
         \def\locpDY {1.07*\widthEllipseY}
         \def\locpEX {0.69*\widthEllipseX}
         \def\locpEY {1.33*\widthEllipseY}
  \shade[shading=ball,left color=\colorFaceDark]
     (-\locpx-\locpEX,-\locpy+\locpEY)
     [out=\locpa,in=-\locpa] to (-\locpx-\locpEX,\locpy-\locpEY)
     -- (-\locpx+\locpEX,\locpy+\locpEY)
     [out=-90+\locpa,in=-90-\locpa] to (\locpx-\locpEX,\locpy+\locpEY)
     -- (\locpx+\locpEX,\locpy-\locpEY)
     [out=180+\locpa,in=180-\locpa] to (\locpx+\locpEX,-\locpy+\locpEY)
     -- (\locpx-\locpEX,-\locpy-\locpEY)
     [out=90+\locpa,in=90-\locpa] to (-\locpx+\locpEX,-\locpy-\locpEY) -- cycle ;
  \draw[line width=\widthSepline,color=\colorSepline]
     (-\locpx-\locpEX,-\locpy+\locpEY)
     [out=\locpa,in=-\locpa] to (-\locpx-\locpEX,\locpy-\locpEY)
     (-\locpx+\locpEX,\locpy+\locpEY)
     [out=-90+\locpa,in=-90-\locpa] to (\locpx-\locpEX,\locpy+\locpEY)
     (\locpx+\locpEX,\locpy-\locpEY)
     [out=180+\locpa,in=180-\locpa] to (\locpx+\locpEX,-\locpy+\locpEY)
     (\locpx-\locpEX,-\locpy-\locpEY)
     [out=90+\locpa,in=90-\locpa] to (-\locpx+\locpEX,-\locpy-\locpEY) ;
  \Bulkfield {-\locpx} {-\locpy} {-45}
  \Bulkfield {\locpx} {-\locpy} {45}
  \Bulkfield {\locpx} {\locpy} {-45}
  \Bulkfield {-\locpx} {\locpy} {45}
  }

\newcommand\CardyPipeTopLabeled[3]{ 
  \def\BBx{#2}
  \def\FFx{#3}
  \def\widthBdyfield{2*\widthEllipseX}
  \shade[shading=ball,left color=\colorFaceDark]
     (-\widthEllipseX,0) -- ++(0,\locpz) -- ++(2*\widthEllipseX,0) -- ++(0,-\locpz) -- cycle ;
  \fill[\colorHole]
     (\widthEllipseX,0) [out=130,in=0] to (0,\locph)
     [out=180,in=50] to (-\widthEllipseX,0) -- cycle ;
  \draw[line width=\widthSepline,color=\colorSepline]
     (\widthEllipseX,0) [out=90,in=270] to +(0,\locpz)
     (-\widthEllipseX,0) [out=90,in=270] to (-\widthEllipseX,\locpz) ;
  \draw[line width=\widthBdy,\colorBdy]
     (\widthEllipseX,0) [out=130,in=0] to (0,\locph) node [below=-1pt] {$#1$}
     (-\widthEllipseX,0) [out=50,in=180] to (0,\locph) ;
  \BulkfieldLabeled 0 {\locpz} 0
  \BdyfieldLabeled 0 0 0 {#1} {#1}
  }
\def\CardyPipe{ 
\def\widthBdyfield{2*\widthEllipseX}
  \shade[shading=ball,left color=\colorFaceDark]
     (-\widthEllipseX,0) -- ++(0,\locpz) -- ++(2*\widthEllipseX,0) -- ++(0,-\locpz) -- cycle ;
  \fill[\colorHole]
     (\widthEllipseX,0) [out=130,in=0] to (0,\locph)
     [out=180,in=50] to (-\widthEllipseX,0) -- cycle ;
  \draw[line width=\widthSepline,color=\colorSepline]
     (\widthEllipseX,0) [out=90,in=270] to +(0,\locpz)
     (-\widthEllipseX,0) [out=90,in=270] to (-\widthEllipseX,\locpz) ;
  \draw[line width=\widthBdy,\colorBdy]
     (\widthEllipseX,0) [out=130,in=0] to (0,\locph)
     (-\widthEllipseX,0) [out=50,in=180] to (0,\locph) ;
  \Bulkfield 0 {\locpz} 0
  \Bdyfield 0 0 0
  }
\def\CardyPipeFilledBdy { 
\def\widthBdyfield{2*\widthEllipseX}
  \shade[shading=ball,left color=\colorFaceDark]
     (-\widthEllipseX,0) -- ++(0,\locpz) -- ++(2*\widthEllipseX,0) -- ++(0,-\locpz) -- cycle ;
  \fill[\colorFace]
     (\widthEllipseX,0) [out=130,in=0] to (0,\locph)
     [out=180,in=50] to (-\widthEllipseX,0) -- cycle ;
  \draw[line width=\widthSepline,color=\colorSepline]
     (\widthEllipseX,0) [out=90,in=270] to +(0,\locpz)
     (-\widthEllipseX,0) [out=90,in=270] to (-\widthEllipseX,\locpz) ;
  \draw[line width=\widthBdy,\colorBdy]
     (\widthEllipseX,0) [out=130,in=0] to (0,\locph)
     (-\widthEllipseX,0) [out=50,in=180] to (0,\locph) ;
  \Bulkfield 0 {\locpz} 0
  \Bdyfield 0 0 0
  }

\newcommand\CardyPipeFilledBulkLabeled[2] { 
\def\widthBdyfield{2*\widthEllipseX}
  \def\BBx{#2}
  \Bulkfield 0 {\locpz} 0
  \shade[shading=ball,left color=\colorFaceDark]
     (-\widthEllipseX,0) -- ++(0,\locpz) arc (180:0:\widthEllipseX cm and \widthEllipseY cm)
     -- ++(0,-\locpz) -- cycle ;
  \fill[\colorHole]
     (\widthEllipseX,0) [out=130,in=0] to (0,\locph)
     [out=180,in=50] to (-\widthEllipseX,0) -- cycle ;
  \draw[line width=\widthSepline,color=\colorSepline]
     (\widthEllipseX,0) [out=90,in=270] to +(0,\locpz)
     (-\widthEllipseX,0) [out=90,in=270] to (-\widthEllipseX,\locpz) ;
  \draw[line width=\widthBdy,\colorBdy]
     (\widthEllipseX,0) [out=130,in=0] to (0,\locph) node [below=-1pt] {$#1$}
     (-\widthEllipseX,0) [out=50,in=180] to (0,\locph) ;
  \BdyfieldLabeled 0 0 0 {#1} {#1}
  }
\newcommand\LabelBdy[4]{%
     \node[rotate=#3,\colorBdy] at (#1,#2) {$#4$} ;}
\newcommand\LabelBdyField[5]{%
     \node[rotate=#3,\colorBulkfieldDark] at (#1,#2) {$\BB^{#4,#5}$} ;}
\def\Dinatmor  {
     \draw[line width=0.7*\widthBdy,draw=\colorBdyfieldB,fill=\colorBdyLite]
     (\widthDinatmor,0) [out=130,in=0] to (0,1.5*\widthDinatmor) [out=180,in=50]
     to (-\widthDinatmor,0) -- cycle ;}

\newcommand\Cite[2] {\cite[#1]{#2}}
\def\act           {\,{.}\,}
\def\Act           {{.}}
\def\Amod          {\text{mod-}A}
\def\BB            {{\mathbb B}}
\def\be            {\begin{equation}}
\def\bearl         {\begin{array}{l}}
\def\bearll        {\begin{array}{ll}}

\def\boti          {\,{\boxtimes}\,}
\def\bs            {\raisebox{3pt}{$\chi$}}
\def\C             {{\ensuremath\calc}}
\def\cala          {{\mathcal A}}

\def\calc          {{\mathcal C}}
\def\calcm         {\calc^{\star_{\phantom l}}_\Calm}
\def\calm          {{\mathcal M}}
\def\Calm          {{\!\mathcal M}}
\def\caln          {{\mathcal N}}

\def\calz          {{\mathcal Z}}
\def\cb            {\sigma}     
\def\cB            {\gamma}     
\def\Cb            {{\ensuremath{\mathcal C^\text{rev}_{\phantom|}}}}
\def\CbC           {{\ensuremath{\mathcal C^\text{rev}_{\phantom|}\boti\mathcal C}}}
\def\cir           {\,{\circ}\,}

\def\Colon         {:\quad}
\def\complex       {{\mathbbm C}}

\def\DD            {{\mathbb D}}
\def\dsty          {\displaystyle }
\def\ee            {\end{equation}}
\def\eear          {\end{array}}

\def\Enumerate     {\def\leftmargini{1.34em}~\\[-1.42em]\begin{enumerate}}
\def\eq            {\,{=}\,}
\def\ev            {{\mathrm{ev}\!}}   

\def\FF            {{\mathbb F}}
\def\Fun           {{\mathcal Fun}}
\def\Funre         {{\mathcal Rex}}

\def\GC            {{\ensuremath{\Xi_\C}}}

\def\Hom           {\mathrm{Hom}}
\def\HomC          {\ensuremath{\Hom_\calc}}
\def\HomM          {\ensuremath{\Hom_\calm}}

\def\id            {{\mathrm{id}}}
\def\Id            {{\mathrm{Id}}}

\def\iDelta        {\underline\Delta}

\def\iHom          {\underline{\Hom}}
\def\iHomM         {\underline{\Hom}_\calm}
\def\iev           {\underline{\mathrm{ev}}}
\def\ievF          {\underline{\mathrm{ev}}^\FF}

\def\imu           {\underline\mu}
\def\iN            {\,{\in}\,}
\def\Itemize       {\def\leftmargini{1.05em}~\\[-1.42em]\begin{itemize}}
\def\Itemizeiii    {\def\leftmargini{2.14em}~\\[-1.52em]\begin{enumerate}
                   \addtolength\itemsep{-6pt}}
\def\ko            {{\ensuremath{\Bbbk}}}

\def\M             {{\ensuremath\calm}}
\def\N             {{\ensuremath\caln}}
\def\Nat           {\mathrm{Nat}}

\newcommand\nxl[1] {\\[#1mm]}
\newcommand\Nxl[1] {\\[-1.3em]\\[#1mm]}
\def\Ol            {}  
\def\one           {{\bf1}}
\def\ootimes       {{\overline\otimes}}

\def\oti           {\,{\otimes}\,}

\def\ra            {^{\rm r.a.}}
\newcommand\rarr[1]{\xrightarrow{~#1~}}

\def\relNat        {\underline{\Nat}}

\def\Times         {\,{\times}\,}
\def\To            {\,{\to}\,}
\def\tr            {\mathrm{tr}}
\def\UC            {\mathrm U_\C}

\def\Vee           {{}^{\vee\!}}

\def\ZC            {{\ensuremath{\calz(\calc)}}}
\def\ZCM           {\calz(\calcm)}

\theoremstyle{plain}

\newtheorem*{Cor}{Corollary}

\theoremstyle{definition}

\newtheorem*{Defi}{Definition}
\newtheorem*{Prop}{Proposition}

\newtheorem*{Rem}{Remark}

\newtheorem*{Proposal}{Proposal}
\newtheorem{Assumption}{Assumption}

\newtheorem*{Problem}{Problem}

\begin{document}
 \numberwithin{equation}{section}

\begin{flushright}
   {\sf ZMP-HH/20-24}\\
   {\sf Hamburger$\;$Beitr\"age$\;$zur$\;$Mathematik$\;$Nr.$\;$878}\\[2mm] December 2020
\end{flushright}

\vskip 2.0em

\begin{center}
{\bf \Large Bulk from boundary in finite CFT}
\\[7pt]
{\bf \Large by means of pivotal module categories}

\vskip 18mm

{\large \  \ J\"urgen Fuchs\,$^{\,a} \quad$ and $\quad$ Christoph Schweigert\,$^{\,b}$ }

\vskip 12mm

 \it$^a$
 Teoretisk fysik, \ Karlstads Universitet\\
 Universitetsgatan 21, \ S\,--\,651\,88\, Karlstad \\[9pt]
 \it$^b$
 Fachbereich Mathematik, \ Universit\"at Hamburg\\
 Bereich Algebra und Zahlentheorie\\
 Bundesstra\ss e 55, \ D\,--\,20\,146\, Hamburg

\end{center}

\vskip 3.2em

\noindent{\sc Abstract}\\[3pt]
We present explicit mathematical structures that allow for the reconstruction of the 
field content of a full local conformal field theory from its boundary fields. Our 
framework is the one of modular tensor categories, without requiring semisimplicity, 
and thus covers in particular finite rigid logarithmic conformal field theories. We
assume that the boundary data are described by a pivotal module category over the
modular tensor category, which ensures that the algebras of boundary fields are
Frobenius algebras. Bulk fields and, more generally, defect fields inserted on defect
lines, are given by internal natural transformations between the functors that label
the types of defect lines. We use the theory of internal natural transformations to identify
candidates for operator products of defect fields (of which there are two types, either
along a single defect line, or accompanied by the fusion of two defect lines), and
for bulk-boundary OPEs. We show that the so obtained OPEs pass various consistency 
conditions, including in particular all genus-zero constraints in Lewellen's list.

  \newpage \tableofcontents \newpage


\section{Introduction}

Motivated by their fundamental importance in areas like condensed matter physics, 
statistical mechanics and string theory, two-dimensional conformal field theories --
CFTs, for short -- have been under intense scrutiny for several decades. The particular
class of rational conformal field theories, i.e.\ models for which the representations
of the chiral symmetry algebra form a semisimple modular tensor category, is now very
well understood. On the other hand, applications like the theory of critical polymers,
percolation, sandpile models and various critical disordered systems, rely on CFTs whose
chiral data (that is, fusion rules, fusing and braiding matrices, and fractional parts
of conformal weights) are encoded in a \emph{non-semisimple} tensor
category. Owing to the appearance of logarithmic branch cuts in their conformal blocks,
such chiral conformal field theories are often called \emph{logarithmic} conformal
field theories. Provided that suitable finiteness conditions are met, the tensor category
of chiral data of such CFTs is still modular, albeit non-semisimple; this is e.g.\ the 
case for the $c\,{=}\,{-}2$ CFT used in the study of critical dense polymers
\cite{dupl2,reSa4}. In this paper we restrict our attention to such models, which
still goes far beyond the rational case. Adopting the terminology of \cite{fgsS},
we refer to this class of CFTs as \emph{finite} conformal field theories.

One and the same \emph{chiral} conformal field theory can yield several different 
\emph{full} local conformal field theories. Initially, the quest for classifying the full 
CFTs that share the same chiral rational CFT concentrated on the search for modular invariants, 
i.e.\ modular invariant non-ne\-ga\-tive integral combinations of chiral characters 
with unique vacuum, corresponding to obtaining bulk fields by different ways of ``combining 
left- and right-moving degrees of freedom''. However, it was eventually recognized that
the problem of classifying modular invariants has many spurious unphysical solutions 
which cannot realize the torus partition function of a consistent full CFT (see e.g.\
\cite{fusS,gann17,soSc,davy27}). It is now known \cite{fuRs4,fjfrs,fjfrs2} that 
in the case of rational CFTs, the appropriate datum that is needed to specify a 
full conformal field theory with chiral data given by a (semisimple) modular tensor 
category \C\ is an indecomposable semisimple module category \M\ over \C.  
In the present paper, we provide evidence that, similarly, within the more general
framework of finite conformal field theories an appropriate datum is an indecomposable
\emph{pivotal module category} \M\ over the (generically non-semisimple) modular tensor
category \C\ of chiral data. This is the first result of our paper.

We arrive at this evidence by making a concrete proposal for the field content of the 
full CFT. This includes boundary fields and bulk fields, but in our context
it is most natural to admit world sheets with topological defect lines and consider also 
general defect fields which can change the type of defect line. Bulk fields can be
understood as particular defect fields, namely those which preserve the transparent 
defect line. Defect fields play an important role in applications, e.g.\ disorder fields
(defect fields on which a defect line starts or ends, meaning that it is changed to a 
transparent defect line) naturally appear as partners of bulk fields in Kramers-Wannier 
dualities. Moreover, they shed much light on the genuine mathematical structure of the theory.
The proposal for the boundary fields and defect fields is the second result of this paper. 
We furthermore show that our proposal reproduces the known field content for the case that
the category of chiral data is semisimple, and that it satisfies the genus-zero 
bulk-boundary sewing constraints. We also briefly discuss the resulting boundary states.
	     
Our final goal is to ensure, for the proposed field content, the existence of a consistent
set of correlation functions, and thereby complete the construction of a full local conformal
field theory from a given chiral theory. Several techniques for achieving this goal are 
available in the literature: using the relation with three-dimensional topological field 
theories \cite{fuRs4}, string nets \cite{scYa,traub}, or Lego-Teichm\"uller games \cite{fuSc22}.
The first two of these constructions have so far been sufficiently developed only for 
rational CFTs. Accordingly we work in the context of Lego-Teichm\"uller games, in which 
the correlators are expressed in terms of basic building blocks (generators) and
consistency conditions (relations) among them. In the physics literature, a traditional 
way of formulating the building blocks is in terms of \emph{operator product expansions}
(OPEs). For the case of bulk and boundary fields this has been done in \cite{lewe3,prss3}.
The formulation of \cite{lewe3} has to be adapted in order to account also for defects and
defect fields \cite{fuRs10,fjfs}, and to be refined \cite{kolR} in order to implement
a concise notion of world sheet, including in particular the proper distinction between 
incoming and outgoing field insertions.
  
For our present purposes, for simplicity we stick with the elementary formulation
of \cite{lewe3}. This involves three building blocks: the bulk OPE, the boundary OPE,
and the bulk-boundary OPE, corresponding to the correlator of three bulk fields on a
sphere, of three boundary fields on a disk, and of one bulk and one boundary field on
a disk, respectively. (In the precise setting of \cite{kolR}, each of these comes in two 
variants related by the exchange of incoming and outgoing fields and there are six
further building blocks with a smaller number of field insertions \Cite{Prop.\,2.6}{kolR}.)
We will use the following pictorial description of the three building blocks
(compare Figure 1 in \cite{lewe3}):
    \def\locph  {0.62}   
    \def\locpx  {1.6}    
    \def\locpX  {1.8}    
    \def\locpy  {2.5}    
    \def\locpY  {2.8}    
    \def\locpz  {2.85}   
  \be
  \text{bulk:} \hspace*{-1.3em}
  \raisebox{-4.1em} 
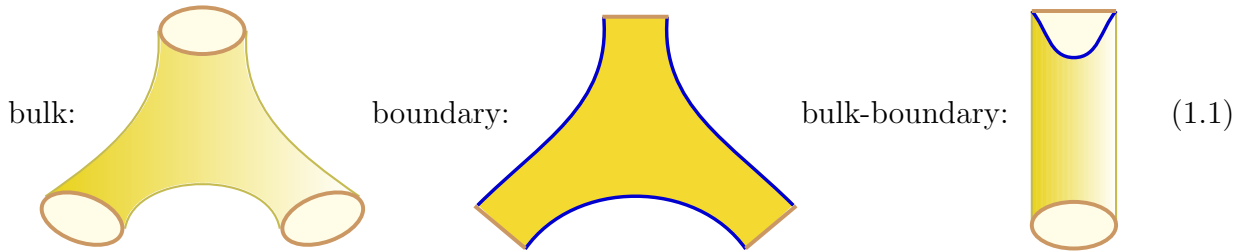
{\begin{tikzpicture} \begin{scope}[rotate=180] 
       \ThreeBulkFieldsOnSphere \end{scope} \end{tikzpicture}}
  \hspace*{0.1em} \text{boundary:} \hspace*{-1.2em}
  \raisebox{-4.1em} {\begin{tikzpicture} \begin{scope}[rotate=180]
       \ThreeBdyFieldsOnDisk \end{scope} \end{tikzpicture}}
  \hspace*{0.1em} \text{bulk-boundary:~~}
  \raisebox{-4.1em} {\begin{tikzpicture} \begin{scope}[rotate=180]
       \CardyPipe \end{scope} \end{tikzpicture}}
  \label{eq:3opes}
  \ee
Here the circles and straight intervals which are part of the boundary of the world sheet
(also called \emph{gluing boundaries}) stand for the insertion of bulk and boundary fields,
respectively, while the remaining segments of the boundary of the disk (which are drawn in
a different color) are \emph{physical boundaries} on which a boundary condition has to be 
specified. Thus denoting the space of bulk fields by $\FF$, using labels $m$, $n$ etc.\
for the possible boundary conditions, and denoting the space of boundary fields
that change the boundary condition from $m$ to $n$ by $\BB^{n,m}$, 
a more detailed graphical description of the boundary and bulk-boundary operator products is
  \be
  \text{boundary:} \hspace*{-0.8em}
  \raisebox{-5.3em}
  
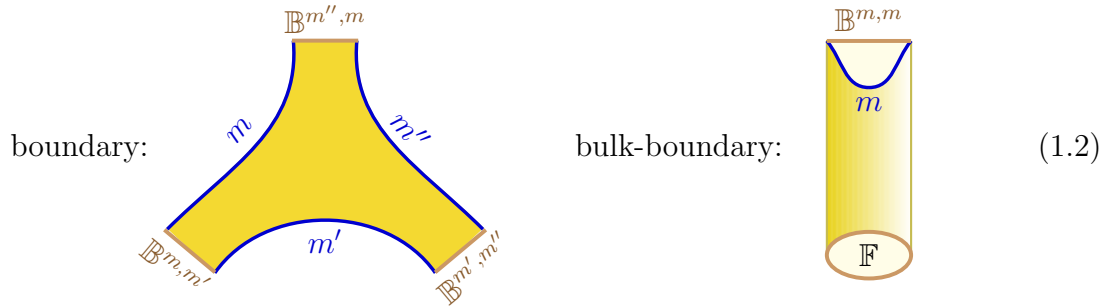
{\begin{tikzpicture} \begin{scope}[rotate=180]
  \ThreeBdyFieldsOnDiskTopLabeled {m}{m''}{m'}{\BB} \end{scope} \end{tikzpicture}}
  \hspace*{1.4em} \text{bulk-boundary:~~~~}
  \raisebox{-4.0em} {\begin{tikzpicture} \begin{scope}[rotate=180]
       \CardyPipeTopLabeled m{\BB}{\FF} \end{scope} \end{tikzpicture}}
  \label{eq:2opes}
  \ee

Based on results of \cite{fuSc25}, our proposal for the field content -- which includes
also defect fields -- leads very naturally to a proposal for the OPEs \eqref{eq:3opes}
as well as for the  two types of OPEs of defect fields. 
This is the third result of the present paper. We furthermore show that the
OPEs we propose satisfy all genus-0 constraints which the building blocks must satisfy,
namely crossing symmetries of the following correlators: four bulk fields on a sphere,
four boundary fields on a disk, one bulk and two boundary fields on a disk, and
one boundary and two bulk fields on a disk. In the pictorial description \eqref{eq:3opes}
these constraints look as follows:
 \\[3pt]
(1) Crossing symmetry for the correlator of four bulk fields on the sphere:
    \def\locpx  {1.7}    
    \def\locpy  {1.4}    
    \def\locpX  {0.18*\locpx}   
    \def\locpY  {0.81*\locpy}   
    \def\locXshift {0.04*\locpx}
  \be
  \raisebox{-4.4em} 
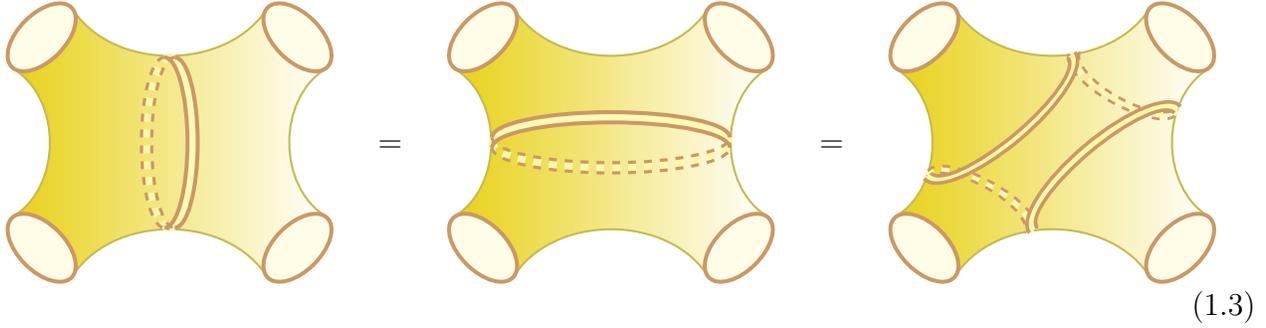
{\begin{tikzpicture}
  \FourBulkFieldsOnSphere
  \draw[\colorHoleLite,line width=2*\widthFieldline]
       (0,-\locpY) arc (-90:90:\locpX cm and \locpY cm) ; 
  \draw[\colorHoleDark,line width=2*\widthFieldline,dashed]
       (0,\locpY) arc (90:270:\locpX cm and \locpY cm) ; 
  \draw[\colorBulkfield,line width=\widthFieldline]
       (-\locXshift,-\locpY) arc (-90:90:\locpX cm and \locpY cm)
       (\locXshift,-\locpY) arc (-90:90:\locpX cm and \locpY cm) ; 
  \draw[\colorBulkfield,line width=0.8*\widthFieldline,dashed]
       (-\locXshift,\locpY) arc (90:270:\locpX cm and \locpY cm)
       (\locXshift,\locpY) arc (90:270:\locpX cm and \locpY cm) ; 
  \end{tikzpicture}
  }
  ~~=~~~
    \def\locpX  {0.93*\locpx}   
    \def\locpY  {0.24*\locpy}   
    \def\locYshift {0.04*\locpx}
  \raisebox{-4.4em} {\begin{tikzpicture}
  \FourBulkFieldsOnSphere
  \draw[\colorHoleLite,line width=2*\widthFieldline]
       (\locpX,0) arc (0:180:\locpX cm and \locpY cm) ;
  \draw[\colorHoleDark,line width=2*\widthFieldline,dashed]
       (-\locpX,0) arc (180:360:\locpX cm and \locpY cm) ;
  \draw[\colorBulkfield,line width=\widthFieldline]
       (\locpX,-\locYshift) arc (0:180:\locpX cm and \locpY cm)
       (\locpX,\locYshift) arc (0:180:\locpX cm and \locpY cm) ;
  \draw[\colorBulkfield,line width=0.8*\widthFieldline,dashed]
       (-\locpX,-\locYshift) arc (180:360:\locpX cm and \locpY cm)
       (-\locpX,\locYshift) arc (180:360:\locpX cm and \locpY cm) ;
  \end{tikzpicture}
  }
  ~~=~~~
    \def\locpX  {1.21*\locpx}   %
    \def\locpxl {0.48*\locpx}   
    \def\locpXl {0.77*\locpx}   
    \def\locpxm {0.46*\locpx}   
    \def\locpXm {0.73*\locpx}   
    \def\locpxs {0.44*\locpx}   
    \def\locpXs {0.70*\locpx}   
    \def\locpY  {0.18*\locpy}   %
    \def\locpyl {0.20*\locpy}   
    \def\locpYl {0.30*\locpy}   
    \def\locpym {0.20*\locpy}   
    \def\locpYm {0.265*\locpy}  
    \def\locpys {0.20*\locpy}   
    \def\locpYs {0.23*\locpy}   
    \def\locpcx {0.78*\locpx}   
    \def\locpcy {0.75*\locpy}   
    \def\locpdx {0.27*\locpx}   
    \def\locpdy {0.23*\locpy}   
    \def\locps  {1.10}          
  \raisebox{-4.4em} {\begin{tikzpicture}
  \FourBulkFieldsOnSphere
  \begin{scope}[rotate=-29]
  \draw[\colorHoleLite,line width=1.5*\widthFieldline,dashed]
       (0.272*\locpx-\locpxm,0.845*\locpy) arc (-180:0:\locpxm cm and \locpym cm)
       (-0.272*\locpx-\locpxm,-0.845*\locpy) arc (180:0:\locpxm cm and \locpym cm) ;
  \draw[\colorBulkfield,line width=0.7*\widthFieldline,dashed]
       (0.285*\locpx-\locpxs,0.88*\locpy) arc (-180:0:\locpxs cm and \locpys cm)
       (0.26*\locpx-\locpxl,0.81*\locpy) arc (-180:0:\locpxl cm and \locpyl cm)
       (-0.285*\locpx-\locpxs,-0.88*\locpy) arc (180:0:\locpxs cm and \locpys cm)
       (-0.26*\locpx-\locpxl,-0.81*\locpy) arc (180:0:\locpxl cm and \locpyl cm) ;
  \end{scope}
  \begin{scope}[rotate=40]
  \draw[\colorHoleDark,line width=1.5*\widthFieldline]
       (-0.18*\locpx-\locpXm,0.525*\locpy) arc (-180:0:\locpXm cm and \locpYm cm)
       (0.18*\locpx-\locpXm,-0.525*\locpy) arc (180:0:\locpXm cm and \locpYm cm) ;
  \draw[\colorBulkfield,line width=0.9*\widthFieldline]
       (-0.18*\locpx-\locpXs,0.52*\locpy) arc (-180:0:\locpXs cm and \locpYs cm)
       (-0.18*\locpx-\locpXl,0.52*\locpy) arc (-180:0:\locpXl cm and \locpYl cm)
       (0.18*\locpx-\locpXs,-0.53*\locpy) arc (180:0:\locpXs cm and \locpYs cm)
       (0.18*\locpx-\locpXl,-0.52*\locpy) arc (180:0:\locpXl cm and \locpYl cm) ;
  \end{scope}
  \end{tikzpicture}
  }
  \label{eq:relation9a}
  \ee
(2) Crossing symmetry for the correlator of four boundary fields on the disk:
    \def\locpX  {1.54}   
    \def\locpY  {2.14}   
    \def\locpXl {2.36}   
    \def\locpXm {2.01}   
    \def\locpYl {1.66}   
    \def\locpYm {0.75}   
    \def\widthBdyfield {0.8*\StdwidthBdyfield}
    \def\locXshift {0.05*\locpY}
    \def\locYshift {0.05*\locpY}
  \be
  \raisebox{-5.3em} 
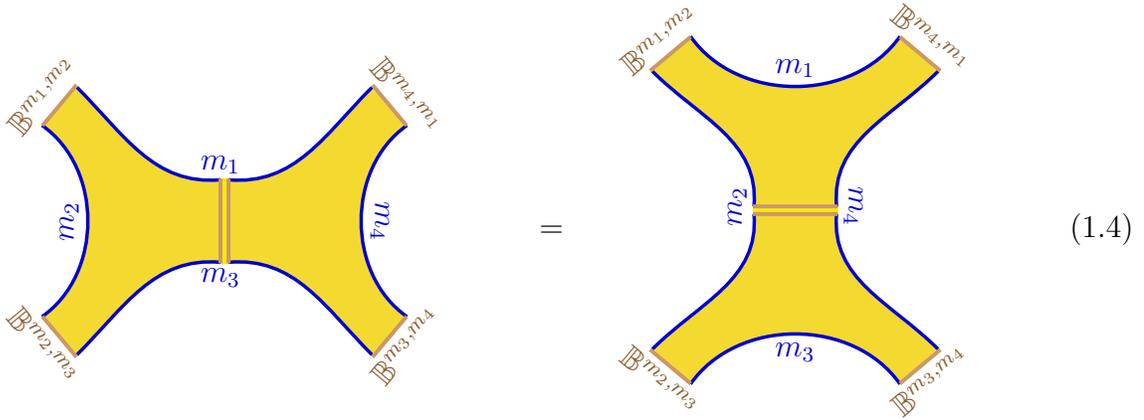
{\begin{tikzpicture}
  \fill[\colorFaceLite]
       (-\locXshift,-0.8*\widthBdyfield) rectangle (2*\locXshift,0.8*\widthBdyfield) ;
  \begin{scope}[rotate=90] \TwoplusonebroadBdyFieldsOnDisk \end{scope}
  \begin{scope}[shift={(\locXshift,0)}]
       \begin{scope}[rotate=-90] \TwoplusonebroadBdyFieldsOnDisk
  \end{scope} \end{scope}
  \LabelBdy {0} {\locpYm}    {0} {m_1} 
  \LabelBdy {-\locpXm} {0}  {90} {m_2} 
  \LabelBdy {0} {-\locpYm}   {0} {m_3} 
  \LabelBdyField {-\locpXl} {\locpYl}   {50} {m_1} {m_2} 
  \LabelBdyField {-\locpXl} {-\locpYl} {310} {m_2} {m_3} 
  \begin{scope}[shift={(\locXshift,0)}]
  \LabelBdy {\locpXm} {0}  {270} {m_4} 
  \LabelBdyField {\locpXl}  {-\locpYl}  {50} {m_3} {m_4} 
  \LabelBdyField {\locpXl}  {\locpYl}  {310} {m_4} {m_1} 
  \end{scope}
  \end{tikzpicture}
  }
  \qquad = \quad
    \def\locpX  {1.66}   
    \def\locpY  {2.03}   
    \def\locpXl {1.80}   
    \def\locpXm {0.77}   
    \def\locpYl {2.22}   
    \def\locpYm {1.84}   
  \raisebox{-6.3em} {\begin{tikzpicture} 
  \fill[\colorFaceLite]
       (-0.8*\widthBdyfield,-\locYshift) rectangle (0.8*\widthBdyfield,2*\locYshift) ;
  \TwoplusonebroadBdyFieldsOnDisk
  \begin{scope}[shift={(0,-\locYshift)}]
       \begin{scope}[rotate=180] \TwoplusonebroadBdyFieldsOnDisk
  \end{scope} \end{scope}
  \LabelBdy {0} {\locpYm}    {0} {m_1} 
  \LabelBdy {-\locpXm} {0}  {90} {m_2} 
  \LabelBdy {\locpXm} {0}  {270} {m_4} 
  \LabelBdyField {-\locpXl} {\locpYl}   {40} {m_1} {m_2} 
  \LabelBdyField {\locpXl}  {\locpYl}  {320} {m_4} {m_1} 
  \begin{scope}[shift={(0,-\locYshift)}]
  \LabelBdy {0} {-\locpYm}   {0} {m_3} 
  \LabelBdyField {-\locpXl} {-\locpYl} {320} {m_2} {m_3} 
  \LabelBdyField {\locpXl}  {-\locpYl}  {40} {m_3} {m_4} 
  \end{scope}
  \end{tikzpicture}
  }
  \label{eq:relation9c}
  \ee
~\\ 
(3) Compatibility of moving a bulk field to different segments of
the boundary of a disk with two boundary field insertions:
    \def\locph  {0.82}   
    \def\locpX  {2.36}   
    \def\locpXB {0.26}   
    \def\locpY  {1.77}   
    \def\locpYm {0.22}   
    \def\locpz  {1.65}   
    \def\widthEllipseX {\StdwidthEllipseX}
    \def\widthEllipseY {\StdwidthEllipseY}
    \def\widthBdyfield {2.12*\widthEllipseX}
    \def\locYshift {0.05*\locpY}
  \be
  \raisebox{-6.8em} 
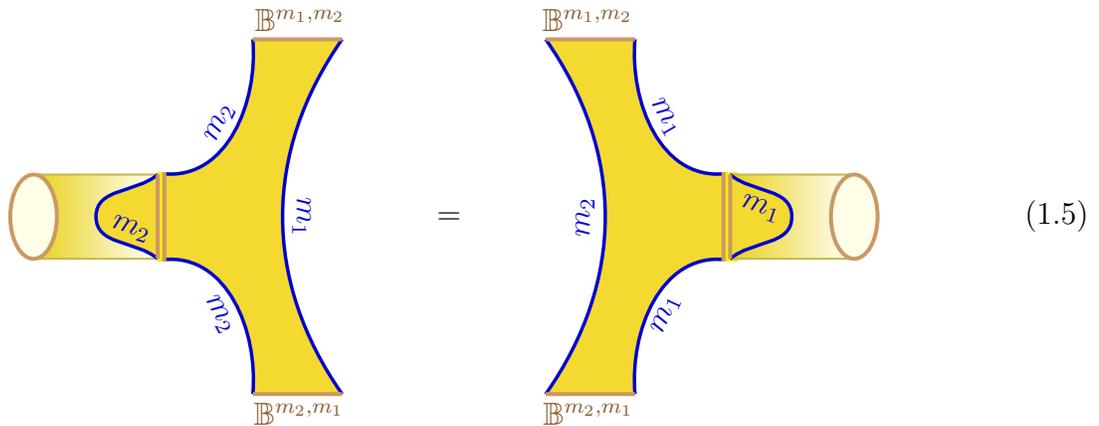
{\begin{tikzpicture}
  \fill[\colorFaceLite]
       (-\locYshift,-0.5*\widthBdyfield) rectangle (2*\locYshift,0.5*\widthBdyfield) ;
  \begin{scope}[rotate=-90]\TwoplusoneBdyFieldsOnDisk \end{scope}
  \begin{scope}[shift={(-\locYshift,0)}] \begin{scope}[rotate=90]
      \CardyPipeFilledBdy
  \end{scope} \end{scope}
  \LabelBdy {1.03*\locpY} {0}  {-90} {m_1} 
  \LabelBdy {-0.25*\locpY} {-0.07*\locpX} {-19} {m_2} 
  \LabelBdy {0.41*\locpY} {-0.55*\locpX} {-61} {m_2} 
  \LabelBdy {0.42*\locpY} {0.53*\locpX} {64} {m_2} 
  \LabelBdyField {\locpY+0.02*\widthBdyfield} {-\locpX-\locpXB} {0} {m_2} {m_1} 
  \LabelBdyField {\locpY+0.03*\widthBdyfield} {\locpX+\locpXB}  {0} {m_1} {m_2}
  \end{tikzpicture}
  }
  \qquad = \qquad
  \raisebox{-6.8em} {\begin{tikzpicture}
  \fill[\colorFaceLite]
       (-\locYshift,-0.5*\widthBdyfield) rectangle (2*\locYshift,0.5*\widthBdyfield) ;
  \begin{scope}[rotate=90] \TwoplusoneBdyFieldsOnDisk \end{scope}
  \begin{scope}[shift={(\locYshift,0)}] \begin{scope}[rotate=-90]
       \CardyPipeFilledBdy
  \end{scope} \end{scope}
  \LabelBdy {-\locpz-0.85*\locpYm} {0} {90} {m_2} 
  \LabelBdy {0.28*\locpY} {0.11} {-19} {m_1} 
  \LabelBdy {-0.43*\locpY} {-0.55*\locpX} {65} {m_1} 
  \LabelBdy {-0.44*\locpY} {0.56*\locpX} {-60} {m_1} 
  \LabelBdyField {-\locpY-0.02*\widthBdyfield} {-\locpX-\locpXB} {0} {m_2} {m_1} 
  \LabelBdyField {-\locpY-0.03*\widthBdyfield} {\locpX+\locpXB}  {0} {m_1} {m_2}
  \end{tikzpicture}
  }
  \label{eq:relation9d}
  \ee
(4) Compatibility of the boundary OPE and bulk OPE:
    \def\locph  {0.62}   
    \def\locpx  {1.6}    
    \def\locpX  {1.44}   
    \def\locpy  {2.0}    
    \def\locpY  {2.24}   
    \def\locpz  {1.65}   
    \def\widthEllipseX {0.8*\StdwidthEllipseX}
    \def\widthEllipseY {0.8*\StdwidthEllipseY}
    \def\widthBdyfield {2*\widthEllipseX}
  \be
  \raisebox{-5.1em} {\begin{tikzpicture}
  \begin{scope}[shift={(.5*\locpX,-1.25*\locpY)}]
  \fill[rotate=40,\colorFaceLite] (0.5*\widthBdyfield,-0.3) rectangle ++(\widthBdyfield,0.5) ;
  \end{scope}
  \begin{scope}[shift={(-1.05*\locpX,-1.05*\locpY)}]
  \fill[rotate=-40,\colorFaceLite] (-0.5*\widthBdyfield,-0.3) rectangle ++(\widthBdyfield,0.5) ;
  \end{scope}
      \begin{scope}[rotate=180] \ThreeBdyFieldsOnDiskTopLabeled mmm {\BB} \end{scope}
  \begin{scope}[shift={(1.05*\locpX,-1.05*\locpY)}] \begin{scope}[rotate=220]
      \CardyPipeFilledBdy
  \end{scope} \end{scope}
  \begin{scope}[shift={(-1.05*\locpX,-1.05*\locpY)}] \begin{scope}[rotate=-220]
      \CardyPipeFilledBdy
  \end{scope} \end{scope}
  \end{tikzpicture}
  }
  \quad = \quad
  \raisebox{-5.1em} {\begin{tikzpicture}
  \begin{scope}[rotate=180] \ThreeBulkFieldsOnSphere \end{scope}
  \BulkfieldLite 0 {-0.01*\locpX} 0
  \Bulkfield 0 {0.09*\locpX} 0
  \begin{scope}[shift={(0,0.91*\locpy)}] \begin{scope}[rotate=-180]
	  \CardyPipeFilledBulkLabeled m {\BB}
  \end{scope} \end{scope}
  \end{tikzpicture}
  }
  \label{eq:relation9e}
  \ee
Besides these four genus-0 relations, there are two further relations at genus 1.
In the setting of \cite{kolR}, there is a total if 32 relations, see Section 2.4 
and Remark 3.4 of \cite{kolR}.

\medskip

The two genus-1 constraints, given by the items (b) and (f) in the list in Figure 9 
of \cite{lewe3}, ensure the compatibility of correlation functions on higher-genus surfaces:
Relation (b) which requires the modular invariance of a one-bulk field correlator on a torus
amounts to the statement that the object of bulk fields is a \emph{modular} Frobenius algebra
in the sense of \Cite{Sect.\,3.1}{koRu2} and \Cite{Def.\,4.9}{fuSc22}, while relation (f) 
-- the so-called Cardy condition for a two-boundary field correlator on an annulus -- describes 
the compatibility of handle-generating sewings that involve bulk and boundary fields, 
respectively. As is generally true for higher-genus issues, these relations are 
considerably more subtle than the genus-0 constraints.  We do expect that they can 
be derived from our proposal as well, but have to leave their discussion to future work.

\medskip

This paper is organized as follows. We start in Section \ref{sec:ass-chiral} with 
presenting the requirements we impose on the underlying chiral conformal field theory
(Assumption 1). Given this assumption, we can work with finite categories or, more 
specifically, with finite tensor categories and finite module categories over them. In 
Section \ref{sec:bc+bf} we then explain that an indecomposable pivotal module category
\M\ provides a consistent boundary theory (Assumption 2), and how boundary fields and
their OPE are expressed in terms of \M\ (Assumption 3). The remaining steps may be 
summarized as the statement that we then construct the bulk theory, including defect fields,
from the boundary theory. For doing so various results of \cite{fuSc25} are crucial. We first
expound, in Section \ref{sec:dc+df}, that defect conditions should be interpreted as 
right exact module functors (Assumption 4). Section \ref{sec:theproblem} is devoted 
to a precise statement of the problem of reconstructing the bulk from the boundary.
After an overview of pertinent mathematical structures and results in Section \ref{sec:iNat}, 
we are then ready to state, in Section \ref{sec:theproposal}, our proposal. In the
remainder of Section \ref{sec:3} we perform several consistency checks which
corroborate the validity of our proposal. In the final Section \ref{sec:out}
we conclude with an outlook on open issues and future directions of research.


\section{Field content and operator products in full CFT}

In this section we carefully formulate all requirements that will be assumed 
in our proposal. As already pointed out, these assumptions are satisfied for a large 
class of models, including in particular all rational CFTs as well as many
logarithmic CFTs. In passing, we also provide various pertinent background information.

\subsection{Assumptions on the chiral data} \label{sec:ass-chiral}

We first state our assumptions about the chiral data of the class
of conformal field theories for which we formulate our proposal.

\begin{Assumption} \label{ass1}
The chiral data of a chiral conformal field theory are given by a not necessarily semisimple
modular tensor category \C.
\end{Assumption}

The notion of a modular tensor category arises as an abstraction of the structure and 
properties of the representations of the chiral symmetry algebra of the CFT (concretely, 
a vertex operator algebra with appropriate properties, including in particular 
$C_2$-cofiniteness). We do not fully unravel its definition,
referring to Section 2.1 of \cite{fuSc25} for further pertinent mathematical 
details. Instead we just highlight those aspects that are most relevant to our proposal. 
First of all, a modular tensor category is \emph{linear} over some ground field \ko, 
in particular the morphism sets are \ko-vector spaces. In the CFT context, \ko\ is given
by the complex numbers $\complex$. It is also worth mentioning that in the semisimple case 
the 6j-symbols (or fusing matrices, in CFT terminology \cite{mose3}) are already encoded 
in the monoidal structure, namely in the associativity constraint for the tensor product.
Next we recall that a modular tensor category $\calc$ is in particular a finite 
ribbon category. The ribbon structure comprises a \emph{braiding}, i.e.\ a family of
isomorphisms $\cb_{c,c'}\colon c\oti c'\To c' \oti c$ that is natural in both arguments
$c,c'\iN\calc$ and obeys the two standard hexagon identities. The 
structure of a braiding accounts for the fundamental fact that chiral conformal 
field theories realize braid group statistics. Examples of braided tensor categories
are given by the Drinfeld center $\calz(\cala)$ of any monoidal category $\cala$. An
object of $\calz(\cala)$ is a pair $(a,\cB)$, consisting of an object $a \iN \cala$ 
and a half-braiding. (A half-braiding  for an object $a_0\iN\cala$ is a natural family
$\cB \eq (\cB_a)_{a\in \cala}$ of morphisms $\cB_a \colon a \oti a_0\To a_0 \oti a$ 
obeying a single hexagon identity.) Besides the braiding there are two other ingredients of a
ribbon structure: first, a \emph{ribbon twist}, i.e.\ a natural family $\theta_c\colon c\To c$
of endomorphisms, for any $c\iN\calc$, which keeps track of the exponentials of the conformal 
weights; and second, a \emph{rigidity structure}, i.e.\ for any $c\iN\calc$ an assignment of
a left dual object $\Vee c$ and a right dual object $c^\vee$ together with corresponding
evaluation and coevaluation morphisms.

The braiding in a modular tensor category is non-degenerate. In a finitely semisimple 
modular tensor category, this property amounts to invertibility 
of the modular S-matrix, which in the context of three-dimensional topological field theory
describes the invariants associated to the Hopf link in the three-manifold
${\mathbb S}^3$, colored by simple objects of \C. In the present paper we do \emph{not}
impose semisimplicity and accordingly use a different non-degeneracy condition on the 
braiding.\,%
 \footnote{~Several other non-degeneracy conditions on a braiding have been enunciated.
 It has been shown \cite{shimi10} that for braided finite tensor categories
 all those non-degeneracy conditions are equivalent.}
To formulate the latter, denote by $\Cb$ the \emph{reverse} of a braided category, i.e.\
the same category, but with inverse braiding $\cb^\text{rev}_{c,c'} \,{:=}\, \cb_{c',c}^{-1}
\colon c\oti c'\To c' \oti c$. There is a canonical braided functor
  \be
  \GC: \quad \Cb \boti \C \xrightarrow~ \ZC
  \label{eq:CbC->ZC}
  \ee
from the \emph{enveloping category} of \C, i.e.\ the Deligne product of \Cb\ with \C,
to the Drinfeld center of \C. As a functor, \GC\ maps an object 
$\Ol u \boti v \iN \Cb \boti \C$ to the tensor product $\Ol u \oti v \iN \C$ endowed 
with the half-braiding $\cB_{\Ol u\otimes v}$ whose components are $\cB_{\Ol u\otimes v;c}
\,{:=}\, ( \id_{\Ol u} \oti \cb_{c,v}^{} ) \cir (\cb_{\Ol u,c}^{-1} \oti \id_v )$ for $c \iN \C$,
with $\cb$ the braiding in \C. Note that this implies in particular that the composition of \GC\
with the \emph{forgetful functor} $\UC\colon \ZC \,{\to}\, \C$ -- the functor that ignores the
half-braiding, i.e.\ acts as $\UC(c,\cB) \eq c$, is nothing but the tensor product in \C,
  \be
  \begin{array}{rl}
  \UC \circ \GC : \quad \CbC \!\!\! & \rarr{} \, \C \,,
  \Nxl2
  c \boti c' \!\!\! & \xmapsto{{~~}} \, c \oti c' .
  \eear
  \label{eq:UC.GC}
  \ee

\begin{Defi}
A \emph{modular tensor category}
is a finite ribbon category such that the braided mo\-no\-idal functor $\GC$ is an equivalence.
\end{Defi}

Bulk fields (and, more generally, defect fields) are obtained by combining left and right 
movers or, put differently, carry two commuting representations of the chiral algebra. They
are thus naturally objects in the enveloping category \CbC. In the sequel it will be
important that by using the equivalence \eqref{eq:CbC->ZC} we can alternatively study bulk 
fields as objects in the Drin\-feld center \ZC\ of the monoidal category that encodes
the chiral data.

Being a finite abelian category, a modular category \C\ has various finiteness properties: 
the number of isomorphism classes of simple objects is finite, all morphism spaces are
finite-di\-men\-si\-onal, and all objects have finite length (compare \Cite{Ch.\,1.5}{EGno}).
These requirements can be summarized as the statement that, as a \ko-linear abelian 
category, \C\ is equivalent to the category of finite-dimensional modules over a
finite-dimensional \ko-algebra. 
Being a finite \emph{tensor} category, \C\ is, quite importantly, in addition 
rigid, so that in particular the tensor product functor is exact in both variables.

In a modular tensor category, the double dual is trivialized. This is formalized
by the notion of a pivotal structure.

\begin{Defi}
A \emph{pivotal structure} on a right rigid monoidal category \C\ is a monoidal	
natural iso\-morphism $\pi\colon \Id_\C \,{\xrightarrow~}\, -^{\!\vee\vee}$ from the
identity functor to the double-dual functor.
\end{Defi}

A modular tensor category comes with a canonical pivotal structure. We tacitly regard it as
a pivotal category endowed with this pivotal structure and use it to identify an object 
with its double dual or, equivalently, the left and right duals of an object.

\medskip

Admittedly, Assumption \ref{ass1} excludes interesting types of chiral conformal 
field theories, like e.g.\ the uncompactified free boson, Liouville
theory, critical percolation, WZW models at fractional level, and ghost systems, to name
just a few popular ones. Let us stress, however, that we do not impose semisimplicity. As
a consequence, there is still a very large class of examples to which our arguments
apply. It includes on the one hand all semisimple modular categories (corresponding to
rational chiral CFTs), and precise criteria are known for a vertex algebra \cite{huan21}
or a net of observables \cite{kalM} to have a semisimple modular category as its
category of representations. On the other hand, screening charge constructions yield
many examples that are not semisimple \cite{galO}.
Also note that a central tool of our construction is given by the internal Homs.
These still exist when \C\ is no longer a finite tensor category (such as
for categories that are not rigid but still have a Grothendieck-Verdier duality);
this suggests that some of our structural insights can survive in more general situations
than those we consider here.


\subsection{Boundary conditions and boundary fields} \label{sec:bc+bf}

Recall that our final goal is to construct a full local conformal field theory 
from a given chiral theory. This requires in particular a description of bulk fields 
and, more generally, defect fields, in which left and right movers are combined. Already 
the basic example of unitary Virasoro minimal models shows that to this end additional 
data need to be specified. It has been clear since long that -- as witnessed by the 
existence of unphysical modular invariants which we mentioned in the Introduction -- the
additional datum in question cannot be a modular invariant. And indeed it is well understood 
\cite{fjfrs2} that in case the modular tensor category $\calc$ is semisimple, 
the required datum is an equivalence class of special symmetric Frobenius algebras 
internal to the category $\calc$. Any such algebra $A$ plays the role of an 
\emph{algebra of boundary fields}, while its modules are the possible boundary conditions. 
The category of $A$-modules has the structure of a semisimple module category \M\ over \C,
i.e.\ there is an exact \emph{action functor} $\C \Times \M \To \M$, together with a mixed 
associator and a mixed unitor that obey mixed pentagon and triangle relations. 
Moreover, the algebra $A$ must be simple as a bimodule over itself, implying that
\M\ is an \emph{indecomposable} module category. (For pertinent information
on module categories see e.g.\ \Cite{Ch.\,7}{EGno} or \Cite{Sect.\,2.3}{shimi17}.)

We will need the following information about the relation between algebras and module 
categories. For an algebra $A\iN\calc$, the category $\Amod$ of \emph{right} $A$-modules
becomes a \emph{left} module category over \C\ by endowing the object $c\oti \dot m$,
for a right module $(\dot m,\rho)$ -- with $\dot m\iN\calc$ and right action 
$\rho\colon \dot m\oti A \To \dot m$ -- with the right action $\id_c\oti \rho$. To 
appreciate the converse relationship we need in addition
the notion of an internal Hom, which will play a central role in our arguments.

\begin{Defi}
Let $\calc$ be a monoidal category and $\calm$ be a left $\calc$-module category.
For any pair $m,m'\iN\calm$ of objects in the module category, the \emph{internal Hom}
$\iHom_\calm(m,m')$ is an object of $\calc$ together with a natural family
  \be
  \HomC(c,\iHomM(m,m')) \xrightarrow{~\cong~} \HomM(c\act m,m')
  \label{eq:def:iHom}
  \ee
of isomorphisms, for $c\iN\calc$. 
\end{Defi}

In full generality, internal Homs need not exist. In our framework their existence is,
however, guaranteed because $\iHom(-,-)$ is by definition right adjoint to the action 
functor (which is still required to be exact in its first variable), and any right exact 
functor on a finite category has a right adjoint. (Finiteness of \C\ is, however,
not a necessary condition; internal Homs exist in other classes of categories as well.) 
In case \C\ is semisimple, the internal Hom can be expressed as
  \be
  \iHom_\M(n,m) \cong m \otimes_A\! \Vee n \,,
  \label{eq:iHom-otA}
  \ee
i.e.\ as a tensor product over $A$, when \M\ is realized as the category of right modules
over a special symmetric Frobenius algebra $A$ in \C.

If the module category \M\ is clear from the context, we suppress it in the notation and
just write $\iHom(m,m')$. The following fact (see e.g.\ Chapter 7.9 of \cite{EGno}) plays a
crucial role for our proposal:

\begin{Prop}
Let $\calc$ be a monoidal category and $\calm$ be a left $\calc$-module category
for which internal Homs exist. 
 \\[2pt]
{\rm (i)}\,
For any pair $m,m'\iN\calm$ there is a canonical \emph{evaluation morphism}
  \be
  \iev_{m,m'}:\quad \iHom(m,m')\act m \xrightarrow{~~} m' 
  \ee
in \M.
 \\[2pt]
{\rm (ii)}\, 
For any triple $m,m',m''\iN\calm$ there is an associative \emph{multiplication morphism}
  \be
  \imu_{m,m',m''}^{}:\quad \iHom(m',m'')\otimes \iHom(m,m') \xrightarrow{\phantom{~~}} \iHom(m,m'')
  \label{eq:mult-iHom}
  \ee
in \C.
In particular, for any $m\iN\M$, the object $\iHom(m,m)$ is an associative 
(and actually also unital) algebra in $\calc$.
\end{Prop}

\begin{Rem}
(i) The evaluation morphism $\iev_{m,m'}$ is the image of the identity morphism
$\id_{\iHom(m',m)}$ under the adjunction \eqref{eq:def:iHom}. As a consequence we have
  \be
  \alpha \,=\, \iev_{m,m'} \circ (\widetilde\alpha\act\id_{m}) : \quad
  c\act m \xrightarrow{\,\widetilde\alpha\act\id_{m}~} \iHom(m,m')\act m
  \xrightarrow{\,\iev_{m,m'}\,} m'
  \label{eq:tworems2}
  \ee	
as an equality of morphisms in \M, where $\widetilde \alpha\iN \Hom_\C(c,\iHom(m,m'))$ 
denotes the image of $\alpha\iN\Hom_\M(c\act m,m')$ under the adjunction
(compare also the proof of Lemma 4.2.2 of \cite{schaum2}).
 \\[2pt]
(ii) The multiplication morphisms are the image of the composite morphisms
  \be
  \iHom(m',m'')\otimes\iHom(m,m')\act m
  \xrightarrow{\,\id_{\iHom(m',m'')}\otimes\,\iev_{m,m'}\,} \iHom(m',m'')\act m' 
  \xrightarrow{\,\iev_{m',m''}\,} m''
  \label{eq:tworems1}.
  \ee		
under the adjunction \eqref{eq:def:iHom}.	
 \\[2pt]
(iii) 
The internal Hom is a bimodule functor \Cite{Lemma\,2.7}{shimi20}: we have
  \be
  \iHom(c\Act m,c'\Act m') = c' \oti \iHom(m,m') \oti c^\vee
  \label{eq:iHomBimod}	
  \ee
for all $c,c' \iN \C$ and all $m,m' \iN \M$.	
\end{Rem} 

Under conditions that are satisfied for the module categories of our interest (and spelled
out e.g.\ in Theorem 7.10.1 of \cite{EGno}), the category of right 
$\iHom(m,m)$-mo\-du\-les in $\calc$ is equivalent to \M\ as a left \C-module category.
For \C\ a finitely semisimple category the adjunction \eqref{eq:def:iHom} implies
immediately that the internal Hom has the direct sum decomposition
  \be
  \iHom(m,n) \cong\, \bigoplus_{i\in I_\calc} \Hom_\C(U_i,\iHom(m,n)) \,{\otimes_\complex}\,
  U_i \cong\, \bigoplus_{i\in I_\calc}\Hom_\M(U_i\act m,n) \,{\otimes_\complex}\, U_i \,,
  \label{eq:innerhomssi}
  \ee
as an object in \C, where the sum is over a set $I_\C$ of representatives for the isomorphism
classes of simple objects of \C. Combining the behavior of the Hom functor with respect 
to coends (see e.g.\ \Cite{Prop.\,2.7}{fScS2}) with the internal Hom adjunction,
the decomposition \eqref{eq:innerhomssi} generalizes to non-semisimple \C\ as follows:
  \be
  \iHom(m,n) \,\cong \int^{c\in\C}\!\! \Hom_\C(c , \iHom(m,n)) \,{\otimes_\complex}\, c
  \,\cong \int^{c\in\C}\!\! \Hom_\M(c\act m,n) \,{\otimes_\complex}\, c \,.
  \label{eq:iHom-intHom}
  \ee

\medskip

We now discuss the relation between internal Homs and boundary fields. That such a relation 
exists should not come as a
surprise. Let us first have a look at this issue for the case of a semisimple modular 
tensor category. As shown in \cite{fuRs4}, in this case boundary fields which change a 
boundary condition $m\iN \M \eq \Amod$ to $n\iN\Amod$ and whose chiral degree of freedom 
is described by an object $c \iN \C$ come with a multiplicity space $\Hom_\M(c\act m,n)$.
(These spaces satisfy various consistency conditions, given in Theorem 5.20 of \cite{fuRs4}.)
Boundary fields can therefore be labeled as $\Psi^{n,m;\alpha}_{c}$ with $\alpha\iN 
\Hom_\M(c\act m,n)$ (see Table 1 of \cite{fuRs4}). This can be described graphically as
     \def\locpx  {1.5}
     \def\locpy  {2.8}
  \be
  \Psi^{n,m;\alpha}_{c} ~~\hat= \quad
  \raisebox{-5.1em} {\begin{tikzpicture}
  \coordinate (cooal) at (0.9*\locpx,0.5*\locpy) ;
  \coordinate (coom1) at (\locpx,0) ;
  \coordinate (coom2) at (\locpx,\locpy) ;
  \coordinate (cooc)  at (0,0) ;
  \scopeArrow{0.59}{\arrowDefect}
  \draw[\colorC,line width=\widthObj,postaction={decorate}]
       (cooc) node[below]{$c$} [out=90,in=245] to (cooal) ; 
  \end{scope}
  \scopeArrow{0.19}{\arrowDefect} \scopeArrow{0.92}{\arrowDefect}
  \draw[\colorM,line width=\widthObj,postaction={decorate}]
       (coom1) node[below]{$m$} [out=100,in=270] to (cooal)
       [out=90,in=250] to (coom2) node[above]{$n$} ;
  \end{scope} \end{scope}
  \node[line width=\widthboxline,fill=blue!10,draw=\colorBdyfield] at (cooal)
       {$\,\scriptstyle\alpha\,$} ;
  \end{tikzpicture}
  }
  \label{eq:pic1}
  \ee
In this picture, $m,n \iN \M$ are labels for boundary conditions on boundary segments of the
world sheet (that is, of the two-dimensional manifold on which the full CFT is considered),
while the label $c\iN\C$ embodies the chiral field content of the boundary field.
In more detail, the picture \eqref{eq:pic1} can be interpreted as the standard graphical 
representation of a morphism in a monoidal category \C, with \M\ identified as $\Amod$ and
with the modules $m$ and $n$ in $\Amod$ identified with their underlying objects in \C,
and thus $\alpha \iN \Hom_A(c \oti m,n)$ regarded as a morphism in \C. 
But alternatively we can interpret \eqref{eq:pic1} in terms of a graphical calculus for
the monoidal category \C\ and its module category \M, as developed in \cite{schaum}; then
it describes a morphism $\alpha \iN \Hom_\M(c\act m,n)$. 
It is tempting to think of the picture \eqref{eq:pic1} also more directly as showing the
relevant region of an actual world sheet. This is indeed possible in the semisimple case,
in which the construction of correlation functions of rational CFT in terms of ribbon graphs 
in three-manifolds \cite{fffs3,fuRs4} is available. The lines labeled by $m$ and $n$ then stand
for actual segments of the boundary of the world sheet, while the line with chiral label $c$ is 
located in a part of the three-manifold outside the (embedded) world sheet (compare e.g.\ 
Figure 1 in \cite{fffs3} or the picture (4.15) in \cite{fuRs10}).

Adopting the interpretation of $\alpha$ as a morphism in \M\ and invoking the equality 
\eqref{eq:tworems2}, $\alpha$ can also be expressed as
     \def\locpx  {1.5}
     \def\locpy  {3.7}
  \be
  \raisebox{-5.0em} {\begin{tikzpicture}
  \coordinate (cooal) at (0.9*\locpx,0.5*\locpy) ;
  \coordinate (coom1) at (\locpx,0) ;
  \coordinate (coom2) at (\locpx,\locpy) ;
  \coordinate (cooc)  at (0,0) ;
  \scopeArrow{0.59}{\arrowDefect}
  \draw[\colorC,line width=\widthObj,postaction={decorate}]
       (cooc) node[below]{$c$} [out=90,in=245] to (cooal) ;
  \end{scope}
  \scopeArrow{0.19}{\arrowDefect} \scopeArrow{0.92}{\arrowDefect}
  \draw[\colorM,line width=\widthObj,postaction={decorate}]
       (coom1) node[below]{$m$} [out=100,in=270] to (cooal)
       [out=90,in=250] to (coom2) node[above]{$n$} ;
  \end{scope} \end{scope}
  \node[line width=\widthboxline,fill=blue!10,draw=\colorBdyfield] at (cooal)
       {$\,\scriptstyle\alpha\,$} ;
  \end{tikzpicture}
  }
  ~\quad = \quad
  \raisebox{-5.0em} {\begin{tikzpicture}
  \coordinate (cooal) at (0.16*\locpx,0.35*\locpy) ;
  \coordinate (cooev) at (0.9*\locpx,0.7*\locpy) ;
  \coordinate (coom1) at (\locpx,0) ;
  \coordinate (coom2) at (\locpx,\locpy) ;
  \coordinate (cooc)  at (0,0) ;
  \scopeArrow{0.22}{\arrowDefect}
  \scopeArrow{0.76}{\arrowDefect}
  \draw[\colorC,line width=\widthObj,postaction={decorate}]
       (cooc) node[below]{$c$} [out=90,in=210] to
       node[rotate=41,left,very near end,yshift=7pt,xshift=-1pt] {$\scriptstyle\iHom(m,n)$} (cooev) ; 
  \end{scope} \end{scope}
  \scopeArrow{0.19}{\arrowDefect} \scopeArrow{0.95}{\arrowDefect}
  \draw[\colorM,line width=\widthObj,postaction={decorate}]
       (coom1) node[below]{$m$} [out=100,in=270] to (cooev)
       [out=90,in=250] to (coom2) node[above]{$n$} ;
  \end{scope} \end{scope}
  \node[line width=0.7*\widthboxline,fill=blue!10,draw=blue] at (cooev)
       {$\scriptstyle\iev_{m,n}$} ;
  \node[rotate=-20,rounded corners,line width=\widthboxline,fill=blue!10,draw=\colorBdyfield]
       at (cooal) {$\scriptstyle\widetilde\alpha$} ;
  \end{tikzpicture}
  }
  \label{eq:pic2}
  \ee
This way we have managed to describe boundary fields of all chiral types $c \iN \C$ 
for fixed boundary conditions $m,n$ naturally via a single internal Hom object,
in a way that no longer requires \C\ to be semisimple. And indeed, 
as has been seen in e.g.\ \Cite{Sect.\,3}{garW} and \Cite{Sect.\,4.4}{fgsS}, objects of
boundary fields can be expressed beyond semisimplicity through internal Homs as
  \be
  \BB^{n,m} = \iHom(m,n) \,.
  \label{eq:BB=iHom}
  \ee

It is then natural to expect that the composition of internal Homs -- which is 
automatically associative -- provides the boundary operator products. To see that this is
consistent, we first note that as a consequence of the equality \eqref{eq:tworems1} we have,
with the identification \eqref{eq:BB=iHom},
     \def\locpx  {1.5}
     \def\locpy  {5.7}
  \be
  \raisebox{-7.6em} {\begin{tikzpicture}
  \coordinate (cooal1) at (0.14*\locpx,0.21*\locpy) ;
  \coordinate (cooal2) at (-0.7*\locpx,0.36*\locpy) ;
  \coordinate (cooev1) at (0.9*\locpx,0.47*\locpy) ;
  \coordinate (cooevb) at (0.81*\locpx,0.47*\locpy) ;
  \coordinate (cooev2) at (0.9*\locpx,0.78*\locpy) ;
  \coordinate (coom1)  at (\locpx,0) ;
  \coordinate (coom2)  at (\locpx,\locpy) ;
  \coordinate (cooc1)  at (0,0) ;
  \coordinate (cooc2)  at (-\locpx,0) ;
  \scopeArrow{0.23}{\arrowDefect} \scopeArrow{0.74}{\arrowDefect}
  \draw[\colorC,line width=\widthObj,postaction={decorate}] (cooc1) node[below]{$c_1$}
       [out=90,in=225] to node[rotate=40,left,very near end,yshift=4pt,xshift=-3pt]
       {$\BB^{m_2,m_1}$} (cooevb) ; 
  \end{scope} \end{scope}
  \scopeArrow{0.18}{\arrowDefect} \scopeArrow{0.57}{\arrowDefect}
  \draw[\colorC,line width=\widthObj,postaction={decorate}] (cooc2) node[below]{$c_2$}
       [out=90,in=210] to node[rotate=45,left,near end,yshift=8pt,xshift=19pt]
       {$\BB^{m_3,m_2}$} (cooev2) ; 
  \end{scope} \end{scope}
  \scopeArrow{0.15}{\arrowDefect} \scopeArrow{0.64}{\arrowDefect}
  \scopeArrow{0.95}{\arrowDefect}
  \draw[\colorM,line width=\widthObj,postaction={decorate}]
       (coom1) node[below]{$m_1$} [out=105,in=270] to (cooev1) -- node[right=-1pt,yshift=3pt]
       {$m_2$} (cooev2) [out=90,in=258] to (coom2) node[above=-1pt]{$m_3$} ;
  \end{scope} \end{scope} \end{scope}
  \node[line width=0.7*\widthboxline,fill=blue!10,draw=blue] at (cooev1)
       {$\scriptstyle\iev_{m_1,m_2}$} ;
  \node[line width=0.7*\widthboxline,fill=blue!10,draw=blue] at (cooev2)
       {$\scriptstyle\iev_{m_2,m_3}$} ;
  \node[rotate=-20,rounded corners,line width=\widthboxline,fill=blue!10,draw=\colorBdyfield]
       at (cooal1) {$\scriptstyle\widetilde\alpha_1$} ;
  \node[rotate=-24,rounded corners,line width=\widthboxline,fill=blue!10,draw=\colorBdyfield]
       at (cooal2) {$\scriptstyle\widetilde\alpha_2$} ;
  \end{tikzpicture}
  }
  ~\qquad = \qquad
  \raisebox{-7.6em} {\begin{tikzpicture}
  \coordinate (cooal1) at (-0.02*\locpx,0.24*\locpy) ;
  \coordinate (cooal2) at (-0.98*\locpx,0.24*\locpy) ;
  \coordinate (cooev)  at (0.9*\locpx,0.77*\locpy) ;
  \coordinate (cooevb) at (0.81*\locpx,0.77*\locpy) ;
  \coordinate (coom1)  at (\locpx,0) ;
  \coordinate (coom2)  at (\locpx,\locpy) ;
  \coordinate (coomu)  at (-0.5*\locpx,0.46*\locpy) ;
  \coordinate (cooc1)  at (0,0) ;
  \coordinate (cooc2)  at (-\locpx,0) ;
  \scopeArrow{0.23}{\arrowDefect} \scopeArrow{0.85}{\arrowDefect}
  \draw[\colorC,line width=\widthObj,postaction={decorate}]
       (cooc1) node[below] {$c_1$} [out=90,in=-25]
       to node[sloped,near end,above=-1pt,xshift=1pt] {$\BB^{m_2,m_1}$} (coomu) ;
  \draw[\colorC,line width=\widthObj,postaction={decorate}]
       (cooc2) node[below] {$c_2$} [out=90,in=205]
       to node[sloped,near end,above=-2pt,xshift=1pt] {$\BB^{m_3,m_2}$} (coomu) ;
  \end{scope} \end{scope}
  \scopeArrow{0.75}{\arrowDefect}
  \draw[\colorC,line width=\widthObj,postaction={decorate}] (coomu) [out=90,in=250]
       to node[sloped,midway,above=-3pt,xshift=-5pt] {$\BB^{m_3,m_1}$} (cooevb) ;
  \end{scope}
  \scopeArrow{0.24}{\arrowDefect} \scopeArrow{0.95}{\arrowDefect}
  \draw[\colorM,line width=\widthObj,postaction={decorate}]
       (coom1) node[below]{$m_1$} [out=105,in=270] to (cooev) 
       [out=90,in=258] to (coom2) node[above=-1pt]{$m_3$} ;
  \end{scope} \end{scope}
  \node[line width=0.7*\widthboxline,fill=blue!10,draw=blue] at (cooev)
       {$\scriptstyle\iev_{m_1,m_3}$} ;
  \node[rotate=4,rounded corners,line width=\widthboxline,fill=blue!10,draw=\colorBdyfield]
       at (cooal1) {$\scriptstyle\widetilde\alpha_1$} ;
  \node[rotate=-4,rounded corners,line width=\widthboxline,fill=blue!10,draw=\colorBdyfield]
       at (cooal2) {$\scriptstyle\widetilde\alpha_2$} ;
  \node[rounded corners,line width=\widthboxline,fill=blue!10,draw=\colorBdyfield]
       at (coomu) {$\scriptstyle\imu$} ;
  \end{tikzpicture}
  }
  \label{eq:bdyOPE}
  \ee
Here $\imu \,{\equiv}\, \imu{}_{m_3,m_2,m_1}$ is the canonical multiplication \eqref{eq:mult-iHom}
of internal Homs. In terms of OPEs, this means that the operator product of the boundary fields
$\Psi^{m_3,m_2;\alpha_2}_{c_2}$ and $\Psi^{m_2,m_1;\alpha_1}_{c_1}$ is the field
$\Psi^{m_3,m_2;\alpha}_{c_2\otimes c_1}$ with $\alpha \iN \Hom_\M((c_2\oti c_1)\act m_1,m_3)$
corresponding to the composition
  \be
  \widetilde\alpha := \imu{}_{m_1,m_2,m_3} \circ (\widetilde\alpha_2 \oti \widetilde\alpha_1)
  :\quad c_2\oti c_1 \to \BB^{m_3,m_1} .
  \label{eq:bdyOPEformula}
  \ee
In line with the different possible interpretations of \eqref{eq:pic1} described above,
the pictures \eqref{eq:pic2} and \eqref{eq:bdyOPE} may be either regarded as 
equalities of morphisms in \C\
or as equalities of morphisms in \M, and in the semisimple case also as equalities of
the invariants that a three-dimensional topological field theory associates to two ribbon 
graphs in a three-manifold that locally differ in the way indicated in the pictures.

We conclude that the boundary OPE is indeed captured by the canonical associative 
multiplication of internal Homs for the module category \M.
Note that the description \eqref{eq:bdyOPEformula} of the boundary OPE is
\emph{relative} to the tensor product in \C\ and cannot, in general, be simplified further,
simply because the tensor product of two objects is generically not fully reducible, not
even if both objects themselves are simple. In contrast, if \C\ is semisimple, then we can
restrict our attention to simple objects $c_1 \eq U_i$ and $c_2 \eq U_j$ and use the semisimple
decomposition $U_i \oti U_j \,{\cong}\, \bigoplus_k \Hom_\C(U_i \oti U_j,U_k)\oti U_k$
(with the summation ranging over a set of representatives for the isomorphism classes
of simple objects, as in \eqref{eq:innerhomssi}) to write the OPE in the familiar\,%
 \footnote{~The dependence of the coefficients on the positions of the fields on the
 world sheet, and thus in particular their pole structure, is obtained when realizing 
 the conformal blocks explicitly as meromorphic sections of vector bundles over the
 moduli space of conformal structures of the world sheet. In our context, invoking
 a Riemann-Hilbert correspondence allows one to suppress this purely chiral issue.}
form
  \be
  \Psi^{m_3,m_2;\alpha_2}_{U_j} * \Psi^{m_2,m_1;\alpha_1}_{U_i}
  = \sum_{k,\gamma} C^{m_1,m_2,m_3;\alpha_1,\alpha_2}_{j,i;k,\gamma}\, \Psi^{m_1,m_3;\gamma}_{U_k} 
  \label{eq:bdyOPEssi}
  \ee
with the $\gamma$-summation being over a basis of $\Hom_\C(U_j \oti U_i,U_k)$.
The index structure of the coefficients $C^{m_1,m_2,m_3;\alpha_1,\alpha_2}_{j,i;k,\gamma}$ 
appearing here is the same as the one of 6j-symbols. And indeed it is easily recognized
that for the local conformal field theory obtained when taking \M\ to be \C\ as a module 
category over itself, these OPE coefficients are precisely the 6j-symbols for the 
monoidal category \C\ \cite{fffs2,fffs3}, while for general local conformal 
field theories, they are the mixed 6j-symbols for the \C-module category \M\
(see e.g.\ \Cite{Sect.\,4.2.1}{bppz2} and \Cite{Sect.\,2.1}{fuRs10}).

Now the crossing symmetry condition \eqref{eq:relation9c} on the boundary OPE 
-- when allowing for arbitrary choices of incoming versus outgoing boundary field insertions --
amounts to the requirement that the algebras of boundary fields that preserve a boundary
condition $m$, and thus the internal Homs $\iHom(m,m)$ of the module category \M, are not
just algebras but even symmetric Frobenius algebras.
It has been shown \cite{schaum} that the module categories over semisimple modular categories
which are equivalent to the category of modules over a Frobenius algebra are those which have a 
\emph{module trace}, i.e.\ \Cite{Def.\,3.7}{schaum} a collection of linear maps
$\Hom_\M(m,m) \,{\rarr~}\, \complex$ satisfying natural consistency conditions. For general
pivotal finite tensor categories, similar results are available \cite{schaum5,shimi20}. 
(Recall that modular tensor categories come with a distinguished pivotal structure.) 
It turns out that validity of the crossing symmetry \eqref{eq:relation9c} for general boundary 
fields that are allowed to change the boundary condition, i.e.\ for all internal Homs 
$\iHom(m,n)$ of \M, amounts to requiring that \M\ is a \emph{pivotal} module category
over a pivotal finite tensor category, a notion that is defined as follows
($\!$\Cite{Def.\,5.2}{schaum5} and \Cite{Def.\,3.11}{shimi20}):

\begin{Defi}
A \emph{pivotal module category} over a pivotal finite tensor category \C\ is a module category 
\M\ over \C\ such that there are functorial isomorphisms $\iHom(m,n)^\vee \,{\cong}\,
\iHom(n,m)$, for $m,n\iN\M$, compatible with the pivotal structure of \C.
\end{Defi}

The collection of such isomorphisms is called a pivotal structure on \M\ and is denoted by
$\pi^\M$. By a Schur lemma-type argument, it can be shown \Cite{Lemma\,3.12}{shimi20}
that a pivotal structure on an indecomposable module category $\calm$ (if it 
exists) is unique up to a scalar multiple.
In a bit more detail, a pivotal module category admits \emph{relative Serre functors}
\cite{fScS2} and is thus an exact module category and, moreover, the relative 
Serre functors are trivialized as twisted module functors. The existence of such 
a trivialization can be regarded as Calabi-Yau type condition \cite{costK2}.
It fits with this point of view that \Cite{Thm.\,3.15}{shimi20} for a pivotal module
category for any $m\iN\calm$ the algebra $\iHom(m,m)$ in \C\ has the structure of a
symmetric Frobenius algebra. Its Frobenius counit is the composition
  \be
  \underline\varepsilon{}_m^{} \Colon \iHom(m,m)
  \rarr{(\pi^\M_{m,m})^{-1}_{}} \iHom(m,m)^\vee
  \rarr{(\underline\eta{}_m)^\vee_{}} \one_\C^\vee = \one_\C^{} \,,
  \label{eq:iHom-counit}
  \ee
where $\pi^\M_{m,m}$ is the $m$-$m$-component of the pivotal structure $\pi^\M$
and $\underline\eta{}_m\colon \one_\C \,{\to}\, \iHom(m,m)$ is the unit morphism
of the unital algebra $\iHom(m,m)$.

These observations lead us to impose two further requirements. 
The first of these specifies an additional input needed beyond the chiral data,
while the second describes the precise role played by the additional datum:

\begin{Assumption}
Within the mathematical framework of finite categories, an indecomposable pivotal module
category $\calm$ over a modular tensor category $\calc$
specifies a full local conformal field theory whose chiral data are encoded in \C.
\end{Assumption}

\begin{Assumption}
The objects of the pivotal module category \M\ are the possible boundary conditions,
the internal Homs $\iHom(m,n)\iN\calc$ provide the boundary fields that change the 
boundary condition from $m\iN\M$ to $n\iN\M$, and the composition \eqref{eq:mult-iHom}
of internal Homs describes the boundary OPE.
\end{Assumption}

The classification of indecomposable module categories over a given modular category is,
in general, a hard problem, and deciding whether a given module category is pivotal is 
difficult as well.  But for any modular category $\calc$, there is at least one example of 
an indecomposable pivotal module category: $\calc$ seen as a module category over itself
-- that it is pivotal as a module category follows directly from the fact that it is pivotal
as a tensor category.
This particular example $\M \eq \C$ is commonly referred to as the \emph{Cardy case}. It is
immediate that for $\M \eq \C$ the boundary conditions are in bijection with the objects
$c$ of $\calc$; the boundary fields relating two boundary conditions $c$ and $c'$ are given 
by $\iHom_\calc(c,c') \eq c'\oti \Vee c$,  which for semisimple \C\ is a special case 
of \eqref{eq:iHom-otA}. Beyond the Cardy case, simple current techniques 
\cite{scya6,fuRs9} allow one to construct examples of indecomposable module categories
that can be realized as categories of modules over Frobenius algebras whose underlying
objects are direct sums of invertible objects of \C.  (In classifications of full local
conformal field theories, often the letter D is used to denote the corresponding models.)


\subsection{Defect conditions and defect fields} \label{sec:dc+df}

To account for Assumptions 1\,--\,3 we fix a modular tensor category \C\ and an
indecomposable pivotal module category \M\ over \C. This may be rephrased by saying 
that we take the chiral data as well as all boundary fields, including their OPE, as an
input.\,%
 \footnote{~The amount of independent input data is in fact considerably smaller than
 it might appear. Indeed, it suffices to know a single boundary condition
 and the Frobenius algebra $A\iN\calc$ of boundary fields which preserve that
 boundary condition. The category \M\ can then be recovered, as a pivotal module category,
 as the category $\Amod$ of right $A$-modules.}	
Our goal is to construct from this input the bulk fields and their operator products.

It is most natural -- and also helps to clarify the conceptual setup -- to investigate
not only bulk fields, but also general defect fields in the bulk (including, as another
special case, disorder fields). To this end we must first provide the possible
\emph{types of defect lines}, or `defect conditions'. All defects considered here
are topological and preserve the full chiral symmetry \C. Unlike more general defects which
are of interest as well, such as conformal ones, topological defects automatically come with
a topological fusion product. Moreover, among the topological defects there are the invertible
defects and the duality defects, which allow one \cite{ffrs5} to extract symmetries 
and order-disorder dualities, respectively, of a full CFT.

A defect line can separate regions supporting two different full conformal field theories 
that are built on the same chiral CFT. Defect fields can change the type of defect line. We
therefore now consider  a pair of pivotal module categories \M\ and $\M'$ over \C\ assigned
to regions of the world sheet that are separated by a defect line, as well as a point-like
insertion $\DD$ on the defect at which the defect condition changes, say from $G$ to $H$.
This  local situation on the world sheet is illustrated in the following picture:
     \def\locpx  {1.45}
     \def\locpy  {3.55}
  \be
  \raisebox{-4.3em}
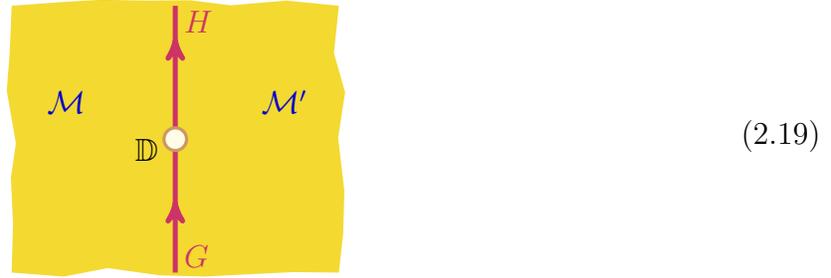
{\begin{tikzpicture}[decoration={random steps,segment length=6mm}]
  \fill[\colorFace,decorate]
       (-1.5*\locpx,0) -- ++(0,\locpy) -- ++(3*\locpx,0) -- ++(0,-\locpy) -- cycle ;
  \scopeArrow{0.27}{\arrowDefect} \scopeArrow{0.88}{\arrowDefect}
  \draw[\colorDefect,line width=\widthObj,postaction={decorate}]
       (0,0) node[right=-1.4pt,yshift=6pt]{$G$} -- +(0,\locpy)
       node[right=-1.5pt,yshift=-6pt]{$H$} ;
  \end{scope} \end{scope}
  \filldraw[fill=\colorHole,draw=\colorDeffield,very thick] (0,0.5*\locpy) circle (0.15 cm)
       node[black,left=2pt,yshift=-4pt] {$\DD$} ;
  \node[thick,color=\colorM] at (-\locpx,0.64*\locpy) {$\M$} ;
  \node[thick,color=\colorM] at (\locpx,0.64*\locpy) {$\M'$} ;
  \end{tikzpicture}
  }
  \label{pic:defectline}
\ee

In the special situation that the modular category \C\ is semisimple, according to
\cite{fuRs4} a full conformal field theory is given by a simple special symmetric Frobenius
algebra $A$ -- determined up to Morita equivalence -- and the defect conditions for
topological defects separating the full conformal field theories characterized by two such
algebras $A$ and $A'$ are given by $A$-$A'$-bi\-mo\-dules. In a Morita invariant formulation, 
the role of the algebras $A$ and $A'$ is taken over by indecomposable semisimple (and thus
pivotal) \C-module categories
\M\ and $\M'$ such that $\M \,{\simeq}\, \Amod$ and $\M' \,{\simeq}\, \text{mod-}A'$. A
defect condition is then an object of the category $\Fun_\C(\M,\M')$ of \C-module functors
between \M\ and $\M'$. Any such a functor is isomorphic to the functor of taking the tensor 
product over $A$ with a suitable $A$-$A'$-bimodule.  Now tensoring with a bimodule is a
\emph{right exact functor}, even when \C\ is no longer semisimple. We are thus led to make

\begin{Assumption} \label{Ass:defectconditions}
The defect conditions for topological defects that separate full conformal field theories 
described by indecomposable pivotal left \C-module categories \M\ and $\M'$ are the objects
of the category $\Funre_\C(\M,\M')$ of right exact \C-module functors.
\end{Assumption}

In case $\M' \eq \M$, adopting a frequent practice in the literature we write 
$\Funre_\C(\M,\M) \,{=:}\, \calcm$. The functor category $\calcm$ is again a finite tensor 
category, with tensor product given by the composition of functors; it is not braided. By
\Cite{Thm.\,3.13}{shimi20}, $\calcm$ has a pivotal structure, allowing us in particular to
identify left and right duals and thus to describe the orientation reversal of a defect
line unambiguously as replacing the object that gives the defect condition by its dual.
More generally, the composition of functors also provides us with an associative
multiplication $\Funre_\C(\M',\M'') \Times \Funre_\C(\M,\M') \rarr{} \Funre_\C(\M,\M'')$.
A natural interpretation of this multiplication is that the composition of 
module functors describes the \emph{fusion} of topological defect lines. In particular, 
the tensor product on $\calcm$ describes the fusion of topological defect lines
in a full conformal field theory given by \M. Moreover, for $\M \eq \M'$
there is a \emph{transparent defect line}, namely the one that corresponds to the 
module endofunctor $\Id_\M\iN\Funre_\C(\M,\M)$ (which is clearly right exact).
Bulk fields are just those defect fields which preserve the transparent defect.

By studying the sewing conditions for correlation functions of bulk fields, accounting
in particular for the distinction between incoming and outgoing insertions, one learns 
(see \Cite{Prop.\,4.7}{fuSc22}, and for the semisimple case also \cite{koRu}) that the
space of bulk fields in addition has in particular a coalgebra structure, and that the
algebra and coalgebra structures naturally fit into the structure of a symmetric
Frobenius algebra. A similar analysis of four-point correlators 
    \def\locpx  {2.3}    
    \def\locpy  {1.8}    
  \be
  \raisebox{-6.3em} {\begin{tikzpicture}
  \FourDefectFieldsOnSphereLabeled {G_1}{G_2}{G_3}{G_4}{\DD} \end{tikzpicture}}
  \ee
involving general defect fields reveals that the space of defect fields on a defect of
fixed type must carry a natural structure of a symmetric Frobenius algebra as well.
Moreover, the algebra of bulk fields must in addition be commutative.
That the bulk algebra is commutative symmetric Frobenius is precisely what is needed to
satisfy the crossing symmetry condition \eqref{eq:relation9a} for the bulk OPE.


\subsection{The problem: Reconstructing the bulk from the boundary} \label{sec:theproblem}

We are now ready to precisely formulate the problem to which 
we are going to propose a solution:

\begin{Problem}
Let $\calc$ be a modular tensor category and let $\M$, $\M'$ and $\M''$ be 
indecomposable pivotal module categories over \C.
\Enumerate

\item \label{p1}
For each pair of defect conditions $G,G'\iN \Funre_\C(\M,\M')$ for topological defects
separating \M\ and $\M'$, provide an object 
  \be
  {\DD}^{G,G'} \in \CbC \simeq \ZC
  \ee
that describes the space of defect fields which change the defect condition from $G$ to $G'$.

\item \label{p2}
Given three defect conditions $G,G',G''\iN \Funre_\C(\M,\M')$, provide an associative 
composition
  \be
  {\DD}^{G',G''} \!\otimes {\DD}^{G,G'} \rarr{~} {\DD}^{G,G''}
  \label{eq:eqvertical}
  \ee
in \ZC\ that describes the operator product of two defect fields on a defect line
separating \M\ and $\M'$, in such a way that the associative algebras $\DD^{G,G}$
come with a natural structure of a symmetric Frobenius algebra and that the bulk algebra
$\DD^{\Id_\M,\Id_\M}$ is in addition commutative in \ZC.

\item \label{p3}
Given two pairs of segments of topological defect lines, with defect
conditions $G,H\colon \M\,{\rarr{}}\, \M'$ and $G',H'\colon \M'\,{\rarr{}}\, \M''$, 
respectively, provide an associative composition
  \be
  {\DD}^{G,H} \otimes {\DD}^{G',H'} \rarr{~} {\DD}^{G'\circ G,H'\circ H}
  \label{eq:eqhorizontal}
  \ee
that describes what happens to defect fields upon fusing the segments of defect lines
pairwise to segments labeled by defect conditions 
$G'\cir G\colon \M \,{\rarr{}}\, \M''$ and $H'\cir H\colon \M \,{\rarr{}}\, \M''$.

\item \label{p4}
Finally, obtain natural bulk-boundary OPEs, corresponding
to the third building block in \eqref{eq:3opes}.
\end{enumerate}
\end{Problem}

These problems have already been completely solved for the case that the modular
tensor category $\calc$ is semisimple \cite{fuRs4,fuRs10}. It will thus be an important 
check of the proposal we are going to formulate that it reproduces these results 
when \C\ is semisimple.

Pictorially, the compositions \eqref{eq:eqvertical} and \eqref{eq:eqhorizontal} amount to
     \def\locpr  {0.15}
     \def\locpx  {1.45}
     \def\locpy  {4.95}
  \be
  \raisebox{-5.9em}{\begin{tikzpicture}[decoration={random steps,segment length=6mm}]
  \fill[\colorFace,decorate]
       (-1.5*\locpx,0) -- ++(0,\locpy) -- ++(3*\locpx,0) -- ++(0,-\locpy) -- cycle ;
  \scopeArrow{0.17}{\arrowDefect} \scopeArrow{0.56}{\arrowDefect}
  \scopeArrow{0.90}{\arrowDefect}
  \draw[\colorDefect,line width=\widthObj,postaction={decorate}]
       (0,0) node[left=-1.4pt,yshift=6pt]{$G$} -- node[left=-1.1pt,yshift=-11pt]{$G'$}
       +(0,\locpy) node[left=-1.1pt,yshift=-6pt]{$G''$} ;
  \end{scope} \end{scope} \end{scope}
  \filldraw[fill=\colorHole,draw=\colorDeffield,very thick]
       (0,0.31*\locpy) circle (\locpr) node[black,right=2pt,yshift=-4pt] {$\DD^{G,G'}$} 
       (0,0.69*\locpy) circle (\locpr) node[black,right=2pt,yshift=-4pt] {$\DD^{G',G''}$} ;
  \node[thick,color=\colorM] at (-\locpx,0.82*\locpy) {$\calm$} ;
  \node[thick,color=\colorM] at (\locpx,0.82*\locpy) {$\caln$} ;
  \end{tikzpicture}
  }
  \qquad\longmapsto\qquad
  \raisebox{-5.9em}{\begin{tikzpicture}[decoration={random steps,segment length=6mm}]
  \fill[\colorFace,decorate]
       (-1.5*\locpx,0) -- ++(0,\locpy) -- ++(3*\locpx,0) -- ++(0,-\locpy) -- cycle ;
  \scopeArrow{0.23}{\arrowDefect} \scopeArrow{0.88}{\arrowDefect}
  \draw[\colorDefect,line width=\widthObj,postaction={decorate}]
       (0,0) node[left=-1.4pt,yshift=6pt]{$G$} -- +(0,\locpy)
       node[left=-1.1pt,yshift=-6pt]{$G''$} ;
  \end{scope} \end{scope}
  \filldraw[fill=\colorHole,draw=\colorDeffield,very thick] (0,0.5*\locpy) circle (\locpr)
       node[black,right=2pt,yshift=-4pt] {$\DD^{G,G''}$} ;
  \node[thick,color=\colorM] at (-\locpx,0.80*\locpy) {$\calm$} ;
  \node[thick,color=\colorM] at (\locpx,0.80*\locpy) {$\caln$} ;
  \end{tikzpicture}
  }
  \label{eq:verticalOPE}
  \ee
and to\,%
 \footnote{~Note that the order of the terms in the composition of functors is, according
 to standard conventions, opposite to the order of factors that would arise
 when describing the fusion of defect lines as a tensor product.}
     \def\locpr  {0.15}
     \def\locpx  {1.95}
     \def\locpy  {4.25}
  \be
  \raisebox{-4.8em}{\begin{tikzpicture}[decoration={random steps,segment length=6mm}]
  \fill[\colorFace,decorate]
       (-1.5*\locpx,0) -- ++(0,\locpy) -- ++(3*\locpx,0) -- ++(0,-\locpy) -- cycle ;
  \scopeArrow{0.23}{\arrowDefect} \scopeArrow{0.88}{\arrowDefect}
  \draw[\colorDefect,line width=\widthObj,postaction={decorate}]
       (-0.39*\locpx,0) node[left=-1.4pt,yshift=6pt]{$G$} -- +(0,\locpy)
       node[left=-2.1pt,yshift=-6pt]{$H$} ;
  \draw[\colorDefect,line width=\widthObj,postaction={decorate}]
       (0.39*\locpx,0) node[left=-1.4pt,yshift=6pt]{$G'$} -- +(0,\locpy)
       node[left=-2.1pt,yshift=-6pt]{$H'$} ;
  \end{scope} \end{scope}
  \filldraw[fill=\colorHole,draw=\colorDeffield,very thick]
       (-0.39*\locpx,0.5*\locpy) circle (\locpr) node[black,left=2pt,yshift=-4pt] {$\DD^{G,H}$}
       (0.39*\locpx,0.5*\locpy) circle (\locpr) node[black,right=2pt,yshift=-4pt] {$\DD^{G',H'}$} ;
  \node[thick,color=\colorM] at (-\locpx,0.74*\locpy) {$\calm$} ;
  \node[thick,color=\colorM] at (0,0.74*\locpy) {$\calm'$} ;
  \node[thick,color=\colorM] at (\locpx,0.74*\locpy) {$\calm''$} ;
  \end{tikzpicture}
  }
  \qquad\longmapsto\qquad
     \def\locpx  {1.63}
  \raisebox{-4.8em}{\begin{tikzpicture}[decoration={random steps,segment length=6mm}]
  \fill[\colorFace,decorate]
       (-1.5*\locpx,0) -- ++(0,\locpy) -- ++(3*\locpx,0) -- ++(0,-\locpy) -- cycle ;
  \scopeArrow{0.23}{\arrowDefect} \scopeArrow{0.88}{\arrowDefect}
  \draw[\colorDefect,line width=\widthObj,postaction={decorate}]
	  (0,0) node[left=-1.4pt,yshift=5pt]{$G'{\circ}G$} -- +(0,\locpy)
	  node[left=-2.1pt,yshift=-5pt]{$H'{\circ}H$} ;
  \end{scope} \end{scope}
  \filldraw[fill=\colorHole,draw=\colorDeffield,very thick] (0,0.5*\locpy) circle (\locpr)
       node[black,right=2pt,yshift=-4pt] {$\DD^{G'\circ G,H'\circ H}$} ;
  \node[thick,color=\colorM] at (-0.85*\locpx,0.74*\locpy) {$\calm$} ;
  \node[thick,color=\colorM] at (0.85*\locpx,0.74*\locpy) {$\calm''$} ;
  \end{tikzpicture}
  }
  \label{eq:horizontalOPE}
  \ee
respectively. As suggested by these pictures, we choose the terminology 
\emph{vertical OPE} for the operator product \eqref{eq:eqvertical} along a defect line,
and \emph{horizontal OPE} for the operator product \eqref{eq:eqhorizontal} that 
is accompanied by the fusing of defect lines.
In the case of bulk fields, the vertical OPE is just the ordinary
bulk OPE, which is the second building block in \eqref{eq:3opes}.


\section{Proposal: Defect fields in finite conformal field theory} \label{sec:3}

A simple observation that motivates our proposal for defect fields and their OPE 
is the fact that the Poincar\'e dual of the picture \eqref{pic:defectline}, i.e.
  \be
  \begin{tikzpicture}
  \node[left,color=\colorM] at (0,0) {$\calm$};
  \node[right,color=\colorM] at (1.96,0) {$\caln$};
  \draw[very thick,color=\colorDefect,->] (0,0.1) .. controls (0.7,0.8) and (1.3,0.8) .. (2,0.1);
  \node[color=\colorDefect] at (1.02,0.91) {$G$};
  \draw[very thick,color=\colorDefect,->] (0,-0.1) .. controls (0.7,-0.8) and (1.3,-0.8) .. (2,-0.1);
  \node[color=\colorDefect] at (1.02,-0.93) {$G'$};
  \draw[thick,double,->] (1,0.5) -- node[right=-1.4pt] {$\scriptstyle \DD$} (1,-0.5) ;
  \end{tikzpicture}
  \label{pic:nattrafo}
  \ee
is reminiscent of the standard graphical description of natural transformations between 
functors. Moreover, the vertical OPE considered in  Part \ref{p2} of our problem
is reminiscent of vertical composition, and the OPE considered in 
Part \ref{p3} of horizontal composition of natural transformations. 

On the negative side, module natural transformations form a finite-dimensional vector 
space. In contrast, what we need to describe defect fields are \emph{objects} 
${\DD}^{G,G'}{\in}\,\ZC \,{\cong}\, \CbC$. But as it turns out, these objects can still be
described in close analogy with natural transformations. Indeed, in \cite{fuSc25} it has
been shown that they can be constructed as \emph{internal natural transformations}.
In the next subsection we briefly explain the theory of those objects.


\subsection{Internal natural transformations} \label{sec:iNat}

We first need to recall a basic fact about module categories: 

\begin{Prop}
For \M\ and \N\ finite module categories over a finite tensor category \C,
the finite category $\Funre_\C(\M,\N)$ of right exact module functors is a finite 
module category over the Drinfeld center \ZC\ $($which is a finite tensor category$)$. 
\end{Prop}

Indeed it is readily checked that, for any object $z\iN\ZC$ in the Drinfeld center and 
any module functor $G\iN \Funre_\C(\M,\N)$, the functor $(z\act G) \iN \Funre(\M,\N)$ 
defined by
  \be
  (z\act G)(m) := \dot z\act (G(m)) \,,
  \label{eq:z.G}
  \ee
with $\dot z \iN \C$ the object in \C\ underlying the object $z \iN \ZC$,
becomes a \C-module functor via the isomorphisms
  \be
  (z\act G)(c\act m) = \dot z\act G(c\act m) \rarr\cong (\dot z \oti c)\act G(m) 
  \rarr \cong (c \oti \dot z)\act G(m) = c\act \big( (z\act G)(m) \big)
  \ee
for all $c\iN\C$. Here in the first isomorphism we use the module functor structure 
of $G$ and in the second isomorphism the $c$-component of the half-braiding of $z$.

In view of this result it is natural to study the internal Homs $\iHom(G,H)$ 
for module functors $G,H \iN \Funre_\C(\M,\N)$. By definition, these are objects in
the Drinfeld center \ZC; their existence is again ensured by the finiteness properties
that are included in our setting. Being internal Homs of a functor category, these objects
have been called \emph{internal natural transformations} in \cite{fuSc25} and are also
denoted by $\relNat(G,H)$. The internal natural transformations come with the standard
associative composition of internal Homs. We will see that these account for the 
vertical OPEs of defect fields.
 
The behavior of internal natural transformations in fact largely parallels
the one of ordinary natural transformations. In particular there is also
a horizontal composition, which is compatible with the vertical composition
in the usual way. Here we highlight two other aspects:
First, the vector space of ordinary natural transformations between two linear 
functors $G,H\colon \M \,{\to}\, \N$ can be written as 
  \be
  \Nat(G,H) \,= \int_{\!m\in\M}\! \Hom_\N(G(m),H(m))
  \label{eq:Nat=end}
  \ee
i.e.\ as an \emph{end} over morphism spaces. 

In case \M\ is finitely semisimple, the end reduces to a sum over iso\-mor\-phism classes
of simple objects of \M. The structure morphisms 
  \be
  \Nat(G,H) \,= \int_{\!m\in\M}\! \Hom_\N(G(m),H(m)) \rarr~ \Hom_\N(G(m'),H(m')) \,,
  \ee
for $m'\iN\M$, of the end are just the components of the natural transformation, and
the defining constraints on the components of the natural transformation are the same as 
the dinaturality relation for the structure morphisms \Cite{p.\,223}{MAcl}.
Recalling that the vertical composition of natural transformations amounts to the composition
of components, we see that for $G \eq H$ these structure maps are morphisms of algebras.

In the situation captured by the picture \eqref{pic:nattrafo}, an expression similar to
\eqref{eq:Nat=end} is valid for internal natural transformations \Cite{Thm.\,9}{fuSc25}:
  \be
  \relNat(G,H) \,= \int_{\!m\in\M} \iHom_\N(G(m),H(m)) \,.
  \label{eq:iNat=end}
  \ee
Concerning this equality one should appreciate the fact that, while for any $m\iN\M$ the 
internal Hom $\iHom_\N(F(m),G(m))$ is an object in \C, the end on the right hand side has a 
natural structure of an object in the Drinfeld center \Cite{Thm.\,8}{fuSc25}.
In particular,
the equality \eqref{eq:iNat=end} is to be understood as an equality of objects in \ZC.


\subsection{The proposal} \label{sec:theproposal}

Motivated by the considerations above we formulate the following 

\begin{Proposal} \label{prp:1} 
 ~\\[2pt]
Let $G$ and $G'$ be types of defect lines separating two full local conformal field theories
based on chiral data that are encoded in a modular category \C\ and in \C-module categories
\M\ and \N, respectively (so that, by Assumption \ref{Ass:defectconditions},
$G,G' \iN \Funre_\C(\M,\N)$).
\Enumerate
	
\item 
The defect fields that separate defect lines labeled by $G$ and $G'$ are given by the
object $\relNat_\M(G,G')$ of internal natural transformations in the Drinfeld center:
  \be
  \DD^{G,G'} = \relNat(G,G') ~\in \ZC \,.
  \label{DD=relNat}
  \ee
In particular, the bulk fields for a full CFT based on \C\ and on a \C-module category \M\
are given by the internal natural transformations from the identity functor to itself,
  \be
  \FF = \relNat(\Id_\calm,\Id_\calm) ~\in \ZC \,,
  \label{FF=relNat}
  \ee
and the disorder fields at which a defect line of type $G$ starts or ends are given
by $\relNat(\Id_\M,G)$ and $\relNat(G,\Id_\M)$, respectively.

\item 
The OPEs of defect fields are given by the horizontal and vertical composition of
these internal natural transformations.
\end{enumerate}
\end{Proposal}

\smallskip

As we will show in Sections \ref{sec:2ssi}\,--\,\ref{sec:9e9d},
our proposal passes significant consistency checks.


\subsection{Comments} \label{sec:rmks}

Before we proceed to these consistency checks we comment on a few immediate    
consequences of our proposal.

\begin{Rem}
The bulk algebra $\FF$ is commutative. For the Cardy case bulk algebra this is e.g.\ shown 
in Lemma 3.5 of \cite{dmno}, which is formulated for semisimple \C, but with a proof
that extends to the non-semisimple case. Our proposal allows for an independent proof: We
can show that for any finite module category \M\ over a finite tensor category \C\
the object $\FF \eq \relNat(\Id_\M,\Id_\M)$ is a commutative algebra in $\ZC$. The proof is 
based on the description of $\FF$ as an end and on a compatibility between the half-braiding
on $\FF$ and the product of boundary fields; we present it in Appendix \ref{sec:app}. It is
worth noting that our proof does not require \C\ to be braided nor to be pivotal.
\end{Rem}

\begin{Rem}
There is a braided equivalence $\theta_\calm \colon \ZC \rarr{} \ZCM$ between the
Drinfeld centers of \C\ and of $\calcm$ \cite{schau13}. It can be shown 
\Cite{Rem.\,7}{fuSc25} that 
  \be
  \theta_\calm (\relNat(\Id_\M,\Id_\M)) \,\cong \int_{\!G\in\calcm}\!\!G\ra\cir G \,,
  \ee
where $G\ra$ is the right adjoint of the functor $G$, i.e.\ the dual of $G$ in the pivotal
category $\calcm$. This means that the bulk algebra
becomes diagonal when regarding it not as an object in \ZC, but instead as an object 
in the equivalent category $\ZCM$. Or, stated more succinctly: When expressed in terms of
module functors, \emph{the torus partition function of any full finite CFT is diagonal.}
\end{Rem}

\begin{Rem}
Applying the argument that in the case of a \C-module category \M\ leads to the expression
\eqref{eq:iHom-intHom} for boundary fields to the \ZC-module category $\Funre_\C(\M,\N)$,
we can exhibit the chiral content of the defect fields as the following coend:
  \be
  \relNat(G,H) \,\cong \int^{z\in\ZC}\!\! \Hom_{\Funre_\C(\M,\N)}(z\act G,H) 
  \,{\otimes_\complex}\, z 
  \label{eq:relNat-multspace}
  \ee
(compare also \Cite{Rem.\,5}{fuSc25}). In the semisimple case this coend reduces to a
direct sum; the corresponding formula will be given in \eqref{eq:deffs=relNat} below. The
expression \eqref{eq:relNat-multspace} may be further combined with the braided equivalence
$\GC\colon \CbC \To \ZC$ to write $\relNat(G,H)$ as a corresponding coend over the category
$\CbC$.
\end{Rem}

\begin{Rem}
We denote by $\FF_\text{Cardy}$ the object of bulk fields in the Cardy case, i.e.\ for
\M\ given by \C\ as a module category over itself. Owing to the adjunction 
  \be
  \Hom_{\Funre_\C(\M_1,\M_2)}(z\act G,H) \cong \Hom_{\ZC}(z,\iHom(G,H)) \,,
  \label{eq:Rex-ZC}
  \ee
for any $z \iN \ZC$ we have
  \be
  \bearll
  \Hom_{\Funre_\C(\C,\C)}(z\act \Id_\C,\Id_\C) \!\! & \cong \Hom_\ZC(z,\FF_\text{Cardy})
  \nxl2 &
  = \Hom_\ZC(z,\text{Coind}(\one_\C)) \cong \Hom_\C(\dot z,\one_\C) \,.
  \eear
  \label{eq:HomRz-HomCzdot}
  \ee
Here $\text{Coind}$ is the coinduction functor from \C\ to \ZC\ and the last isomorphism
holds because $\text{Coind}$ is right adjoint to the forgetful functor. The relation
\eqref{eq:HomRz-HomCzdot} can be used
to obtain a convenient expression for $\FF_\text{Cardy}$. Namely, invoking 
the equivalence \eqref{eq:CbC->ZC} between the center \ZC\ and the enveloping category \CbC\
(together with the fact that taking the coend over the Deligne product \CbC\ can be done 
as a double coend over its two factors \Cite{Cor.\,3.12}{fScS2}), it follows that
  \be
  \bearll
  \FF_\text{Cardy} \!\!\!\! & \dsty \stackrel{\eqref{eq:relNat-multspace}}\cong\!\!
  \int^{z\in\ZC}\!\! \Hom_{\Funre_\C(\C,\C)}(z\act \Id_\C,\Id_\C) \,{\otimes_\complex}\, z
  \Nxl2 & \dsty
  \stackrel{\eqref{eq:HomRz-HomCzdot}}\cong\!\!
  \int^{z\in\ZC}\!\! \Hom_\C(\dot z,\one_\C) \,{\otimes_\complex}\, z \,\cong
  \int^{c,c'\in\C}\!\! \Hom_\C(c \oti c',\one_\C) \,{\otimes_\complex}\, c\,{\ootimes}\,c' ,
  \eear
  \label{eq:FFCardy2}
  \ee
where we use the notation $c\,{\ootimes}\,c' \,{=}\, \GC(c\boti c')$.
Invoking duality and the identity $\int^{c \in \calc}\! \Hom_\calc( c,-) 
   $
   \linebreak[0]${\otimes_\complex}\, G(c) \,{\cong}\, G$ that is valid for any linear functor $G$ 
(compare e.g.\ \Cite{Prop.\,2.7}{fScS2}), we can now rewrite the Cardy case bulk fields as
  \be
  \FF_\text{Cardy} \,\cong \int^{c\in\C}\!\! c \,{\ootimes}\, \Vee c ~\in\ZC \,.
  \label{eq:FFCardy}
  \ee
The result \eqref{eq:FFCardy} agrees with the description of bulk fields in Section 2.2
of \cite{fgsS}. Note that in view of the equivalence \eqref{eq:CbC->ZC} between \ZC\ and
\CbC, it shows in particular that the Cardy case bulk algebra $\FF_\text{Cardy}$ deserves
to be called the \emph{charge conjugate bulk object} also beyond semisimplicity.
\end{Rem}

\begin{Rem} \label{rem:Carl} 
In the particular case that $\C \eq H\text{-mod}$ is the category of modules over
a finite-di\-men\-si\-onal factorizable ribbon Hopf algebra $H$ and \M\ is \C\ as a module
category over itself the bulk object is known \Cite{Sect.\,2.3}{fuSs3} to be the coregular 
$H$-bimodule, with underlying vector space $H^*$. That this object is isomorphic to 
$\FF_\text{Cardy}$ as given in \eqref{eq:FFCardy} can be seen by a categorical variant 
of the Peter-Weyl theorem which states \Cite{Cor.\,2.9}{fScS2} that the coregular 
bimodule can be expressed as the coend $\int^{c\in H\text{-mod}}\! c \boti c^*$ in
$H\text{-bimod} \,{\simeq}\, H\text{-mod} \boti (H\text{-mod})^\text{rev}_{\phantom|}$.
(For the latter equivalence, see Appendix A of \cite{fuSs3}.)
\end{Rem}


\subsection{Consistency check: Recovering the semisimple case} \label{sec:2ssi}

As already mentioned, defect fields are completely understood when \C\ is a 
($\complex$-linear) semisimple modular tensor category \cite{fuRs4,fuRs10}. We now
explain how the description of the object of defect fields in \cite{fuRs4,fuRs10}
is recovered from our proposal. As a crucial ingredient we use the adjunction 
\eqref{eq:def:iHom} that defines an internal Hom.

Let us first recall the pertinent results for the semisimple case. Given a semisimple 
modular category \C\ select, as already done in formula \eqref{eq:innerhomssi}, a set 
$(U_i)_{i\in I_\C}^{}$ of representatives for the isomorphism classes of simple objects of
\C. Write $U_i \,\ootimes\, U_j \,{:=}\, \GC(\overline U_i\boti U_j) \iN \ZC$
for the object in the Drinfeld center to which the functor $\GC$ (see \eqref{eq:CbC->ZC})
maps the simple object $\overline U_i\boti U_j$ of \CbC. Since $\calc$ is modular, 
$\GC$ is a braided equivalence, and hence the objects $(U_i\,\ootimes\, U_j)_{i,j\in I}^{}$
form a set of representatives for the isomorphism classes of simple objects of $\ZC$.
Let $\calm_1$ and $\calm_2$ be indecomposable pivotal module categories over $\calc$. There
are symmetric Frobenius algebras $A_1$ and $A_2$ in $\calc$ (determined up to Morita
equivalence) such that $\calm_i \,{\simeq}\, \text{mod-}A_i$ as $\calc$-module categories.

Now let $G,H\colon \calm_1\To \calm_2$ be \C-module functors. As module functors out of an
exact module category, they are exact functors \Cite{Prop.\,7.6.9}{EGno}. They describe two
types of defect lines, each separating the full conformal field theory that corresponds to 
$\calm_1$ and the one that corresponds to $\calm_2$. There then exist 
$A_1\text{-}A_2$-bimodules $B^G$ and $B^H$ such that we have isomorphisms
  \be
  G \,\cong\, {-}\,{\otimes_{\!A_1}^{}} B^G \qquad\text{and}\qquad
  H \,\cong\, {-}\,{\otimes_{\!A_1}^{}} B^H 
  \label{eq:BG,BH}
  \ee
of module functors. By the results of \cite{fuRs4,fuRs10} (see also the dictionary
in \Cite{Sect.\,7}{fuRs11} for a compact compilation) the object in \CbC\ of defect fields that 
transform the defect line of type $B^G$ into the defect line of type $B^H$ is the direct sum
  \be
  \bigoplus_{i,j\in I_\C} \mathrm Z_{i,j}^{B^G\!,B^H} {\otimes_\complex}\, U_i\ootimes U_j 
  \ee
with multiplicity spaces   given by the spaces
  \be
  \mathrm Z_{i,j}^{B^G\!,B^H} := \Hom_{A_1|A_2}^{}(U_i\,{\otimes^+} B^G {\otimes^-}\, U_j\,, B^H) 
  \label{eq:ZijGH}
  \ee
of $A_1$-$A_2$-bimodule morphisms.

Here $U_i\,{\otimes^+} B^G {\otimes^-}\, U_j$ is the $A_1\text{-}A_2$-bimodule with
underlying object $U_i\oti B^G{\otimes}\, U_j$ for which the left action is obtained by
combining the inverse braiding $c_{U_i,A_1}^{-1}$ with the left $A_1$-action on $B^G$ and 
the right action is given by combining the inverse braiding $c_{A_2,U_j}^{-1}$ 
with the right $A_2$-action on $B^G$ (for details, see e.g.\ 
\Cite{Eqs.\,(2.17),\,(2.18)}{fuRs10}). As described in Section 5.10 of \cite{fuRs4}, the 
dimensions $\mathrm z_{i,j}^{B^G\!,B^H} {=}\, \dim_\complex(\mathrm Z_{i,j}^{B^G\!,B^H})$ of 
the spaces \eqref{eq:ZijGH} are the coefficients of the characters of $U_i\ootimes U_j$ in
the partition function on a torus with defect lines $G$ and $H$ (in the literature,
e.g.\ in \cite{pezu5}, these are also known as twisted partition functions).

We are now in a position to state the following result:

\begin{Prop}
For semisimple \C\ the object of defect fields coincides with the object of internal natural
transformations:
  \be
  \bigoplus_{i,j\in I_\C} \mathrm Z_{i,j}^{B^G\!,B^H} {\otimes_\complex}\, U_i\ootimes U_j
  \,\cong\, \relNat(G,H)
  \label{eq:deffs=relNat}
  \ee
as objects in \ZC.
\end{Prop}

\begin{proof}
The adjunction defining internal natural transformations as an internal Hom can be written as
  \be
  \Hom_\ZC(U_i\ootimes U_j,\relNat(G,H)) \cong
  \Hom_{\Funre_\calc(\calm_1,\calm_2)}((U_i\ootimes U_j)\act G,H) \,.
  \ee
The functor underlying the module functor $(U_i\ootimes U_j)\act G$ 
is tensoring with the $A_1\text{-}A_2$-bimodule that is defined on the object
$U_i\oti U_j\oti B^G\iN\calc$ with right action given by the right $A_2$-action
$\rho_{\mathrm r}^{B^G}$ on $B^G$ and left $A_1$-action given by
$(\id_{U_i\otimes U_j}\oti\rho_{\mathrm l}^{B^G}) \cir(\cB_{U_i\oti U_j;A_1}\oti\id_{B^G})$.
The isomorphism $\id_{U_i}\oti c_{B^G,U_j}^{-1}$ exhibits that this bimodule is isomorphic
to the $A_1\text{-}A_2$-bimodule $U_i\otimes^+B^G\otimes^-U_j$. As a consequence we have
  \be
  \Hom_\ZC(U_i\ootimes U_j,\relNat(G,H))
  \cong \Hom_{A_1|A_2}^{}(U_i\,{\otimes^+} B^G {\otimes^-}\, U_j\,, B^H) \,.
  \ee
This correctly reproduces the spaces
\eqref{eq:ZijGH} and thus gives the correct object of defect fields.
\end{proof}

It follows in particular that for semisimple \C\ and $\M_1 \eq \M_2 \,{=:}\, \M$,
the defect fields possess all the properties listed in Theorem 5.23 of \cite{fuRs4}.
By taking in addition $G \eq H \eq \Id_\M$ we arrive at the following

\begin{Cor}
For semisimple \C\ the object of bulk fields coincides with the internal natural
endotransformations of the identity functor: 
  \be
  \FF = \relNat(\Id_\M,\Id_\M) \cong 
  \bigoplus_{i,j\in I_\C} \mathrm Z_{i,j}^{} \,{\otimes_\complex}\, U_i\ootimes U_j
  \label{eq:FFssi1}
  \ee
with
  \be
  \mathrm Z_{ij}^{} = \Hom_{A|A}(U_i \,{\otimes^+}\! A \,{\otimes^-} U_j, A) \,.
  \label{eq:Zij}
  \ee
\end{Cor}

Since this reproduces the description of bulk fields in \cite{fuRs4}, we can
in particular conclude that in the semisimple case the bulk fields $\relNat(\Id_\M,\Id_\M)$
have the properties listed in Theorem 5.1 of \cite{fuRs4}, including notably
modular invariance of the torus partition function.


\subsection{Consistency check: Operator product of defect fields}

In this section we show that our proposal gives rise to the correct operator products of
defect fields, and thus in particular of bulk fields.
We first note that by the results of \Cite{Sect.\,4.3}{fuSc25} all required horizontal 
and vertical compositions exist and are associative.
Next we use the fact that a modular tensor category $\calc$ is also unimodular. This allows
us to apply \Cite{Cor.\,19}{fuSc25} to conclude that, given two indecomposable pivotal
\C-module categories $\M_1$ and $\M_2$ describing full local conformal field theories with
the same chiral data \C, the \ZC-mo\-du\-le category $\Funre_\C(\M_1,\M_2)$ is again
a pivotal module category. It then follows \Cite{Thm.\,3.15}{shimi20} that all algebras 
$\relNat(G,G)$ are symmetric Frobenius algebras. This shows that we have indeed found 
very natural candidates for defect fields and all OPEs involving defect fields.

\begin{Rem}
Besides the expression \eqref{eq:FFssi1} there is an alternative description of the bulk fields
in the semisimple case: In this case the end in the formula \eqref{eq:iNat=end} is a direct
sum over a set $I_\M$ of representatives $M_\kappa$ for the isomorphism classes of simple
objects of the semisimple module category \M, so that for $G \eq H \eq \Id_\M$ we obtain
  \be
  \FF \cong \bigoplus_{\kappa\in I_\M} \iHom(M_\kappa,M_\kappa)
  \stackrel{\eqref{eq:iHom-otA}}= \bigoplus_{\kappa\in I_\M} M_\kappa \otimes_A\! \Vee M_\kappa \,.
  \label{eq:FFssi2}
  \ee
We conclude in particular that the objects in \ZC\ on the right hand sides of 
\eqref{eq:FFssi1} and \eqref{eq:FFssi2} are isomorphic.
\end{Rem}

To obtain the operator products of defect fields turns out to be quite straightforward. The
crucial observation is that owing
to the fact that $\Funre_\C(\M_1,\M_2)$ is a module category over \ZC\ -- with action given
in \eqref{eq:z.G} -- we can study defect fields fully parallel to the treatment of boundary 
fields in Section \ref{sec:bc+bf}, by simply replacing the role of \C\ by the Drinfeld 
center \ZC\ and the role of the \C-module category \M\ whose objects are boundary conditions 
by the \ZC-module category $\Funre_\C(\M_1,\M_2)$ whose objects are defect conditions.
In particular, by the result \eqref{eq:relNat-multspace} the multiplicity space 
for defect fields changing an $\M_1$-$\M_2$-defect condition $G \iN \Funre_\C(\M_1,\M_2)$
to $H \iN \Funre_\C(\M_1,\M_2)$ and of chiral type given by an object $z \iN \ZC$ 
is the morphism space $\Hom_{\Funre_\C(\M_1,\M_2)}(z\act G,H)$. Accordingly we denote
defect fields changing the defect type from $G$ to $H$ and of chiral type $z$ by
$\Phi^{G,H;\beta}_{z}$ with $z \iN \ZC$ and $\beta \iN \Hom_{\Funre_\C(\M_1,\M_2)}(z\act G,H)$.
Furthermore we can then make use of the adjunction \eqref{eq:Rex-ZC}
to relate the morphism $\beta$ to a morphism $\widetilde\beta \iN \Hom_{\ZC}(z,\iHom(G,H))$
analogously as in the relation \eqref{eq:pic2} for boundary fields, yielding the description
     \def\locpx  {1.5}
     \def\locpy  {3.9}
  \be
  \Phi^{G,H;\beta}_{z} ~~\equiv\quad
  \raisebox{-5.3em} {\begin{tikzpicture}
  \coordinate (cooal) at (0.9*\locpx,0.5*\locpy) ;
  \coordinate (coom1) at (\locpx,0) ;
  \coordinate (coom2) at (\locpx,\locpy) ;
  \coordinate (cooc)  at (0,0) ;
  \scopeArrow{0.59}{\arrowDefect}
  \draw[\colorZ,line width=\widthObj,postaction={decorate}]
       (cooc) node[below]{$z$} [out=90,in=245] to (cooal) ;
  \end{scope}
  \scopeArrow{0.19}{\arrowDefect} \scopeArrow{0.92}{\arrowDefect}
  \draw[\colorDefect,line width=\widthObj,postaction={decorate}]
       (coom1) node[below]{$G$} [out=100,in=270] to (cooal)
       [out=90,in=250] to (coom2) node[above]{$H$} ;
  \end{scope} \end{scope}
  \node[line width=\widthboxline,fill=blue!10,draw=\colorDeffield] at (cooal)
       {$\,\scriptstyle\beta\,$} ;
  \end{tikzpicture}
  }
  ~\quad = \quad
  \raisebox{-5.3em} {\begin{tikzpicture}
  \coordinate (cooal) at (0.16*\locpx,0.35*\locpy) ;
  \coordinate (cooev) at (0.9*\locpx,0.7*\locpy) ;
  \coordinate (coom1) at (\locpx,0) ;
  \coordinate (coom2) at (\locpx,\locpy) ;
  \coordinate (cooc)  at (0,0) ;
  \scopeArrow{0.22}{\arrowDefect}
  \scopeArrow{0.76}{\arrowDefect}
  \draw[\colorZ,line width=\widthObj,postaction={decorate}]
       (cooc) node[below]{$z$} [out=90,in=210] to
       node[rotate=40,left,very near end,yshift=6pt] {$\scriptstyle\DD^{G,H}$} (cooev) ; 
  \end{scope} \end{scope}
  \scopeArrow{0.19}{\arrowDefect} \scopeArrow{0.95}{\arrowDefect}
  \draw[\colorDefect,line width=\widthObj,postaction={decorate}]
       (coom1) node[below]{$G$} [out=100,in=270] to (cooev)
       [out=90,in=250] to (coom2) node[above]{$H$} ;
  \end{scope} \end{scope}
  \node[line width=0.7*\widthboxline,fill=blue!10,draw=purple] at (cooev)
       {$\scriptstyle\ev_{G,H}$} ;
  \node[rotate=-20,rounded corners,line width=\widthboxline,fill=blue!10,draw=\colorDeffield]
       at (cooal) {$\scriptstyle\widetilde\beta$} ;
  \end{tikzpicture}
  }
  \label{eq:pic:DD}
  \ee
Concerning the precise interpretation of this equality, analogous considerations as in the 
case of \eqref{eq:pic1} and \eqref{eq:bdyOPEformula} apply: We can think of it alternatively
as an equality of morphisms in \ZC, an equality of morphisms in $\Funre_\C(\M_1,\M_2)$ or, in 
the semisimple case in which we can invoke the connection with three-dimensional topological 
field theory \cite{fuRs4}, as an equality of invariants of ribbon graphs. Concerning the latter
interpretation, recall that the chiral parts of the ribbon graphs are contained in the
complement of the embedded world sheet in the relevant three-manifold (or, in more fancy
terms, in the holographic direction of the three-manifold). It is worth noting that these
parts must now be labeled by an object of \ZC\ that has the factorized form
$z \eq c \,{\ootimes}\, c'$ (which in the 
semisimple case is no loss of generality), compare e.g.\ the picture (4.38) in \cite{fuRs10}.

Finally, the vertical operator product of defect fields is obtained from the canonical
associative composition of internal Homs analogously as the boundary OPE \eqref{eq:bdyOPE}:
     \def\locpx  {1.5}
     \def\locpy  {5.7}
  \be
  \raisebox{-7.6em} {\begin{tikzpicture}
  \coordinate (cooal1) at (0.11*\locpx,0.18*\locpy) ;
  \coordinate (cooal2) at (-0.7*\locpx,0.36*\locpy) ;
  \coordinate (cooev1) at (0.9*\locpx,0.47*\locpy) ;
  \coordinate (cooevb) at (0.81*\locpx,0.47*\locpy) ;
  \coordinate (cooev2) at (0.9*\locpx,0.78*\locpy) ;
  \coordinate (coom1)  at (\locpx,0) ;
  \coordinate (coom2)  at (\locpx,\locpy) ;
  \coordinate (cooc1)  at (0,0) ;
  \coordinate (cooc2)  at (-\locpx,0) ;
  \scopeArrow{0.18}{\arrowDefect} \scopeArrow{0.74}{\arrowDefect}
  \draw[\colorZ,line width=\widthObj,postaction={decorate}] (cooc1) node[below]{$z_1$}
       [out=90,in=225] to node[rotate=40,left,very near end,yshift=4pt,xshift=-2pt]
       {$\DD^{G_1,G_2}$} (cooevb) ; 
  \end{scope} \end{scope}
  \scopeArrow{0.18}{\arrowDefect} \scopeArrow{0.57}{\arrowDefect}
  \draw[\colorZ,line width=\widthObj,postaction={decorate}] (cooc2) node[below]{$z_2$}
       [out=90,in=210] to node[rotate=45,left,near end,yshift=8pt,xshift=19pt]
       {$\DD^{G_2,G_3}$} (cooev2) ; 
  \end{scope} \end{scope}
  \scopeArrow{0.15}{\arrowDefect} \scopeArrow{0.64}{\arrowDefect}
  \scopeArrow{0.95}{\arrowDefect}
  \draw[\colorDefect,line width=\widthObj,postaction={decorate}]
       (coom1) node[below]{$G_1$} [out=105,in=270] to (cooev1) -- node[right=-1pt,yshift=3pt]
       {$G_2$} (cooev2) [out=90,in=258] to (coom2) node[above=-1pt]{$G_3$} ;
  \end{scope} \end{scope} \end{scope}
  \node[line width=0.7*\widthboxline,fill=blue!10,draw=purple] at (cooev1)
       {$\scriptstyle\iev_{G_1,G_2}$} ;
  \node[line width=0.7*\widthboxline,fill=blue!10,draw=purple] at (cooev2)
       {$\scriptstyle\iev_{G_2,G_3}$} ;
  \node[rotate=-20,rounded corners,line width=\widthboxline,fill=blue!10,draw=\colorDeffield]
       at (cooal1) {$\scriptstyle\widetilde\beta_1$} ;
  \node[rotate=-24,rounded corners,line width=\widthboxline,fill=blue!10,draw=\colorDeffield]
       at (cooal2) {$\scriptstyle\widetilde\beta_2$} ;
  \end{tikzpicture}
  }
  ~\qquad = \qquad
  \raisebox{-7.6em} {\begin{tikzpicture}
  \coordinate (cooal1) at (0.01*\locpx,0.21*\locpy) ;
  \coordinate (cooal2) at (-1.01*\locpx,0.21*\locpy) ;
  \coordinate (cooev)  at (0.9*\locpx,0.77*\locpy) ;
  \coordinate (cooevb) at (0.81*\locpx,0.77*\locpy) ;
  \coordinate (coom1)  at (\locpx,0) ;
  \coordinate (coom2)  at (\locpx,\locpy) ;
  \coordinate (coomu)  at (-0.5*\locpx,0.46*\locpy) ;
  \coordinate (cooc1)  at (0,0) ;
  \coordinate (cooc2)  at (-\locpx,0) ;
  \scopeArrow{0.21}{\arrowDefect} \scopeArrow{0.81}{\arrowDefect}
  \draw[\colorZ,line width=\widthObj,postaction={decorate}]
       (cooc1) node[below] {$z_1$} [out=90,in=-25]
       to node[sloped,near end,above=-1pt,xshift=1pt] {$\DD^{G_1,G_2}$} (coomu) ;
  \draw[\colorZ,line width=\widthObj,postaction={decorate}]
       (cooc2) node[below] {$z_2$} [out=90,in=205]
       to node[sloped,near end,above=-2pt,xshift=1pt] {$\DD^{G_2,G_3}$} (coomu) ;
  \end{scope} \end{scope}
  \scopeArrow{0.69}{\arrowDefect}
  \draw[\colorZ,line width=\widthObj,postaction={decorate}] (coomu) [out=90,in=250]
       to node[sloped,midway,above=-3pt,xshift=-5pt] {$\DD^{G_1,G_3}$} (cooevb) ;
  \end{scope}
  \scopeArrow{0.24}{\arrowDefect} \scopeArrow{0.95}{\arrowDefect}
  \draw[\colorDefect,line width=\widthObj,postaction={decorate}]
       (coom1) node[below]{$G_1$} [out=105,in=270] to (cooev) 
       [out=90,in=258] to (coom2) node[above=-1pt]{$G_3$} ;
  \end{scope} \end{scope}
  \node[line width=0.7*\widthboxline,fill=blue!10,draw=purple] at (cooev)
       {$\scriptstyle\iev_{G_1,G_3}$} ;
  \node[rotate=3,rounded corners,line width=\widthboxline,fill=blue!10,draw=\colorDeffield]
       at (cooal1) {$\scriptstyle\widetilde\beta_1$} ;
  \node[rotate=-3,rounded corners,line width=\widthboxline,fill=blue!10,draw=\colorDeffield]
       at (cooal2) {$\scriptstyle\widetilde\beta_2$} ;
  \node[rounded corners,line width=\widthboxline,fill=blue!10,draw=\colorBdyfield]
       at (coomu) {$\scriptstyle\imu$} ;
  \end{tikzpicture}
  }
  \label{eq:defOPE}
  \ee
Formulated as an OPE, this equality states that the operator product of the defect fields 
$\Phi^{G_2,G_3;\beta_2}_{z_2}$ and $\Phi^{G_1,G_2;\beta_1}_{z_1}$
is the defect field $\Phi^{G_1,G_3;\beta}_{z_2\otimes z_1}$, where 
$\otimes \,{\equiv}\, \otimes_{\ZC}$ is the tensor product of objects in the Drinfeld center 
and $\beta$ is the morphism in $\Hom_{\Funre_\C(\M_1,\M_2)}((z_2\oti z_1)\act G_1,G_3)$ 
that under the internal Hom adjunction \eqref{eq:Rex-ZC} corresponds to the composite morphism
  \be
  \widetilde\beta := \imu_{G_1,G_2,G_3}^{} \circ (\widetilde\beta_2 \oti \widetilde\beta_1)
  \in \Hom_{\ZC}(z_2\oti z_1,\DD^{G_1,G_3}) \,.
  \label{eq:defOPEformula}
  \ee
In particular, the OPE of bulk fields $\Phi^{\beta_1}_{z_1}$ and $\Phi^{\beta_2}_{z_2}$
is graphically represented by
     \def\locpx  {1.5}
     \def\locpy  {4.5}
  \be
  \raisebox{-5.6em} {\begin{tikzpicture}
  \coordinate (cooal1) at (0,0.28*\locpy) ;
  \coordinate (cooal2) at (-\locpx,0.28*\locpy) ;
  \coordinate (cooev1) at (0,0.67*\locpy) ;
  \coordinate (cooev2) at (-\locpx,0.67*\locpy) ;
  \coordinate (cooc1)  at (0,0) ;
  \coordinate (cooc2)  at (-\locpx,0) ;
  \scopeArrow{0.20}{\arrowDefect} \scopeArrow{0.75}{\arrowDefect}
  \draw[\colorZ,line width=\widthObj,postaction={decorate}]
       (cooc1) node[below]{$z_1$} -- node[sloped,near end,below] {$\FF$} (cooev1) ; 
  \draw[\colorZ,line width=\widthObj,postaction={decorate}]
       (cooc2) node[below]{$z_2$} -- node[sloped,near end,above] {$\FF$} (cooev2) ; 
  \end{scope} \end{scope}
  \node[line width=0.7*\widthboxline,fill=blue!10,draw=purple] at (cooev1)
       {$\scriptstyle\ievF$} ;
  \node[line width=0.7*\widthboxline,fill=blue!10,draw=purple] at (cooev2)
       {$\scriptstyle\ievF$} ;
  \node[rounded corners,line width=\widthboxline,fill=blue!10,draw=\colorBulkfield]
       at (cooal1) {$\scriptstyle\widetilde\beta_1$} ;
  \node[rounded corners,line width=\widthboxline,fill=blue!10,draw=\colorBulkfield]
       at (cooal2) {$\scriptstyle\widetilde\beta_2$} ;
  \end{tikzpicture}
  }
  \qquad = \qquad
  \raisebox{-5.6em} {\begin{tikzpicture}
  \coordinate (cooal1) at (0,0.28*\locpy) ;
  \coordinate (cooal2) at (-\locpx,0.28*\locpy) ;
  \coordinate (cooev)  at (-0.5*\locpx,\locpy) ;
  \coordinate (cooc1)  at (0,0) ;
  \coordinate (cooc2)  at (-\locpx,0) ;
  \coordinate (coomu)  at (-0.5*\locpx,0.66*\locpy) ;
  \scopeArrow{0.61}{\arrowDefect}
  \draw[\colorZ,line width=\widthObj,postaction={decorate}]
       (coomu) -- node[sloped,midway,above=-1pt,xshift=5pt] {$\FF$} (cooev) ; 
  \end{scope}
  \scopeArrow{0.18}{\arrowDefect} \scopeArrow{0.71}{\arrowDefect}
  \draw[\colorZ,line width=\widthObj,postaction={decorate}]
       (cooc1) node[below]{$z_1$} -- (cooal1)
       [out=90,in=-25] to node[sloped,midway,above=-6pt,yshift=5pt] {$\FF$} (coomu) ;
  \draw[\colorZ,line width=\widthObj,postaction={decorate}]
       (cooc2) node[below]{$z_2$} -- (cooal2)
       [out=90,in=205] to node[sloped,midway,above=-1pt,xshift=1pt] {$\FF$} (coomu) ;
  \end{scope} \end{scope}
  \node[line width=0.7*\widthboxline,fill=blue!10,draw=purple] at (cooev) {$\scriptstyle\ievF$} ;
  \node[rounded corners,line width=\widthboxline,fill=blue!10,draw=\colorBulkfield]
       at (cooal1) {$\scriptstyle\widetilde\beta_1$} ;
  \node[rounded corners,line width=\widthboxline,fill=blue!10,draw=\colorBulkfield]
       at (cooal2) {$\scriptstyle\widetilde\beta_2$} ;
  \node[rounded corners,line width=\widthboxline,fill=blue!10,draw=\colorBdyfield]
       at (coomu) {$\scriptstyle\imu^\FF$} ;
  \end{tikzpicture}
  }
  \label{eq:bulkOPE}
  \ee
with multiplication $\imu^\FF \,{\equiv}\, \imu{}_{\Id,\Id,\Id}\colon \FF \oti \FF \To \FF$ and
evaluation morphism 
  \be
  \ievF \equiv \iev_{\Id,\Id} :\quad \FF\act\Id \to \Id \,.
  \ee
Here $\Id \,{\equiv}\, \Id_\M$ is the identity module functor on the \C-module category \M, 
which is the monoidal unit of $\Funre_\C(\M,\M)$. The defect line labeled by $\Id$ is thus 
\emph{transparent}, and accordingly is not drawn in the picture (nor in any of the pictures 
below). Note that when interpreting (in the semisimple case) the pictures \eqref{eq:bulkOPE} 
in terms of ribbon graphs, now the whole graphs except for the evaluation morphisms $\iev$ are
contained in the complement of the world sheet in the relevant three-manifold
(i.e., the \emph{connecting manifold} as defined in Section 5.1 of \cite{fuRs4}
and Section3.1.2 of \cite{fuRs10}).

It is worth stressing that the description \eqref{eq:defOPEformula} of the OPE of
defect fields is, again in full analogy with the boundary OPE \eqref{eq:bdyOPEformula},
relative to the tensor product in \ZC. In case \C\ is semisimple, we can restrict our
attention to simple objects of \ZC\ as chiral labels. Using that the isomorphism
classes of the latter are represented by $U_i^{} \,{\ootimes}\, U_i'$ with $i,i'\iN I_\C$,
we then get the analogue
  \be
  \Phi^{G_2,G_3;\beta_2}_{U_j^{}\ootimes U_j'} * \Phi^{G_1,G_2;\beta_1}_{U_i^{}\ootimes U_i'}
  = \sum_\beta \sum_{k,k'\in I_\C}\sum_{\lambda,\lambda'}
  C^{G_1,G_2,G_3;\beta_1,\beta_2;\beta}_{j,i;k,\lambda;j',i';k',\lambda'}
  \, \Phi^{G_1,G_3;\beta}_{U_k^{}\ootimes U_k'} 
  \ee
of the semisimple boundary OPE \eqref{eq:bdyOPEssi}, where the summations range over bases
of morphisms $\beta$ in $\Hom_{\Funre_\C(\M_1,\M_2)}((U_k^{}\,{\ootimes}\, U_k') \act G_1,G_3)$
as well as $\lambda$ and $\lambda'$ in the multiplicity spaces $\Hom_\C(U_j\oti U_i,U_k)$ and 
$\Hom_\C(U_j'\oti U_i',U_k')$, respectively.
The coefficients $C$ appearing here are the structure constants for the product
$\imu \,{\equiv}\, \imu{}_{G_1,G_2,G_3}$. As discussed in Section 2.2 of \cite{fuRs10},
they can be regarded as generalized fusing matrices, analogous to the interpretation 
of the coefficients in \eqref{eq:bdyOPEssi} as (mixed) 6j-symbols.


\subsection{Consistency check: Bulk-boundary operator product}

Recall from \eqref{eq:3opes} that in the setting of \cite{lewe3}
besides the bulk OPE and boundary OPE the third building block of a full CFT is the
bulk-boundary OPE. It is worth pointing out that in this case the terminology operator
`product' is somewhat of a misnomer, as one deals with one input and one output field, 
rather than with two inputs and one output. Still, the terminology is justified, as we 
are free to take a trivial boundary field, with chiral label the monoidal unit $\one_\C$, 
as a second input factor. 

We now show that our proposal naturally leads to an expression for the bulk-boundary OPE as 
well. This OPE captures the situation that a bulk field is moved close to a segment of the
boundary, whereby it induces a boundary field on that segment. If the boundary segment is
labeled by the boundary condition $m \iN \M$,
then the relevant boundary algebra is $\BB^{m.m} \eq \iHom(m,m)$, so that the induced 
boundary field is of type $\Psi^{m,m;\alpha}_c$. The chiral label $c$ of this field is an
object in \C\ that is completely determined by $\FF$: the object $\dot{\FF} \,{:=}\, \UC(\FF)$ 
with $\UC\colon \ZC \,{\to}\, \C$ the forgetful functor, as defined before \eqref{eq:UC.GC}. 

It remains to specify the relevant morphism $\alpha \,{\equiv}\, \alpha^\FF_m$ in
$\Hom_\M(\dot\FF \Act m,m)$ that realizes the bulk-boundary OPE displayed in \eqref{eq:2opes}.
Our proposal provides a distinguished candidate for this morphism: the $m$-component of the 
module natural transformation $\ievF\colon \FF\act\Id \To \Id$. We postulate that this 
is indeed the appropriate morphism, i.e.\ that
  \be
  \alpha^\FF_m = \big(\ievF\big)_m \,.
  \ee
According to the description \eqref{eq:pic2} of boundary fields -- which comes, via the
identity \eqref{eq:tworems2}, from the inner Hom adjunction -- there then exists a unique
morphism $\widetilde\alpha^\FF_m$ in $\Hom_\C(\FF,\BB^{m,m})$ such that
  \be
  (\ievF)_m = \iev_{m,m} \circ (\widetilde\alpha^\FF_m \act \id_m) 
  \label{eq:ievFm=iev_mm.(tildealpha.idm)}
  \ee
or, pictorially,
     \def\locpx  {1.5}
     \def\locpy  {4.1}
     \def\widthDinatmor {0.43}
  \be
  \raisebox{-6.2em} {\begin{tikzpicture}
  \coordinate (cooev1) at (0,0.38*\locpy) ;
  \coordinate (cooev2) at (0.9*\locpx,0.75*\locpy) ;
  \coordinate (cooc)   at (0,-0.1) ;
  \coordinate (coom1)  at (\locpx,-0.15) ;
  \coordinate (coom2)  at (0.95*\locpx,1.04*\locpy) ;
  \scopeArrow{0.55}{\arrowDefect}
  \draw[\colorC,line width=\widthObj,postaction={decorate}]
       (cooc) node[below]{$\dot\FF$} -- (cooev1) ;
  \end{scope}
  \draw[\colorC,line width=\widthObj,postaction={decorate}] (cooev1) [out=90,in=238] to (cooev2) ;
  \scopeArrow{0.22}{\arrowDefect} \scopeArrow{0.95}{\arrowDefect}
  \draw[\colorM,line width=\widthObj,postaction={decorate}]
       (coom1) node[below=1pt]{$m$} [out=100,in=267] to (cooev2) 
       [out=83,in=261] to (coom2) node[above=-1pt]{$m$} ;
  \end{scope} \end{scope}
  \node[line width=0.7*\widthboxline,fill=blue!10,draw=purple] at (cooev2)
       {$\scriptstyle(\ievF)_m$} ;
  \end{tikzpicture}
  }
  \qquad = \qquad~
  \raisebox{-6.2em} {\begin{tikzpicture}
  \coordinate (coodin) at (0,0.38*\locpy) ;
  \coordinate (cooev1) at (0,0.38*\locpy) ;
  \coordinate (cooev2) at (0.9*\locpx,0.75*\locpy) ;
  \coordinate (cooc)   at (0,-0.1) ;
  \coordinate (coom1)  at (\locpx,-0.15) ;
  \coordinate (coom2)  at (0.95*\locpx,1.04*\locpy) ;
  \scopeArrow{0.45}{\arrowDefect}
  \draw[\colorC,line width=\widthObj,postaction={decorate}]
       (cooc) node[below]{$\dot\FF$} -- (cooev1) ;
  \end{scope}
  \scopeArrow{0.55}{\arrowDefect}
  \draw[\colorC,line width=\widthObj,postaction={decorate}] (cooev1)
       [out=90,in=238] to node[sloped,midway,above=-4pt,xshift=-7pt] {$\BB^{m.m}$} (cooev2) ; 
  \end{scope}
  \scopeArrow{0.22}{\arrowDefect} \scopeArrow{0.95}{\arrowDefect}
  \draw[\colorM,line width=\widthObj,postaction={decorate}]
       (coom1) node[below=1pt]{$m$} [out=100,in=267] to (cooev2) 
       [out=83,in=261] to (coom2) node[above=-1pt]{$m$} ;
  \end{scope} \end{scope}
  \node[line width=0.7*\widthboxline,fill=blue!10,draw=blue] at (cooev2)
       {$\scriptstyle\iev_{m,m}$} ;
  \begin{scope}[shift={(coodin)},rotate=180] \Dinatmor \end{scope}
  \node[yshift=-7pt] at (coodin) {$\scriptstyle \widetilde\alpha^\FF_m$} ;
  \end{tikzpicture}
  }
  \label{eq:ev=evmm.iotam}
  \ee
  
Now in view of the expressions \eqref{eq:BB=iHom} and \eqref{FF=relNat} for the boundary and
bulk algebras, our proposal directly provides a distinguished morphism in the space
$\Hom_\C(\dot\FF,\BB^{m,m})$, for every $m \iN \M$, namely the structure morphism
  \be
  \dot\FF = \UC(\relNat(\Id_\M,\Id_\M)) \,\cong\, \UC\Big( \int_{\!m'\in\M} \iHom_\M(m',m')
  \Big)
  \rarr{\,\iota_m\,}\, \iHom(m,m) = \BB^{m,m} 
  \label{eq:iota2}
  \ee
of the end (which is naturally a morphism in \C).

The following considerations show that the morphism $\widetilde\alpha^\FF_m$ is indeed 
given by this structure morphism, i.e.\ that
  \be
  \widetilde\alpha^\FF_m = \iota_m \,.
  \ee
In fact, this boils down to the compatibility between the internal Hom adjunctions for the
module categories \M\ over \C\ and $\Funre_\C(\M,\M)$ over \ZC. In more detail,
first note that under the isomorphism
  \be
  \Hom_{\Funre_\C(\M,\M)}(\FF\act \Id,\Id) \rarr{~\cong~} \Hom_\ZC(\FF,\FF)
  \ee
that is the case $z \eq \FF$ of the internal Hom adjunction
$\Hom_{\Funre_\C(\M,\M)}(z\act \Id,\Id) \,{\rarr{\cong}}\, \Hom_\ZC(z,\FF)$
for the \ZC-module category $\Funre_\C(\M,\M)$, the module natural transformation 
$\ievF\colon \FF\act\Id \,{\to}\, \Id$ is mapped to the identity morphism $\id_\FF$ in \ZC.
Further, by the description of the bulk algebra $\FF$ as the end $\int_{\!m'\in\M} 
\iHom_\M(m',m')$, any morphism $f\colon z\,{\to}\,\FF$ in \ZC\ amounts to a dinatural family
$\iota_m \cir f\colon \dot z \,{\to}\, \iHom_\M(m,m)$, for $m \iN \M$, of morphisms in \C,
and thus in particular the morphism $\id_\FF$ in \ZC\ amounts to the family $\{\iota_m\}$
of structure morphisms of the end itself. Finally, according to the proof of Theorem 9
of \cite{fuSc25}, under the internal Hom 
adjunction $\Hom_\C(\dot\FF,\iHom_\M(m,m)) \,{\rarr{\cong}}\, \Hom_\M(\dot\FF\act m,m)$
of the \C-module category \M, the structure morphism $\iota_m$ is mapped to the 
$m$-component of the natural transformation $\ievF$. Put together, this means that we have
  \be
  (\ievF)_m = \iev_{m,m} \circ (\iota_m \act \id_m) \,,
  \ee
and thus indeed by identifying $\widetilde\alpha^\FF_m$ with $\iota_m$ we satisfy the
condition \eqref{eq:ievFm=iev_mm.(tildealpha.idm)} which fully characterizes 
$\widetilde\alpha^\FF_m$.

\medskip

Let us compare our result to what is known about the bulk-boundary OPE in the literature. To 
do so, we first recall from \eqref{eq:UC.GC} that when composed with the functor \GC\ from 
the enveloping category \CbC\ to the center \ZC\ (see \eqref{eq:CbC->ZC}), the forgetful 
functor $\UC$ amounts to taking a tensor product. Thus once again the description of 
the OPE is relative to the tensor product in \C. And again, in case \C\ is semisimple 
we can restrict our attention to chiral labels given by simple objects, i.e., in the 
situation at hand, to simple summands  isomorphic to $U_i^{}\ootimes U_j$ of the bulk
algebra $\FF$.  Denoting the unary OPE by the symbol $*_m$, this yields the formula
  \be
  *_m \! \big( \Phi^{\beta}_{U_i^{}\ootimes U_j} \big) = \sum_{k\in I_\C}
  \sum_{\alpha,\lambda} C^{m;\beta,\alpha}_{i,j;k,\lambda}\, \Psi^{m,m;\alpha}_{U_k}
  \label{eq:*m}
  \ee
for the bulk-boundary OPE, where the $\lambda$-summation is over a basis of the multiplicity
space $\Hom_\C(U_i\oti U_j,U_k)$. (In the literature (see e.g.\ \Cite{Eq.\,(10)}{prss3}), the
OPE is often written in a form like $\Phi^{\beta}_{U_i^{}\ootimes U_j} \,{\sim}\,
\sum_{k,\alpha,\lambda} C^{m;\beta,\alpha}_{i,j;k,\lambda}\, \Psi^{m,m;\alpha}_{U_k}$,
with the symbol $\sim$ reminding of the fact that the bulk field is imagined to approach the
boundary.) The coefficients $C$ in the OPE \eqref{eq:*m} have been introduced in \cite{cale} 
and have been studied extensively in the literature, such as in \cite{prss3,runk3,bppz2,pezu6}
and in Section 4.3 of \cite{fuRs10}. In the semisimple Cardy case the OPE coefficients are, up 
to twist eigenvalues, given by specific 6j-symbols \Cite{Sect.\,4.4}{fffs3}. The derivation 
above shows that these coefficients may be interpreted as encoding the dinatural structure 
morphism $\iota_m$ of the end $\FF \eq \relNat(\Id_\M,\Id_\M)$.

\medskip

One important application of the bulk-boundary OPE is the calculation of the one-bulk
field correlator on the disk, which is also known as a \emph{boundary state}. For 
the situation that the bulk field is incoming, the relevant space of conformal blocks is,
as a morphism space in \C, the space $\Hom_\C(\dot\FF,\one_\C)$ \cite{fgsS,fScS4}. We 
denote the boundary state for a disk with boundary condition 
$m \iN \M$ by $\bs_m \iN \Hom_\C(\dot\FF,\one_\C)$.
We obtain an explicit proposal for $\bs_m$ by considering the equality
\eqref{eq:ev=evmm.iotam} -- which 
refers to a local region of the world sheet (as befits an OPE) -- in the case
that the world sheet is a disk without any further field insertions. In this situation the 
TFT construction of \cite{fuRs4,fuRs10} suggests to describe $\bs_m$ as
     \def\locpx  {1.5}
     \def\locpy  {4.1}
     \def\widthDinatmor {0.33}
  \be
  \bs_m ~ = \qquad
  \raisebox{-7.6em} {\begin{tikzpicture}
  \coordinate (coodin) at (0,0.38*\locpy) ;
  \coordinate (cooev1) at (0,0.38*\locpy) ;
  \coordinate (cooev2) at (0.94*\locpx,0.81*\locpy) ;
  \coordinate (cooc)   at (0,0) ;
  \coordinate (coom1)  at (\locpx,0.4*\locpy) ;
  \coordinate (coom2)  at (\locpx,\locpy) ;
  \coordinate (coom3)  at (2*\locpx,\locpy) ;
  \coordinate (coom4)  at (2*\locpx,0.4*\locpy) ;
  \scopeArrow{0.45}{\arrowDefect}
  \draw[\colorC,line width=\widthObj,postaction={decorate}]
       (cooc) node[below=-1pt]{$\dot\FF$} -- (cooev1) ;
  \end{scope}
  \scopeArrow{0.55}{\arrowDefect}
  \draw[\colorC,line width=\widthObj,postaction={decorate}] (cooev1)
       [out=90,in=238] to node[sloped,midway,above=-4pt,xshift=-7pt] {$\BB^{m.m}$} (cooev2) ; 
  \end{scope}
  \scopeArrow{0.11}{\arrowDefect}
  \draw[\colorM,line width=\widthObj,postaction={decorate}]
       (coom1) -- node[near start,sloped,yshift=-6pt] {$m$} (coom2) ;
  \draw[\colorM,line width=\widthObj,postaction={decorate}] (coom3) -- (coom4) ;
  \end{scope}
  \draw[\colorM,line width=\widthObj]
       (coom2) arc (180:0:0.5*\locpx) (coom4) arc (360:180:0.5*\locpx) ;
  \node[line width=0.7*\widthboxline,fill=blue!10,draw=blue] at (cooev2)
       {$\scriptstyle\iev_{m,m}$} ;
  \begin{scope}[shift={(coodin)},rotate=180] \Dinatmor \end{scope}
  \node[yshift=-5pt] at (coodin) {$\scriptstyle \iota_m$} ;
  \end{tikzpicture}
  }
  \label{eq:pic:bstate}
  \ee
We have indeed sufficient algebraic information to interpret the picture \eqref{eq:pic:bstate} 
as a morphism in \C: Recalling that the multiplication $\imu{}_{m,m,m}^{}$ is the image of 
a morphism of type \eqref{eq:tworems1} under the internal Hom adjunction, we see that
\eqref{eq:pic:bstate} should be given by the composition
  \be
  \dot\FF \rarr{\iota_m} \BB^{m,m}
  \rarr{\id_{\BB^{m,m}} \otimes \underline\eta{}_m^{}} \BB^{m,m} \oti \BB^{m,m}
  \rarr{\imu{}_{m,m,m}^{}} \BB^{m,m} \rarr{\underline\varepsilon{}_m^{}} \id_\C
  \ee
with $\underline\eta{}_m^{} \colon \one_\C \,{\to}\, \iHom(m,\one_\C\Act m) \eq \iHom(m,m)$
the unit of the internal Hom adjunction and
$\underline\varepsilon{}_m^{}$ the counit \eqref{eq:iHom-counit}, and thus
  \be
  \bs_m = \underline\varepsilon{}_m^{} \circ \iota_m \,.
  \label{eq:def-bs}
  \ee
Note that owing to triviality of the relative Serre functor of \M, the counit 
$\underline\varepsilon{}_m^{}$ (which exists because $\BB^{m,m}$ is Frobenius)
coincides with the internal trace $\underline{\,\tr}_m$.
 
We leave a thorough investigation of the proposal \eqref{eq:def-bs} to future work.
Here we only observe that for semisimple \C\ it amounts to the following suggestive
description. Realizing \M\ as the category $\Amod$ of right modules over
a special symmetric Frobenius algebra $A$, according to \eqref{eq:iHom-otA} the boundary 
algebra $\BB^{m,m}_{}$ can be expressed as $m \,{\otimes_A}\, \Vee m$. Further, invoking
Lemma 3.8(a) of \cite{shimi20}
one sees that the internal trace then reduces to $\underline{\tr}_m \eq \tilde\ev_m \cir P_{\!A}$
with $P_{\!A}$ the projector that realizes the tensor product over $A$.\,%
 \footnote{~For the explicit form of $P_{\!A}$ see e.g.\ \Cite{Eq.\,(5.127)}{fuRs4},
 or \Cite{Lemma\,1.21}{kios} in the case of commutative algebras.}
It follows that 
     \def\locpx  {1.5}
     \def\locpy  {4.2}
     \def\widthDinatmor {0.33}
  \be
  \bs_m ~ = \qquad
  \raisebox{-6.6em} {\begin{tikzpicture}
  \coordinate (cooc)   at (0,0) ;
  \coordinate (coodin) at (0,0.41*\locpy) ;
  \coordinate (cooev1) at (0,0.41*\locpy) ;
  \coordinate (coom0)  at (-0.47*\widthDinatmor,0.4*\locpy) ;
  \coordinate (coom1)  at (0.96*\locpx,0.4*\locpy) ;
  \coordinate (coom2)  at (1.3*\locpx,0.9*\locpy+0.5*\locpx) ;
  \coordinate (coom3)  at (1.8*\locpx,0.9*\locpy) ;
  \coordinate (coom4)  at (1.8*\locpx,0.6*\locpy) ;
  \coordinate (cooprj) at (0,0.54*\locpy) ;
  \scopeArrow{0.45}{\arrowDefect}
  \draw[\colorC,line width=\widthObj,postaction={decorate}]
       (cooc) node[below=-1pt]{$\dot\FF$} -- (cooev1) ;
  \end{scope}
  \scopeArrow{0.48}{\arrowDefect}
  \draw[\colorM,line width=\widthObj,postaction={decorate}] (coom3) -- (coom4) ;
  \draw[\colorM,line width=\widthObj,postaction={decorate}] (coom0) -- ++(0,0.4*\locpx)
       [out=90,in=180] to node[midway,sloped,yshift=6pt] {$m$} (coom2) ;
  \end{scope}
  \draw[\colorM,line width=\widthObj] (coom2) arc (90:0:0.5*\locpx)
       (coom4) -- ++(0,-0.2*\locpy) arc (360:180:0.42*\locpx)
       (coom1) -- ++(0,0.2*\locpy) arc (0:180:0.42*\locpx) -- +(0,-0.2*\locpy) ;
  \begin{scope}[shift={(coodin)},rotate=180] \Dinatmor \end{scope}
  \node[yshift=-5pt] at (coodin) {$\scriptstyle \iota_m$} ;
  \node[line width=0.5*\widthboxline,fill=gray!40,draw=black!60,inner sep=2pt]
       at (cooprj) {$\scriptstyle P_{\!A}$} ;
  \end{tikzpicture}
  }
  ~\qquad = \qquad
  \raisebox{-6.6em} {\begin{tikzpicture}
  \coordinate (cooc)   at (0,0) ;
  \coordinate (coodin) at (0,0.41*\locpy) ;
  \coordinate (coom0)  at (-0.47*\widthDinatmor,0.4*\locpy) ;
  \coordinate (coom1)  at (0.47*\widthDinatmor,0.4*\locpy) ;
  \coordinate (coom2)  at (-0.92,0.8*\locpy) ;
  \coordinate (coom3)  at (0.92,0.8*\locpy) ;
  \coordinate (cooprj) at (0,0.54*\locpy) ;
  \scopeArrow{0.45}{\arrowDefect}
  \draw[\colorC,line width=\widthObj,postaction={decorate}]
       (cooc) node[below=-1pt]{$\dot\FF$} -- (cooev1) ;
  \end{scope}
  \scopeArrow{0.32}{\arrowDefect}
  \draw[\colorM,line width=\widthObj,postaction={decorate}]
       (coom0) -- ++(0,0.5*\locpx) [out=90,in=270] to
       node[near end,sloped,yshift=6pt] {$m$} (coom2) arc (180:0:0.92) ;
  \end{scope}
  \draw[\colorM,line width=\widthObj]
       (coom1) -- ++(0,0.5*\locpx) [out=90,in=270] to (coom3) ;
  \begin{scope}[shift={(coodin)},rotate=180] \Dinatmor \end{scope}
  \node[yshift=-5pt] at (coodin) {$\scriptstyle \iota_m$} ;
  \node[line width=0.5*\widthboxline,fill=gray!40,draw=black!60,inner sep=2pt]
       at (cooprj) {$\scriptstyle P_{\!A}$} ;
  \end{tikzpicture}
  }
  ~\qquad = \qquad
  \raisebox{-6.6em} {\begin{tikzpicture}
  \coordinate (cooc)   at (0,0) ;
  \coordinate (coodin) at (0,0.41*\locpy) ;
  \coordinate (coom0)  at (-0.47*\widthDinatmor,0.4*\locpy) ;
  \coordinate (coom1)  at (0.47*\widthDinatmor,0.4*\locpy) ;
  \coordinate (coom2)  at (-0.92,0.8*\locpy) ;
  \coordinate (coom3)  at (0.92,0.8*\locpy) ;
  \scopeArrow{0.45}{\arrowDefect}
  \draw[\colorC,line width=\widthObj,postaction={decorate}]
       (cooc) node[below=-1pt]{$\dot\FF$} -- (cooev1) ;
  \end{scope}
  \scopeArrow{0.28}{\arrowDefect}
  \draw[\colorM,line width=\widthObj,postaction={decorate}]
       (coom0) -- ++(0,0.2*\locpx) [out=90,in=270] to
       node[near end,sloped,yshift=6pt] {$m$} (coom2) arc (180:0:0.92) ;
  \end{scope}
  \draw[\colorM,line width=\widthObj]
       (coom1) -- ++(0,0.2*\locpx) [out=90,in=270] to (coom3) ;
  \begin{scope}[shift={(coodin)},rotate=180] \Dinatmor \end{scope}
  \node[yshift=-5pt] at (coodin) {$\scriptstyle \iota_m$} ;
  \end{tikzpicture}
  }
  \label{eq:pic:bstate-ssi}
  \ee
Here the last equality follows from the specialness of $A$, analogously as e.g.\ 
in \Cite{Rem.\,4.1}{fuRs10}.

We expect that in the Cardy case the right hand side of 
\eqref{eq:pic:bstate-ssi} can be expressed as a partial trace over the canonical
representation morphism $\rho^{\dot\FF_\text{Cardy}}_m$ that for any $m \iN \C$ is
obtained from the double braiding of $c \iN \C$ with $m$ \Cite{Eq.\,(2.49)}{fgsS}. Thereby
$\bs^{\dot\FF_\text{Cardy}}_m$ is interpreted as the \emph{character} of $m$ as 
a $\dot\FF_\text{Cardy}$-module, in agreement with the results of Section 3.2 of \cite{fgsS}.
We also expect that, just like for the Cardy case \cite{fgsS}, the interpretation as a
character survives also for general \M\ beyond semisimplicity.


\subsection{Consistency check: Bulk-boundary compatibility conditions} \label{sec:9e9d}

Next we recall from the Introduction that the bulk-boundary OPE is required to satisfy the
compatibility conditions \eqref{eq:relation9d} and \eqref{eq:relation9e} with the bulk
and boundary algebras. 

Let us first consider the equality \eqref{eq:relation9e}. Algebraically it amounts to the 
statement that for every $m \iN \M$ the structure morphism $\iota_m\colon \dot\FF \,{\rarr~}\,
\BB^{m,m}$ of the end \eqref{eq:iota2} is a \emph{morphism of al\-ge\-bras} in \C, i.e.\ that
it satisfies
  \be
  \imu{}_{m,m,m} \circ (\iota_m \oti \iota_m) = \iota_m \circ \imu \,,
  \label{eq:alg9e}
  \ee
with $\imu{}_{m,m,m}$ the product of the boundary algebra $\BB^{m,m}$ and $\imu$ the product
of the bulk algebra $\FF$. The equality \eqref{eq:alg9e} is indeed satisfied -- as shown in
Proposition 11 of \cite{fuSc25}, it is nothing but the description of the canonical
multiplication on the end $\FF$ in terms of the dinatural family $\{\iota_m\}_{m\in\M}$.

\medskip

Next we analyze the condition \eqref{eq:relation9d}, which states the equality of two
factorizations of the correlator of a disk with one bulk and two boundary insertions. To
understand its algebraic content, we first rephrase this global description as the local
statement that moving a bulk insertion $\FF$ close to the boundary at a location which is
on one side of a boundary field insertion $\BB^{n.m}$ is the same as moving $\FF$ close 
to the boundary on the other side of the $\BB^{n.m}$-insertion. We can then invoke the
equality \eqref{eq:ev=evmm.iotam} to visualize the first of the two situations as
     \def\locpx  {2.8}
     \def\locpy  {6.3}
     \def\widthDinatmor {0.33}
  \be
  \raisebox{-8.4em} {\begin{tikzpicture}
  \coordinate (cooev1) at (0,0.39*\locpy) ;
  \coordinate (cooev2) at (0.9*\locpx,0.82*\locpy) ;
  \coordinate (cooev3) at (0.89*\locpx,0.46*\locpy) ;
  \coordinate (coob)   at (0.45*\locpx,0) ;
  \coordinate (coob2)  at (0.45*\locpx,0.19*\locpy) ;
  \coordinate (cooc)   at (0,0) ;
  \coordinate (coom1)  at (\locpx,0) ;
  \coordinate (coom2)  at (0.95*\locpx,\locpy) ;
  \scopeArrow{0.75}{\arrowDefect}
  \draw[\colorC,line width=\widthObj,postaction={decorate}]
       (cooc) node[below=-1pt]{$\dot\FF$} -- (cooev1) ;
  \end{scope}
  \scopeArrow{0.33}{\arrowDefect}
  \draw[\colorC,line width=\widthObj,postaction={decorate}]
       (coob) node[below]{$\BB^{n.m}$} -- (coob2) [out=90,in=238] to (cooev3) ;
  \end{scope}
  \draw[\colorC,line width=\widthObj] (cooev1) [out=90,in=238] to (cooev2) ; 
  \scopeArrow{0.21}{\arrowDefect} \scopeArrow{0.65}{\arrowDefect} \scopeArrow{0.96}{\arrowDefect}
  \draw[\colorM,line width=\widthObj,postaction={decorate}] (coom1)
       node[below=1pt]{$m$} [out=100,in=267] to node[sloped,near end,above=-2pt,xshift=14pt] {$n$}
       (cooev2) [out=83,in=261] to (coom2) node[above=-1pt]{$n$} ;
  \end{scope} \end{scope} \end{scope}
  \node[line width=0.7*\widthboxline,fill=blue!10,draw=purple] at (cooev2)
       {$\scriptstyle(\ievF)_n$} ;
  \node[line width=0.7*\widthboxline,fill=blue!10,draw=blue] at (cooev3)
       {$\scriptstyle\iev_{m,n}$} ;
  \end{tikzpicture}
  }
  \qquad = \qquad~
  \raisebox{-8.4em} {\begin{tikzpicture}
  \coordinate (coodin) at (0,0.39*\locpy) ;
  \coordinate (cooev1) at (0,0.39*\locpy) ;
  \coordinate (cooev2) at (0.9*\locpx,0.82*\locpy) ;
  \coordinate (cooev3) at (0.89*\locpx,0.46*\locpy) ;
  \coordinate (coob)   at (0.45*\locpx,0) ;
  \coordinate (coob2)  at (0.45*\locpx,0.19*\locpy) ;
  \coordinate (cooc)   at (0,0) ;
  \coordinate (coom1)  at (\locpx,0) ;
  \coordinate (coom2)  at (0.95*\locpx,\locpy) ;
  \scopeArrow{0.45}{\arrowDefect}
  \draw[\colorC,line width=\widthObj,postaction={decorate}]
       (cooc) node[below=-1pt]{$\dot\FF$} -- (cooev1) ;
  \end{scope}
  \scopeArrow{0.33}{\arrowDefect}
  \draw[\colorC,line width=\widthObj,postaction={decorate}]
       (coob) node[below]{$\BB^{n.m}$} -- (coob2) [out=90,in=238] to (cooev3) ;
  \end{scope}
  \scopeArrow{0.55}{\arrowDefect}
  \draw[\colorC,line width=\widthObj,postaction={decorate}] (cooev1)
       [out=90,in=238] to node[sloped,midway,above=-2pt,xshift=-5pt] {$\BB^{n,n}$} (cooev2) ; 
  \end{scope}
  \scopeArrow{0.21}{\arrowDefect} \scopeArrow{0.65}{\arrowDefect} \scopeArrow{0.96}{\arrowDefect}
  \draw[\colorM,line width=\widthObj,postaction={decorate}] (coom1)
       node[below=1pt]{$m$} [out=100,in=267] to node[sloped,near end,above=-2pt,xshift=14pt] {$n$}
       (cooev2) [out=83,in=261] to (coom2) node[above=-1pt]{$n$} ;
  \end{scope} \end{scope} \end{scope}
  \node[line width=0.7*\widthboxline,fill=blue!10,draw=blue] at (cooev2)
       {$\scriptstyle\iev_{n,n}$} ;
  \node[line width=0.7*\widthboxline,fill=blue!10,draw=blue] at (cooev3)
       {$\scriptstyle\iev_{m,n}$} ;
  \begin{scope}[shift={(coodin)},rotate=180] \Dinatmor \end{scope}
  \node[yshift=-5pt] at (coodin) {$\scriptstyle \iota_n$} ;
  \end{tikzpicture}
  }
  \qquad = \qquad
  \raisebox{-8.4em} {\begin{tikzpicture}
  \coordinate (coodin) at (0,0.22*\locpy) ;
  \coordinate (cooev1) at (0,0.22*\locpy) ;
  \coordinate (cooev2) at (0.9*\locpx,0.82*\locpy) ;
  \coordinate (coomu)  at (0.24*\locpx,0.53*\locpy) ;
  \coordinate (coob)   at (0.48*\locpx,0) ;
  \coordinate (coob2)  at (0.48*\locpx,0.19*\locpy) ;
  \coordinate (cooc)   at (0,0) ;
  \coordinate (coom1)  at (\locpx,0) ;
  \coordinate (coom2)  at (0.95*\locpx,\locpy) ;
  \scopeArrow{0.45}{\arrowDefect}
  \draw[\colorC,line width=\widthObj,postaction={decorate}]
       (cooc) node[below=-1pt]{$\dot\FF$} -- (cooev1) ;
  \end{scope}
  \scopeArrow{0.33}{\arrowDefect}
  \draw[\colorC,line width=\widthObj,postaction={decorate}]
       (coob) node[below]{$\BB^{n.m}$} -- (coob2) [out=90,in=302] to (coomu) ;
  \end{scope}
  \scopeArrow{0.55}{\arrowDefect}
  \draw[\colorC,line width=\widthObj,postaction={decorate}] (cooev1)
       [out=90,in=238] to node[sloped,midway,above=-2pt,xshift=-6pt] {$\BB^{n,n}$} (coomu) ;
  \draw[\colorC,line width=\widthObj,postaction={decorate}] (coomu)
       [out=90,in=238] to node[sloped,midway,above=-2pt,xshift=-4pt] {$\BB^{n,m}$} (cooev2) ; 
  \end{scope}
  \scopeArrow{0.31}{\arrowDefect} \scopeArrow{0.95}{\arrowDefect}
  \draw[\colorM,line width=\widthObj,postaction={decorate}]
       (coom1) node[below=1pt]{$m$} [out=100,in=267] to (cooev2) 
       [out=83,in=261] to (coom2) node[above=-1pt]{$n$} ;
  \end{scope} \end{scope}
  \node[line width=0.7*\widthboxline,fill=blue!10,draw=blue] at (cooev2)
       {$\scriptstyle\iev_{n,n}$} ;
  \node[rounded corners,line width=\widthboxline,fill=blue!10,draw=\colorBdyfield]
       at (coomu) {$\scriptstyle\imu{}^{\phantom m}_{n,n,m}$} ;
  \begin{scope}[shift={(coodin)},rotate=180] \Dinatmor \end{scope}
  \node[yshift=-5pt] at (coodin) {$\scriptstyle \iota_n$} ;
  \end{tikzpicture}
  }
  \label{eq:alg9d-1}
  \ee
where  in the second equality
we use the relation between the boundary multiplication and evaluation (see
\eqref{eq:tworems1}). In order to visualize the other side of the sewing constraint 
\eqref{eq:relation9d} we must find a morphism $\tau \iN \Hom_\C
(\FF \oti \BB^{n,m},\BB^{n,m} \oti \FF)$ that allows us to make sense out of the picture
     \def\locpx  {2.8}
     \def\locpy  {6.3}
     \def\widthDinatmor {0.33}
  \be
  \raisebox{-8.4em} {\begin{tikzpicture}
  \coordinate (cooev)  at (0.46*\locpx,0.40*\locpy) ;
  \coordinate (cooev2) at (0.93*\locpx,0.75*\locpy) ;
  \coordinate (cooev3) at (0.88*\locpx,0.72*\locpy) ;
  \coordinate (coob)   at (0.46*\locpx,0.1*\locpy) ;
  \coordinate (coob2)  at (0.46*\locpx,0.40*\locpy) ;
  \coordinate (cooc)   at (0,0.1*\locpy) ;
  \coordinate (cooc2)  at (0,0.4*\locpy) ;
  \coordinate (coom1)  at (\locpx,0.1*\locpy) ;
  \coordinate (coom2)  at (\locpx,0.9*\locpy) ;
  \scopeArrow{0.50}{\arrowDefect}
  \draw[\colorC,line width=\widthObj,postaction={decorate}]
       (cooc) node[below=-1pt]{$\dot\FF$} -- (cooc2) ; 
  \end{scope}
  \scopeArrow{0.23}{\arrowDefect} \scopeArrow{0.77}{\arrowDefect}
  \draw[\colorC,line width=\widthObj,postaction={decorate}]
       (coob) node[below]{$\BB^{n.m}$} -- (coob2) [out=90,in=238] to
       node[sloped,midway,above=-2pt,xshift=-3pt] {$\BB^{n,m}$} (cooev3) ;
  \end{scope} \end{scope}
  \scopeArrow{0.20}{\arrowDefect} \scopeArrow{0.61}{\arrowDefect} \scopeArrow{0.95}{\arrowDefect}
  \draw[\colorM,line width=\widthObj,postaction={decorate}]
       (coom1) node[below=1pt]{$m$} [out=100,in=267] to (cooev2) 
       [out=83,in=261] to (coom2) node[above=-1pt]{$n$} ;
  \end{scope} \end{scope} \end{scope}
  \node[line width=0.7*\widthboxline,fill=blue!10,draw=purple] at (cooev)
       {~~~~~~$\scriptstyle (\ievF)_m \,\circ\, \tau$~~~~~~} ;
  \node[line width=0.7*\widthboxline,fill=blue!10,draw=blue] at (cooev3)
       {$\scriptstyle\iev_{m,n}$} ;
  \end{tikzpicture}
  }
  \label{eq:tau}
  \ee
Now recall that the bulk algebra $\FF$ is by definition an object in the Drinfeld center, 
so it has a distinguished half-braiding $\gamma$; the inverse of
the $\BB^{n,m}$-component of this half-braiding, which we will denote by $\gamma_{n,m}$,
exactly serves our purpose. One might
have thought that, since $\C$ is braided, one could take instead the braiding of $\dot\FF$
and $\BB^{n,m}$. But the inverse braiding would be equally qualified, and none of these
two morphisms is preferred by the structure of the bulk algebra, in contrast to the 
half-braiding as an object in the Drinfeld center. Thus we set $\tau \eq \gamma_{n,m}^{-1}$
in \eqref{eq:tau}, whereby the second situation is described as
     \def\locpx  {2.8}
     \def\locpy  {6.3}
     \def\widthDinatmor {0.33}
  \be
  \raisebox{-8.4em} {\begin{tikzpicture}
  \coordinate (cooev1) at (0,0.39*\locpy) ;
  \coordinate (cooev2) at (0.9*\locpx,0.81*\locpy) ;
  \coordinate (cooev3) at (0.89*\locpx,0.52*\locpy) ;
  \coordinate (coob)   at (0.45*\locpx,0) ;
  \coordinate (coob2)  at (0.45*\locpx,0.30*\locpy) ;
  \coordinate (cooc)   at (0,0) ;
  \coordinate (coogam) at (0.23*\locpx,0.21*\locpy) ;
  \coordinate (coom1)  at (\locpx,0) ;
  \coordinate (coom2)  at (0.95*\locpx,\locpy) ;
  \scopeArrow{0.11}{\arrowDefect} \scopeArrow{0.71}{\arrowDefect}
  \draw[\colorC,line width=\widthObj,postaction={decorate}]
       (cooc) node[below=-1pt]{$\dot\FF$} -- (cooev1)
       [out=90,in=238] to node[sloped,midway,above=-2pt,xshift=-3pt] {$\BB^{n,m}$} (cooev2) ;
  \end{scope} \end{scope}
  \scopeArrow{0.18}{\arrowDefect} \scopeArrow{0.78}{\arrowDefect}
  \draw[\colorC,line width=\widthObj,postaction={decorate}]
       (coob) node[below]{$\BB^{n.m}$} -- (coob2) node[left=-4pt,yshift=10pt] {$\dot\FF$}
       [out=90,in=238] to (cooev3) ; 
  \end{scope} \end{scope}
  \scopeArrow{0.21}{\arrowDefect} \scopeArrow{0.67}{\arrowDefect} \scopeArrow{0.96}{\arrowDefect}
  \draw[\colorM,line width=\widthObj,postaction={decorate}]
       (coom1) node[below=1pt]{$m$} [out=100,in=267] to node[sloped,very near end,below=-2pt,xshift=1pt] {$m$}
       (cooev2) [out=83,in=261] to (coom2) node[above=-1pt]{$n$} ;
  \end{scope} \end{scope} \end{scope}
  \node[rounded corners,line width=\widthboxline,fill=green!15,draw=green!70!black]
       at (coogam) {~~~$\scriptstyle\gamma_{n,m}^{-1}$~~~~} ;
  \node[line width=0.7*\widthboxline,fill=blue!10,draw=blue] at (cooev2)
       {$\scriptstyle\iev_{m,n}$} ;
  \node[line width=0.7*\widthboxline,fill=blue!10,draw=purple] at (cooev3)
       {$\scriptstyle(\ievF)_m$} ;
  \end{tikzpicture}
  }
  \qquad = \qquad~
  \raisebox{-8.4em} {\begin{tikzpicture}
  \coordinate (coodin) at (0.45*\locpx,0.39*\locpy) ;
  \coordinate (cooev1) at (0,0.39*\locpy) ;
  \coordinate (cooev2) at (0.9*\locpx,0.82*\locpy) ;
  \coordinate (cooev3) at (0.89*\locpx,0.59*\locpy) ;
  \coordinate (coob)   at (0.45*\locpx,0) ;
  \coordinate (coob2)  at (0.45*\locpx,0.36*\locpy) ;
  \coordinate (cooc)   at (0,0) ;
  \coordinate (coogam) at (0.23*\locpx,0.17*\locpy) ;
  \coordinate (coom1)  at (\locpx,0) ;
  \coordinate (coom2)  at (0.95*\locpx,\locpy) ;
  \scopeArrow{0.09}{\arrowDefect} \scopeArrow{0.71}{\arrowDefect}
  \draw[\colorC,line width=\widthObj,postaction={decorate}]
       (cooc) node[below=-1pt]{$\dot\FF$} -- (cooev1)
       [out=90,in=238] to node[sloped,midway,above=-2pt,xshift=-3pt] {$\BB^{n,m}$} (cooev2) ;
  \end{scope} \end{scope}
  \scopeArrow{0.14}{\arrowDefect} \scopeArrow{0.80}{\arrowDefect}
  \draw[\colorC,line width=\widthObj,postaction={decorate}]
       (coob) node[below]{$\BB^{n.m}$} -- node[near end,right=-2pt] {$\dot\FF$} (coob2)
       [out=90,in=238] to node[sloped,midway,above=-3pt,xshift=-6pt] {$\BB^{m,m}$} (cooev3) ; 
  \end{scope} \end{scope}
  \scopeArrow{0.21}{\arrowDefect} \scopeArrow{0.72}{\arrowDefect} \scopeArrow{0.96}{\arrowDefect}
  \draw[\colorM,line width=\widthObj,postaction={decorate}]
       (coom1) node[below=1pt]{$m$} [out=100,in=267] to node[sloped,very near end,below=-2pt,xshift=5pt] {$m$}
       (cooev2) [out=83,in=261] to (coom2) node[above=-1pt]{$n$} ;
  \end{scope} \end{scope} \end{scope}
  \node[rounded corners,line width=\widthboxline,fill=green!15,draw=green!70!black]
       at (coogam) {~~~$\scriptstyle\gamma_{n,m}^{-1}$~~~~} ;
  \node[line width=0.7*\widthboxline,fill=blue!10,draw=blue] at (cooev2)
       {$\scriptstyle\iev_{m,n}$} ;
  \node[line width=0.7*\widthboxline,fill=blue!10,draw=blue] at (cooev3)
       {$\scriptstyle\iev_{m,m}$} ;
  \begin{scope}[shift={(coodin)},rotate=180] \Dinatmor \end{scope}
  \node[yshift=-5pt] at (coodin) {$\scriptstyle \iota_m$} ;
  \end{tikzpicture}
  }
  \qquad = \qquad
  \raisebox{-8.4em} {\begin{tikzpicture}
  \coordinate (coodin) at (0.45*\locpx,0.39*\locpy) ;
  \coordinate (cooev1) at (0,0.22*\locpy) ;
  \coordinate (cooev2) at (0.9*\locpx,0.84*\locpy) ;
  \coordinate (coomu)  at (0.24*\locpx,0.59*\locpy) ;
  \coordinate (coob)   at (0.48*\locpx,0) ;
  \coordinate (coob2)  at (0.45*\locpx,0.36*\locpy) ;
  \coordinate (cooc)   at (0,0) ;
  \coordinate (coogam) at (0.23*\locpx,0.17*\locpy) ;
  \coordinate (coom1)  at (\locpx,0) ;
  \coordinate (coom2)  at (0.95*\locpx,\locpy) ;
  \scopeArrow{0.41}{\arrowDefect}
  \draw[\colorC,line width=\widthObj,postaction={decorate}]
       (cooc) node[below=-1pt] {$\dot\FF$} -- (cooev1) ;
  \end{scope}
  \scopeArrow{0.15}{\arrowDefect} \scopeArrow{0.80}{\arrowDefect}
  \draw[\colorC,line width=\widthObj,postaction={decorate}]
       (coob) node[below]{$\BB^{n.m}$} -- node[near end,right=-2pt,yshift=2pt] {$\dot\FF$}
       (coob2) [out=90,in=302] to node[midway,sloped,below=-16pt,xshift=5pt] {$\BB^{m,m}$} (coomu) ;
  \end{scope} \end{scope}
  \scopeArrow{0.56}{\arrowDefect}
  \draw[\colorC,line width=\widthObj,postaction={decorate}] (cooev1)
       [out=90,in=238] to node[sloped,midway,above=-2pt,xshift=-5pt] {$\BB^{n,m}$} (coomu) ;
  \draw[\colorC,line width=\widthObj,postaction={decorate}] (coomu)
       [out=90,in=238] to node[sloped,midway,above=-2pt,xshift=-2pt] {$\BB^{n,m}$} (cooev2) ; 
  \end{scope}
  \scopeArrow{0.31}{\arrowDefect} \scopeArrow{0.97}{\arrowDefect}
  \draw[\colorM,line width=\widthObj,postaction={decorate}]
       (coom1) node[below=1pt]{$m$} [out=100,in=267] to (cooev2) 
       [out=83,in=261] to (coom2) node[above=-1pt]{$n$} ;
  \end{scope} \end{scope}
  \node[line width=0.7*\widthboxline,fill=blue!10,draw=blue] at (cooev2)
       {$\scriptstyle\iev_{m,n}$} ;
  \node[rounded corners,line width=\widthboxline,fill=green!15,draw=green!70!black]
       at (coogam) {~~~$\scriptstyle\gamma_{n,m}^{-1}$~~~~} ;
  \node[rounded corners,line width=\widthboxline,fill=blue!10,draw=\colorBdyfield]
       at (coomu) {$\scriptstyle\imu{}^{\phantom m}_{n,m,m}$} ;
  \begin{scope}[shift={(coodin)},rotate=180] \Dinatmor \end{scope}
  \node[yshift=-5pt] at (coodin) {$\scriptstyle \iota_m$} ;
  \end{tikzpicture}
  }
  \label{eq:alg9d-2}
  \ee

Algebraically, the condition \eqref{eq:relation9d} can now be stated as the
equality of the right hand sides of \eqref{eq:alg9d-1} and \eqref{eq:alg9d-2} for all
$m,n \iN \M$ or, equivalently, as 
     \def\locpx  {3.2}
     \def\locpy  {4.9}
     \def\widthDinatmor {0.33}
  \be
  \raisebox{-7.4em} {\begin{tikzpicture}
  \coordinate (coodin) at (0,0.36*\locpy) ;
  \coordinate (cooev1) at (0,0.36*\locpy) ;
  \coordinate (cooev2) at (0.24*\locpx,\locpy) ;
  \coordinate (coomu)  at (0.24*\locpx,0.69*\locpy) ;
  \coordinate (coob)   at (0.48*\locpx,0) ;
  \coordinate (coob2)  at (0.48*\locpx,0.19*\locpy) ;
  \coordinate (cooc)   at (0,0) ;
  \coordinate (coom1)  at (\locpx,0) ;
  \coordinate (coom2)  at (0.95*\locpx,\locpy) ;
  \scopeArrow{0.45}{\arrowDefect}
  \draw[\colorC,line width=\widthObj,postaction={decorate}]
       (cooc) node[below=-1pt]{$\dot\FF$} -- (cooev1) ;
  \end{scope}
  \scopeArrow{0.26}{\arrowDefect}
  \draw[\colorC,line width=\widthObj,postaction={decorate}]
       (coob) node[below]{$\BB^{n.m}$} -- (coob2) [out=90,in=302] to (coomu) ;
  \end{scope}
  \scopeArrow{0.55}{\arrowDefect}
  \draw[\colorC,line width=\widthObj,postaction={decorate}] (cooev1)
       [out=90,in=238] to node[sloped,midway,above=-2pt,xshift=-6pt] {$\BB^{n,n}$} (coomu) ;
  \end{scope}
  \scopeArrow{0.75}{\arrowDefect}
  \draw[\colorC,line width=\widthObj,postaction={decorate}]
       (coomu) -- (cooev2) node[above=-2pt] {$\BB^{n,m}$} (cooev2) ;
  \end{scope}
  \node[rounded corners,line width=\widthboxline,fill=blue!10,draw=\colorBdyfield]
       at (coomu) {$\scriptstyle\imu{}^{\phantom m}_{n,n,m}$} ;
  \begin{scope}[shift={(coodin)},rotate=180] \Dinatmor \end{scope}
  \node[yshift=-5pt] at (coodin) {$\scriptstyle \iota_n$} ;
  \end{tikzpicture}
  }
  ~\quad = \qquad~
  \raisebox{-7.4em} {\begin{tikzpicture}
  \coordinate (coodin) at (0.45*\locpx,0.49*\locpy) ;
  \coordinate (cooev1) at (0,0.29*\locpy) ;
  \coordinate (cooev2) at (0.24*\locpx,\locpy) ;
  \coordinate (coomu)  at (0.24*\locpx,0.77*\locpy) ;
  \coordinate (coob)   at (0.48*\locpx,0) ;
  \coordinate (coob2)  at (0.45*\locpx,0.47*\locpy) ;
  \coordinate (cooc)   at (0,0) ;
  \coordinate (coogam) at (0.23*\locpx,0.22*\locpy) ;
  \coordinate (coom1)  at (\locpx,0) ;
  \coordinate (coom2)  at (0.95*\locpx,\locpy) ;
  \scopeArrow{0.14}{\arrowDefect} \scopeArrow{0.71}{\arrowDefect}
  \draw[\colorC,line width=\widthObj,postaction={decorate}]
       (cooc) node[below=-1pt] {$\dot\FF$} -- (cooev1)
       [out=90,in=238] to node[sloped,midway,above=-1pt,xshift=-4pt] {$\BB^{n,m}$} (coomu) ;
  \end{scope} \end{scope}
  \scopeArrow{0.14}{\arrowDefect} \scopeArrow{0.80}{\arrowDefect}
  \draw[\colorC,line width=\widthObj,postaction={decorate}]
       (coob) node[below]{$\BB^{n.m}$} -- node[near end,right=-2pt] {$\dot\FF$}
       (coob2) [out=90,in=302] to node[midway,sloped,below=-16pt,xshift=6pt] {$\BB^{m,m}$} (coomu) ;
  \end{scope} \end{scope}
  \scopeArrow{0.81}{\arrowDefect}
  \draw[\colorC,line width=\widthObj,postaction={decorate}]
       (coomu) -- (cooev2) node[above=-2pt] {$\BB^{n,m}$} ;
  \end{scope}
  \node[rounded corners,line width=\widthboxline,fill=green!15,draw=green!70!black]
       at (coogam) {~~~~$\scriptstyle\gamma_{n,m}^{-1}$~~~~~} ;
  \node[rounded corners,line width=\widthboxline,fill=blue!10,draw=\colorBdyfield]
       at (coomu) {$\scriptstyle\imu{}^{\phantom m}_{n,m,m}$} ;
  \begin{scope}[shift={(coodin)},rotate=180] \Dinatmor \end{scope}
  \node[yshift=-5pt] at (coodin) {$\scriptstyle \iota_m$} ;
  \end{tikzpicture}
  }
  \label{eq:alg9d-1=2}
  \ee
The equality \eqref{eq:alg9d-1=2} is indeed fulfilled -- it can again be deduced from 
basic features of our proposal. We present details of the proof in Appendix \ref{sec:app}.

\begin{Rem}
(i) Each of the morphisms $\dot\FF \oti \BB^{n,m} \,{\rarr~}\, \BB^{n,m}$ on the two sides of 
\eqref{eq:alg9d-1=2} endows the object $\BB^{n,m}$ of \C\ with a structure of left
$\dot\FF$-module. For the left hand side this follows by combining the identity 
\eqref{eq:alg9e} with the associativity of the boundary products $\imu{}_{p,q,r}$. For 
the right hand side one must invoke in addition the naturality of the half-braiding.
\\[2pt]
(ii) The fact that $\FF$ is an object in \ZC, i.e.\ that $\dot\FF \iN \C$ comes with
a half-braiding $\gamma$, allows us to exchange the order of a bulk and a boundary insertion
in an operator product, as when going from the right hand side of \eqref{eq:alg9d-1} to the
right hand side of \eqref{eq:alg9d-2}.  It is tempting to try to ``pull the morphism $\iota_m$
through the half-braiding'', which for $m \eq n$ would allow for an interpretation of the 
equality \eqref{eq:alg9d-1=2} as the statement that the image of $\iota_m$ is contained in the
center of the algebra $\BB^{m,m}$. However, this is not possible because as already pointed 
out, even though \C\ is braided, neither over- nor underbraiding is preferred, so that there 
is no natural way to exchange the two factors in an operator product of boundary insertions.
On the other hand, the idea does work in the case of two-dimensional \emph{topological}
field theory \cite{laza,laPf}. In this case \C\ is the category of vector spaces, which is
symmetric monoidal, so that over- and underbraiding coincide.
\end{Rem}

\medskip

Let us finally recall that in a detailed analysis of bulk-boundary systems, as in 
\cite{laza,laPf,kolR},
we must be careful about the distinction between incoming and outgoing field insertions. In
particular, besides the bulk-boundary OPE discussed above, which describes the situation
with an incoming bulk field and an outgoing boundary field, we also need to handle the
opposite situation in which the boundary field is incoming and the bulk field is outgoing.
Then in addition to the morphisms $\iota_m$ from $\FF$ to $\BB^{m,m}$ we also need morphisms
  \be
  \BB^{m,m} = \iHom(m,m) \rarr~ \relNat(\Id_\M,\Id_\M) = \FF
  \ee
in the opposite direction. Our approach supplies such morphisms as well. Indeed, according to
formula (5.8) in \cite{fuSc25}, the internal natural transformations for \emph{pivotal} module
categories \M\ and \N\ over a modular tensor category \C\ carry also the structure of a
\emph{coend}:\,%
 \footnote{~The situation considered in \cite{fuSc25} and thus the formula given
 there is more general. In our case it reduces to \eqref{eq:relNatcoend} because
 \C\ is in particular \emph{unimodular} and the relative Serre functors of \M\ and \N\
 (and as a consequence of unimodularity also their \emph{Nakayama functors}) are trivialized.}
  \be
  \relNat(F,G) \,\cong \int^{\!m\in\M}\! \iHom_\N(F(m),G(m)) \,.
  \label{eq:relNatcoend}
  \ee
The desired morphisms are provided by the structure morphisms of this coend.

It is not hard to see that the variants of the compatibility conditions \eqref{eq:relation9d} 
and \eqref{eq:relation9e} that are obtained when changing incoming to outgoing bulk field
insertions can be proven by using the structure morphisms of $\FF$ as a coend rather than
as an end, while changing incoming to outgoing boundary field insertions is accounted
for by using that the boundary objects $\BB^{m,n}$ and $\BB^{n,m}$ are each other's duals.
For instance, the in-out-reversed version of \eqref{eq:relation9d} holds because it 
expresses the canonical comultiplication on the coend $\FF$ in terms of the dinatural family
of the coend.


\section{Outlook} \label{sec:out}

Our proposal provides us naturally with candidates for the different types of operator 
products. Furthermore, recent developments in the theory of pivotal module categories 
\cite{shimi20,fuSc25} allow us to decide whether various consistency conditions are 
satisfied by the so obtained candidate operator products. In our opinion, the fact that 
all these consistency conditions are indeed met constitutes convincing evidence for the
viability of our proposal. Notably, we are confident that our arguments, being essentially
categorical, are very stable and have a good chance to survive in more general classes of 
conformal field theories. On the other hand, a lot of work remains to be done before 
full CFTs with non-semisimple chiral data are under control to an extent comparable to 
what has been achieved in the semisimple case.
Specifically, the following topics for future investigations impose themselves:
\Enumerate
	
\item 
Based on basic features of our proposal -- the facts that the object $\FF$ of bulk fields is 
a commutative symmetric Frobenius algebra, that the module category \M\ whose 
objects are the boundary conditions is pivotal, and that the multiplication on 
$\FF$ has a natural expression in terms of the dinatural structure morphisms of 
$\FF$ as an end -- the proposal automatically respects the relations
\eqref{eq:relation9a}\,--\,\eqref{eq:relation9e}, i.e.\ all genus-0 sewing constraints
in the list in Figure 9 of \cite{lewe3}. It is worth noting that the compliance with these
constraints can be verified entirely by considerations that comprise only local situations on
the world sheet, i.e.\ these sewing constraints are \emph{locally analyzable}. In contrast, 
the two genus-1 constraints -- modular invariance of one-bulk field correlators on a
torus and the Cardy condition for two-boundary field correlators on an annulus -- are
\emph{not} locally analyzable in this sense. As already pointed out in the Introduction,
they remain a challenge for our proposal. To establish their validity it will be necessary 
to get a better handle on their distinctive feature of being genuinely global.
 
\item 
Modular invariance of the one-bulk field correlators on a torus is equivalent to the 
Frobe\-ni\-us algebra $\FF \eq \relNat(\Id_\M,\Id_\M)$ being modular in the sense of
Definition 4.9 of \cite{fuSc22}. This property of $\FF$ has been established in the case 
that \C\ is semisimple, as well as \Cite{Cor.\,5.11\,\&\,Prop.\,6.1}{fuSs3} for particular
cases of module categories over representation categories of Hopf algebras (including the
Cardy case, which we have mentioned in  the last remark in Section \ref{sec:rmks}). But showing
it for any indecomposable pivotal module category \M\ over a modular tensor category appears
to be much harder than what the experience from the Hopf algebra case might seem to suggest.

\item 
The ultimate goal in the study of full finite CFTs is to show the existence of, and 
construct, a consistent set of correlators, for arbitrary collections of boundary fields 
and defect fields on world sheets of any topology, that is compatible with the proposed 
field content and the proposed OPEs. When the modular tensor category \C\ of chiral data 
is semisimple, this has already been achieved by the TFT construction of RCFT correlators 
\cite{fuRs4,fjfrs,fjfrs2}. For correlators of bulk fields on oriented world sheets without
boundary, a construction based on a Lego-Teichm\"uller game is available \cite{fuSc22},
provided that the object of bulk fields is a modular Frobenius algebra. However, the
Lego-Teichm\"uller game, as any presentation in terms of generators and relations, is difficult
to handle. To the best of our knowledge, it has not been developed for surfaces with defects.
	
\item 
When it comes to defect fields, we have only considered the situation that the field
is located on a single defect line, changing the defect condition along the line
analogously as boundary fields can change the boundary condition along a segment of the
boundary. There are, however, also more general defect fields which are located at the
junction of three or more defect lines (and similarly, generalized boundary fields located
at the junction of one or more defect lines and a boundary segment). To relate such fields 
to internal natural transformations it will be necessary to invoke suitable fusions of 
defect lines.
	
\item 
We expect that  new insights will be gained by combining the structures exhibited in the
present paper with a string-net approach to conformal blocks of the Drinfeld center
\ZC. Such a description has been explored, in the case of semisimple \C, in \cite{scYa}
for bulk field correlators in the Cardy case. In \cite{traub} it has further been shown
how to construct correlators of bulk and boundary fields for a fixed boundary condition
once a bulk algebra (as a modular Frobenius algebra) as well as a compatible boundary
algebra are given.  For finite spherical categories that are not semisimple,
no string-net construction is known so far. It appears to be a promising task
to try to accomplish such non-semisimple string-net constructions.

\item 
As already stressed, in the present paper we have restricted our attention to conformal 
field theories whose chiral data are encoded in a, possible non-semisimple,
modular tensor category. Eventually we would like to extend our analysis to cover also
conformal field theories for which the category of chiral data, while still being 
finite as an abelian category, is no longer modular, e.g.\ not rigid and with
a non-exact tensor product. Examples of full conformal field theories of this type 
have been discussed in the literature, see e.g.\ \cite{garW}. We believe
that our proposal is structurally very stable and that similar structures will
still be present in such more general classes of conformal field theories.

\item 
Any pivotal module category is in particular an exact module category. This is a very
strong property, e.g.\ for semisimple \C, exactness requires \M\ to be semisimple as well.
One may speculate that module categories of finite (relative) homological 
dimension could be used in more general constructions. An analysis of this issue will
considerably transcend the mathematical setting of the present paper. 

\end{enumerate}


\vskip 1.4em

\noindent
{\sc Acknowledgements:}\\[.3em]
We thank Yang Yang for helpful comments on the manuscript.
JF is supported by VR under project no.\ 2017-03836. CS is partially supported by
the Deutsche Forschungsgemeinschaft (DFG, German Research Foundation) under Germany's
Excellence Strategy - EXC 2121 ``Quantum Universe'' - 390833306 and under SCHW\,1162/6-1.
Part of this work was done while the second author was participating in
the ESI program ``Higher Structures and Field Theory'' in August 2021;
we thank ESI and the organizers for their hospitality.

\vskip 3.5em

\appendix

\section{Half-braiding, end structure, and commutativity of the bulk algebra} \label{sec:app}

The purpose of this appendix is twofold. First we show the validity of the equality
\eqref{eq:alg9d-1=2} of morphisms, which constitutes the algebraic formulation of the
sewing condition \eqref{eq:relation9d}. Second, we use that equality to obtain a proof 
of commutativity of the bulk algebra; our strategy is similar to the one in the proof of
Theorem 4.9 of \cite{shimi17}. We start by noting that the half-braiding $\gamma$ on $\FF$
is determined by its structure of an end, so that by Theorem 8 of \cite{fuSc25} we have
     \def\locpx  {1.5}
     \def\locpy  {4.1}
     \def\widthDinatmor {0.33}
  \be
  \raisebox{-6.6em} {\begin{tikzpicture}
  \coordinate (coodin) at (\locpx,0.79*\locpy) ;
  \coordinate (coob)   at (\locpx,0) ;
  \coordinate (coob2)  at (\locpx,\locpy) ;
  \coordinate (cooc)   at (0,0) ;
  \coordinate (cooc2)  at (0,\locpy) ;
  \coordinate (coogam) at (0.49*\locpx,0.28*\locpy) ;
  \scopeArrow{0.15}{\arrowDefect} \scopeArrow{0.95}{\arrowDefect}
  \draw[\colorC,line width=\widthObj,postaction={decorate}]
       (cooc) node[below=-1pt] {$\dot\FF$} -- (cooc2) node[above=-1pt] {$\BB^{n,m}$} ;
  \end{scope} \end{scope}
  \scopeArrow{0.15}{\arrowDefect} \scopeArrow{0.59}{\arrowDefect} \scopeArrow{0.95}{\arrowDefect}
  \draw[\colorC,line width=\widthObj,postaction={decorate}]
       (coob) node[below]{$\BB^{n.m}$} -- node[midway,right=-2pt,yshift=-5pt] {$\dot\FF$}
       (coob2) node[above=-1pt] {$\BB^{m.m}$} ;
  \end{scope} \end{scope} \end{scope}
  \node[rounded corners,line width=\widthboxline,fill=green!15,draw=green!70!black]
       at (coogam) {~~~~$\scriptstyle\gamma_{n,m}^{-1}$~~~~~} ;
  \begin{scope}[shift={(coodin)},rotate=180] \Dinatmor \end{scope}
  \node[yshift=-5pt] at (coodin) {$\scriptstyle \iota_m$} ;
  \end{tikzpicture}
  }
  \qquad = \quad~
     \def\widthDinatmor {0.43}
  \raisebox{-6.6em} {\begin{tikzpicture}
  \coordinate (coodin) at (0,0.41*\locpy) ;
  \coordinate (cooa1)  at (-0.23,0.41*\locpy) ;
  \coordinate (cooa2)  at (-0.23,0.55*\locpy) ;
  \coordinate (cooa3)  at (-0.5*\locpx,0.85*\locpy) ;
  \coordinate (cooa4)  at (-0.5*\locpx,\locpy) ;
  \coordinate (coob)   at (1.05*\locpx,0) ;
  \coordinate (coob0)  at (0.5*\locpx,0.61*\locpy) ;
  \coordinate (coob1)  at (0.23,0.41*\locpy) ;
  \coordinate (coob2)  at (0,0.55*\locpy) ;
  \coordinate (coob3)  at (0.5*\locpx,0.85*\locpy) ;
  \coordinate (coob4)  at (0.5*\locpx,\locpy) ;
  \coordinate (cooc)   at (0,0) ;
  \scopeArrow{0.38}{\arrowDefect}
  \draw[\colorC,line width=\widthObj,postaction={decorate}]
       (cooc) node[below=-1pt] {$\dot\FF$} -- (coodin) ;
  \end{scope}
  \scopeArrow{0.88}{\arrowDefect}
  \draw[\colorC,line width=\widthObj,postaction={decorate}]
       (cooa1) -- (cooa2) [out=90,in=270] to (cooa3) -- (cooa4) node[above=-1pt] {$\BB^{n,m}$} ;
  \draw[\colorC,line width=\widthObj,postaction={decorate}]
       (coodin) -- (coob2) [out=90,in=270] to (coob3) -- (coob4) node[above=-1pt] {$\BB^{m,m}$} ;
  \end{scope}
  \draw[\colorC,line width=\widthObj]
       (coob) node[below=-1pt] {$\BB^{n,m}$} [out=90,in=0] to (coob0) [out=180,in=90] to (coob1) ;
  \scopeArrow{0.99}{\arrowDefect}
  \draw[\colorC,line width=\widthObj,postaction={decorate}] (coob0) -- ++(-0.1,0) ;
  \end{scope}
  \begin{scope}[shift={(coodin)},rotate=180] \Dinatmor \end{scope}
	  \node[yshift=-5pt] at (coodin) {$\scriptstyle \iota_{c\Act m}$} ;
  \end{tikzpicture}
  }
  \label{eq:gammaF}
  \ee
Here $\gamma_{n,m}^{\phantom m}$ is the $\BB^{n,m}$-component of the half-braiding $\gamma$,
and we set $c \,{:=}\, \BB^{n,m}$ whereby, recalling that the internal Hom is a bimodule 
functor, see \eqref{eq:iHomBimod}, we have 
$\BB^{c\Act m,c\Act m} \eq \BB^{n,m} \oti \BB^{m,m}\oti (\BB^{n,m})^\vee$.
The morphism $(\BB^{n,m})^\vee \oti \BB^{n,m} \,{\to}\, \one_\C$ 
on the right hand side is an evaluation morphism in \C.

As a second ingredient we invoke the dinaturality of the family $(\iota_m)_{m\in\M}$.
Applied to the evaluation morphism $\iev{}_{m,n}$, it states that the two composite morphisms 
  \be
  \begin{array}{lr}
  ~ & \FF \rarr{~\iota_{n}~} \BB^{n,n} \rarr{\,\iHom(\iev{}_{m,n},\id_n)\,} \BB^{n,c\Act m}
  \Nxl3
  \text{and} \quad & \FF \rarr{\,\iota_{c\Act m}\,} \BB^{c\Act m,c\Act m} 
  \rarr{\,\iHom(\id_{c\Act m},\iev{}_{m,n})\,} \BB^{n,c\Act m}
  \eear
  \label{eq:dinat}
  \ee
coincide for any $c\iN \C$. Now note that in our case we have
$\BB^{n,c\Act m} \,{\equiv}\; \BB^{n,\BB^{n,m}\Act m} \eq \BB^{n,m} \oti (\BB^{n,m})^\vee$,
and further that (using \eqref{eq:tworems1})
  \be
  \iHom(\id^{}_{\BB^{n,m}\Act m},\iev{}_{m,n}) = \imu{}_{n,m,m} \oti \id^{}_{\BB^{m,n}}
  \ee  
as well as
  \be
  \iHom(\iev{}_{m,n},\id^{}_n) = \iDelta{}_{n,m,n} \,,
  \ee
with $\iDelta{}_{p,q,r}\colon \BB^{p,r} \,{\to}\, \BB^{p,q} \oti \BB^{q,r}$ the comultiplication
of boundary objects. It follows that the equality of the two morphisms \eqref{eq:dinat} amounts to
     \def\locpx  {1.5}
     \def\locpy  {4.1}
     \def\widthDinatmor {0.33}
  \be
  \raisebox{-6.9em} {\begin{tikzpicture}
  \coordinate (cooa1)  at (-0.23,0.83*\locpy) ;
  \coordinate (cooa2)  at (-0.23,1.25*\locpy) ;
  \coordinate (coob)   at (1.28*\locpx,0) ;
  \coordinate (coobb)  at (1.28*\locpx,0.3*\locpy) ;
  \coordinate (coob0)  at (0.45*\locpx,1.08*\locpy) ;
  \coordinate (coob1)  at (0.23,0.8*\locpy) ;
  \coordinate (cooc)   at (0,0) ;
  \coordinate (coomu)  at (0,0.8*\locpy) ;
  \coordinate (coodin) at (0,0.4*\locpy) ;
  \scopeArrow{0.18}{\arrowDefect} \scopeArrow{0.78}{\arrowDefect}
  \draw[\colorC,line width=\widthObj,postaction={decorate}]
       (cooc) node[below=-1pt] {$\dot\FF$}
       -- node[sloped,near end,below,xshift=-2pt] {$\BB^{n,n}$} (coomu) ;
  \end{scope} \end{scope}
  \draw[\colorC,line width=\widthObj] (coob) node[below=-1pt] {$\BB^{n,m}$}
       -- (coobb) [out=90,in=0] to (coob0) [out=180,in=90] to (coob1) ;
  \scopeArrow{0.99}{\arrowDefect}
  \draw[\colorC,line width=\widthObj,postaction={decorate}] (coob0) -- ++(-0.1,0) ;
  \end{scope}
  \scopeArrow{0.85}{\arrowDefect}
  \draw[\colorC,line width=\widthObj,postaction={decorate}]
       (cooa1) -- (cooa2) node[above=-2pt] {$\BB^{n,m}$} ;
  \end{scope}
  \node[rounded corners,line width=\widthboxline,fill=blue!10,draw=\colorBdyfield]
       at (coomu) {$\scriptstyle\iDelta{}^{\phantom t}_{n,m,n}$} ;
  \begin{scope}[shift={(coodin)},rotate=180] \Dinatmor \end{scope}
  \node[yshift=-5pt] at (coodin) {$\scriptstyle \iota_n$} ;
  \end{tikzpicture}
  }
  \quad = \quad
     \def\widthDinatmor {0.44}
  \raisebox{-6.9em} {\begin{tikzpicture}
  \coordinate (coodin) at (0,0.41*\locpy) ;
  \coordinate (cooa1)  at (-0.23,0.41*\locpy) ;
  \coordinate (cooa2)  at (-0.23,0.53*\locpy) ;
  \coordinate (cooa3)  at (-0.35*\locpx,0.77*\locpy) ;
  \coordinate (cooa4)  at (-0.35*\locpx,0.95*\locpy) ;
  \coordinate (coob)   at (1.14*\locpx,0) ;
  \coordinate (coob0)  at (0.5*\locpx,0.61*\locpy) ;
  \coordinate (coob1)  at (0.23,0.41*\locpy) ;
  \coordinate (coob2)  at (0,0.53*\locpy) ;
  \coordinate (coob3)  at (0.35*\locpx,0.77*\locpy) ;
  \coordinate (coob4)  at (0.35*\locpx,0.95*\locpy) ;
  \coordinate (cooc)   at (0,0) ;
  \coordinate (coomu)  at (0,0.95*\locpy) ;
  \coordinate (cootop) at (0,1.25*\locpy) ;
  \scopeArrow{0.38}{\arrowDefect}
  \draw[\colorC,line width=\widthObj,postaction={decorate}]
       (cooc) node[below=-1pt] {$\dot\FF$} -- (coodin) ;
  \end{scope}
  \scopeArrow{0.79}{\arrowDefect}
  \draw[\colorC,line width=\widthObj,postaction={decorate}] (cooa1) -- (cooa2) [out=90,in=270]
       to node[midway,sloped,below=-1pt,xshift=-11pt] {$\BB^{n,m}$} (cooa3) -- (cooa4) ;
  \draw[\colorC,line width=\widthObj,postaction={decorate}]
       (coodin) -- (coob2) [out=90,in=270]
       to node[near end,sloped,below=-2pt,xshift=12pt] {$\BB^{m,m}$} (coob3) -- (coob4) ;
  \end{scope}
  \scopeArrow{0.91}{\arrowDefect}
  \draw[\colorC,line width=\widthObj,postaction={decorate}]
       (coomu) -- (cootop) node[above=-2pt] {$\BB^{n,m}$} ;
  \end{scope}
  \draw[\colorC,line width=\widthObj]
       (coob) node[below=-1pt] {$\BB^{n,m}$} [out=90,in=0] to (coob0) [out=180,in=90] to (coob1) ;
  \scopeArrow{0.99}{\arrowDefect}
  \draw[\colorC,line width=\widthObj,postaction={decorate}] (coob0) -- ++(-0.1,0) ;
  \end{scope}
  \node[rounded corners,line width=\widthboxline,fill=blue!10,draw=\colorBdyfield]
       at (coomu) {~$\scriptstyle\imu{}^{\phantom m}_{n,m,m}$~} ;
  \begin{scope}[shift={(coodin)},rotate=180] \Dinatmor \end{scope}
	  \node[yshift=-5pt] at (coodin) {$\scriptstyle \iota_{c\Act m}$} ;
  \end{tikzpicture}
  }
  \label{eq:dinat+}
  \ee
where we also use an additional evaluation morphism in \C\ to bend the outgoing 
$(\BB^{n,m})^\vee$-line to an incoming $\BB^{n,m}$-line.

Now owing to the description \eqref{eq:gammaF} of the half-braiding of $\FF$ the 
right hand side of \eqref{eq:dinat+} equals the right hand side of the identity
\eqref{eq:alg9d-1=2} that we want to prove. Concerning the left hand side we note that
the evaluation morphism   $(\BB^{n,m})^\vee \oti \BB^{n,m} \,{\to}\, \one_\C$ 
in \C\ can be expressed 
in terms of the algebra and coalgebra structures as $\underline\varepsilon{}_m^{} 
\cir \imu{}_{m,n,m}$. After doing so, one can use the Frobenius relation
     \def\locpx  {1.5}
     \def\locpy  {4.1}
     \def\widthDinatmor {0.33}
  \be
  \raisebox{-8.1em} {\begin{tikzpicture}
  \coordinate (cooa1)  at (-0.23,0.75*\locpy) ;
  \coordinate (cooa2)  at (-0.23,1.57*\locpy) ;
  \coordinate (coob)   at (1.2*\locpx,0.22*\locpy) ;
  \coordinate (coobb)  at (1.2*\locpx,0.4*\locpy) ;
  \coordinate (coob0)  at (0.6*\locpx,1.2*\locpy) ;
  \coordinate (coob1)  at (0.23,0.7*\locpy) ;
  \coordinate (cooc)   at (0,0.22*\locpy) ;
  \coordinate (coomu)  at (0,0.67*\locpy) ;
  \coordinate (coomu2) at (0.6*\locpx,1.12*\locpy) ;
  \coordinate (coomu3) at (0.6*\locpx,1.42*\locpy) ;
  \scopeArrow{0.42}{\arrowDefect}
  \draw[\colorC,line width=\widthObj,postaction={decorate}]
       (cooc) node[below=-1pt] {$\BB^{n,n}$} -- (coomu) ;
  \end{scope}
  \scopeArrow{0.47}{\arrowDefect}
  \draw[\colorC,line width=\widthObj,postaction={decorate}]
       (coob) node[below=-1pt] {$\BB^{n,m}$} -- (coobb) [out=90,in=290] to (coob0);
  \draw[\colorC,line width=\widthObj,postaction={decorate}]
       (coob1) [out=99,in=260] to node[sloped,midway,below,xshift=-3pt] {$\BB^{m,n}$} (coob0) ;
  \end{scope}
  \scopeArrow{0.68}{\arrowDefect}
  \draw[\colorC,line width=\widthObj,postaction={decorate}] (coomu2) -- (coomu3) ;
  \end{scope}
  \scopeArrow{0.85}{\arrowDefect}
  \draw[\colorC,line width=\widthObj,postaction={decorate}]
       (cooa1) -- (cooa2) node[above=-2pt] {$\BB^{n,m}$} ;
  \end{scope}
  \node[rounded corners,line width=\widthboxline,fill=blue!10,draw=\colorBdyfield]
       at (coomu) {$\scriptstyle\iDelta{}^{\phantom t}_{n,m,n}$} ;
  \node[rounded corners,line width=\widthboxline,fill=blue!10,draw=\colorBdyfield]
       at (coomu2) {$\scriptstyle\imu{}^{\phantom m}_{m,n,m}$} ;
  \node[rounded corners,line width=\widthboxline,fill=blue!10,draw=\colorBdyfield]
       at (coomu3) {$\scriptstyle\underline\varepsilon{}_m^{}$} ;
  \end{tikzpicture}
  }
  \quad = \qquad
  \raisebox{-8.1em} {\begin{tikzpicture}
  \coordinate (cooev2) at (0.44*\locpx,1.57*\locpy) ;
  \coordinate (coomu)  at (0.44*\locpx,0.89*\locpy) ;
  \coordinate (coob)   at (0.88*\locpx,0.22*\locpy) ;
  \coordinate (coob2)  at (0.88*\locpx,0.39*\locpy) ;
  \coordinate (cooc)   at (0,0.22*\locpy) ;
  \coordinate (cooc2)  at (0,0.39*\locpy) ;
  \scopeArrow{0.45}{\arrowDefect}
  \draw[\colorC,line width=\widthObj,postaction={decorate}]
       (cooc) node[below=-1pt]{$\BB^{n,n}$} -- (cooc2) [out=90,in=238] to (coomu) ;
  \draw[\colorC,line width=\widthObj,postaction={decorate}]
       (coob) node[below]{$\BB^{n.m}$} -- (coob2) [out=90,in=302] to (coomu) ;
  \end{scope}
  \scopeArrow{0.79}{\arrowDefect}
  \draw[\colorC,line width=\widthObj,postaction={decorate}]
       (coomu) -- (cooev2) node[above=-2pt] {$\BB^{n,m}$} ;
  \end{scope}
  \node[rounded corners,line width=\widthboxline,fill=blue!10,draw=\colorBdyfield]
       at (coomu) {\,$\scriptstyle\imu{}^{\phantom m}_{n,n,m}$} ;
  \end{tikzpicture}
  }
  \ee
to see that the left hand side of \eqref{eq:dinat+} equals the left hand side
of \eqref{eq:alg9d-1=2}, thereby completing the proof of \eqref{eq:alg9d-1=2}.
 
\medskip

Next we note that the self-braiding of the bulk algebra $\FF$ in \ZC\ is given by the
component $\gamma_{\dot\FF}$ of the half-braiding $\gamma$. Commutativity of the
bulk algebra product $\imu$ thus means that $\imu \cir \gamma_{\dot\FF} \eq \imu$
or, what is the same, $\imu \cir \gamma_{\dot\FF}^{-1} \eq \imu$.
Owing to the universal property of the end this, in turn, is equivalent to having
  \be
  \iota_m \circ \imu \circ \gamma_{\dot\FF}^{-1} = \iota_m \circ \imu 
  \label{eq:imc=im}
  \ee
for every $m\iN \M$. Now the left hand side of \eqref{eq:imc=im} can be rewritten as
     \def\locpx  {3.2}
     \def\locpy  {4.9}
     \def\widthDinatmor {0.33}
  \be
  \raisebox{-7.4em} {\begin{tikzpicture}
  \coordinate (coodin) at (0.23*\locpx,0.94*\locpy) ;
  \coordinate (coomu)  at (0.23*\locpx,0.61*\locpy) ;
  \coordinate (cooc)   at (0,0) ;
  \coordinate (cooc2)  at (0,0.35*\locpy) ;
  \coordinate (cood)   at (0.46*\locpx,0) ;
  \coordinate (cood2)  at (0.46*\locpx,0.35*\locpy) ;
  \coordinate (coogam) at (0.23*\locpx,0.24*\locpy) ;
  \coordinate (cootop) at (0.23*\locpx,1.12*\locpy) ;
  \scopeArrow{0.16}{\arrowDefect} \scopeArrow{0.78}{\arrowDefect}
  \draw[\colorC,line width=\widthObj,postaction={decorate}]
       (cooc) node[below=-1pt] {$\dot\FF$} --
       (cooc2) node[above=1.3pt,xshift=-4pt] {$\dot\FF$} [out=90,in=238] to (coomu) ;
  \draw[\colorC,line width=\widthObj,postaction={decorate}]
       (cood) node[below=-1pt]{$\dot\FF$} --
       (cood2) node[above=1.3pt,xshift=5pt] {$\dot\FF$} [out=90,in=302] to (coomu) ;
  \end{scope} \end{scope}
  \scopeArrow{0.35}{\arrowDefect} \scopeArrow{0.93}{\arrowDefect}
  \draw[\colorC,line width=\widthObj,postaction={decorate}] (coomu) --
       node[near start,left=-2pt,yshift=7pt] {$\dot\FF$} (cootop) node[above=-2pt] {$\BB^{m,m}$} ;
  \end{scope} \end{scope}
  \node[rounded corners,line width=\widthboxline,fill=green!15,draw=green!70!black]
       at (coogam) {~~~~~$\scriptstyle\gamma_{\dot\FF}^{-1}$~~~~~} ;
  \node[rounded corners,line width=\widthboxline,fill=blue!10,draw=\colorBdyfield]
       at (coomu) {$~\,\scriptstyle\imu{}^{\phantom m}_{\phantom m}$} ;
  \begin{scope}[shift={(coodin)},rotate=180] \Dinatmor \end{scope}
  \node[yshift=-5pt] at (coodin) {$\scriptstyle \iota_m$} ;
  \end{tikzpicture}
  }
  \qquad = \qquad
  \raisebox{-7.4em} {\begin{tikzpicture}
  \coordinate (coodi0) at (0,0.56*\locpy) ;
  \coordinate (coodin) at (0.46*\locpx,0.56*\locpy) ;
  \coordinate (coomu)  at (0.23*\locpx,0.86*\locpy) ;
  \coordinate (cooc)   at (0,0) ;
  \coordinate (cooc2)  at (0,0.57*\locpy) ;
  \coordinate (cood)   at (0.46*\locpx,0) ;
  \coordinate (cood2)  at (0.46*\locpx,0.57*\locpy) ;
  \coordinate (coogam) at (0.23*\locpx,0.22*\locpy) ;
  \coordinate (cootop) at (0.23*\locpx,1.12*\locpy) ;
  \scopeArrow{0.12}{\arrowDefect} \scopeArrow{0.44}{\arrowDefect} \scopeArrow{0.82}{\arrowDefect}
  \draw[\colorC,line width=\widthObj,postaction={decorate}]
       (cooc) node[below=-1pt] {$\dot\FF$} -- node[near end,left=-2pt,yshift=-4pt] {$\dot\FF$}
       (cooc2) [out=90,in=238] to node[sloped,midway,above=-3pt,xshift=-8pt] {$\BB^{m,m}$} (coomu) ;
  \draw[\colorC,line width=\widthObj,postaction={decorate}]
       (cood) node[below=-1pt]{$\dot\FF$} -- node[near end,right=-1pt,yshift=-4pt] {$\dot\FF$}
       (cood2) [out=90,in=302] to node[midway,sloped,below=-17pt,xshift=9pt] {$\BB^{m,m}$} (coomu) ;
  \end{scope} \end{scope} \end{scope}
  \scopeArrow{0.81}{\arrowDefect}
  \draw[\colorC,line width=\widthObj,postaction={decorate}]
       (coomu) -- (cootop) node[above=-2pt] {$\BB^{m,m}$} ;
  \end{scope}
  \node[rounded corners,line width=\widthboxline,fill=green!15,draw=green!70!black]
       at (coogam) {~~~~~$\scriptstyle\gamma_{\dot\FF}^{-1}$~~~~~} ;
  \node[rounded corners,line width=\widthboxline,fill=blue!10,draw=\colorBdyfield]
       at (coomu) {$\scriptstyle\imu{}^{\phantom m}_{m,m,m}$} ;
  \begin{scope}[shift={(coodi0)},rotate=180] \Dinatmor \end{scope}
  \begin{scope}[shift={(coodin)},rotate=180] \Dinatmor \end{scope}
  \node[yshift=-5pt] at (coodi0) {$\scriptstyle \iota_m$} ;
  \node[yshift=-5pt] at (coodin) {$\scriptstyle \iota_m$} ;
  \end{tikzpicture}
  }
  \qquad = \qquad
  \raisebox{-7.4em} {\begin{tikzpicture}
  \coordinate (coodi2) at (0.46*\locpx,0.70*\locpy) ;
  \coordinate (coodin) at (0.46*\locpx,0.24*\locpy) ;
  \coordinate (coomu)  at (0.23*\locpx,0.91*\locpy) ;
  \coordinate (cooc)   at (0,0) ;
  \coordinate (cooc2)  at (0,0.61*\locpy) ;
  \coordinate (cood)   at (0.46*\locpx,0) ;
  \coordinate (coogam) at (0.23*\locpx,0.42*\locpy) ;
  \coordinate (cootop) at (0.23*\locpx,1.12*\locpy) ;
  \scopeArrow{0.16}{\arrowDefect} \scopeArrow{0.83}{\arrowDefect}
  \draw[\colorC,line width=\widthObj,postaction={decorate}]
       (cooc) node[below=-1pt]{$\dot\FF$} -- 
       (cooc2) [out=90,in=238] to node[sloped,midway,above=-4pt,xshift=-7pt] {$\BB^{m,m}$} (coomu) ;
  \end{scope} \end{scope}
  \scopeArrow{0.11}{\arrowDefect} \scopeArrow{0.34}{\arrowDefect}
  \scopeArrow{0.60}{\arrowDefect} \scopeArrow{0.83}{\arrowDefect}
  \draw[\colorC,line width=\widthObj,postaction={decorate}]
       (cood) node[below=-1pt]{$\dot\FF$} -- node[midway,right=-1.3pt,yshift=-5pt] {$\BB^{m,m}$}
       node[near end,right=-0.5pt,yshift=4pt] {$\dot\FF$}
       (coodi2) [out=90,in=302] to node[midway,sloped,below=-15pt,xshift=13pt] {$\BB^{m,m}$} (coomu) ;
  \end{scope} \end{scope} \end{scope} \end{scope}
  \scopeArrow{0.84}{\arrowDefect}
  \draw[\colorC,line width=\widthObj,postaction={decorate}]
       (coomu) -- (cootop) node[above=-2pt] {$\BB^{m,m}$} ;
  \end{scope}
  \node[rounded corners,line width=\widthboxline,fill=green!15,draw=green!70!black]
       at (coogam) {~~~~\,$\scriptstyle\gamma_{m,m}^{-1}$~~~~\,} ;
  \node[rounded corners,line width=\widthboxline,fill=blue!10,draw=\colorBdyfield]
       at (coomu) {$\scriptstyle\imu{}^{\phantom m}_{m,m,m}$} ;
  \begin{scope}[shift={(coodi2)},rotate=180] \Dinatmor \end{scope}
  \begin{scope}[shift={(coodin)},rotate=180] \Dinatmor \end{scope}
  \node[yshift=-5pt] at (coodi2) {$\scriptstyle \iota_m$} ;
  \node[yshift=-5pt] at (coodin) {$\scriptstyle \iota_m$} ;
  \end{tikzpicture}
  }
  \label{eq:comm-1}
  \ee
where the first equality holds by the compatibility \eqref{eq:alg9e} of bulk and boundary
products, while the second equality implements the functoriality of the half-braiding.
By invoking the equality \eqref{eq:alg9d-1=2} (specialized to $n \eq m$), the morphism on
the right hand side of \eqref{eq:comm-1} can be rewritten as 
$\imu{}^{\phantom m}_{m,m,m} \cir (\iota_m \oti \iota_m)$. Using once again \eqref{eq:alg9e}
this, in turn, equals the right hand side of \eqref{eq:imc=im}, and thus proves
commutativity of the bulk product $\imu$.

\newpage

\newcommand\wb{\,\linebreak[0]} \def\wB {$\,$\wb}
\newcommand\Bi[2]    {\bibitem[#2]{#1}}
\newcommand\inBo[8]  {{\em #8}, in:\ {\em #1}, {#2}\ ({#3}, {#4} {#5}), p.\ {#6--#7} }
\newcommand\inBO[9]  {{\em #9}, in:\ {\em #1}, {#2}\ ({#3}, {#4} {#5}), p.\ {#6--#7} {\tt [#8]}}
\newcommand\J[7]     {{\em #7}, {#1} {#2} ({#3}) {#4--#5} {{\tt [#6]}}}
\newcommand\JO[6]    {{\em #6}, {#1} {#2} ({#3}) {#4--#5} }
\newcommand\JP[7]    {{\em #7}, {#1} ({#3}) {{\tt [#6]}}}
\newcommand\Jpress[7]{{\em #7}, {#1} {} (in press) {} {{\tt [#6]}}}
\newcommand\BOOK[4]  {{\em #1\/} ({#2}, {#3} {#4})}
\newcommand\PhD[2]   {{\em #2}, Ph.D.\ thesis #1}
\newcommand\Prep[2]  {{\em #2}, preprint {\tt #1}}
\newcommand\uPrep[2] {{\em #2}, unpublished preprint {\tt #1}}
\def\adma  {Adv.\wb Math.}
\def\alrt  {Algebr.\wb Represent.\wB Theory}         
\def\apcs  {Applied\wB Cate\-go\-rical\wB Struc\-tures}
\def\atmp  {Adv.\wb Theor.\wb Math.\wb Phys.}   
\def\comp  {Com\-mun.\wb Math.\wb Phys.}
\def\coma  {Con\-temp.\wb Math.}
\def\cpma  {Com\-pos.\wb Math.}
\def\ijmp  {Int.\wb J.\wb Mod.\wb Phys.\ A}
\def\joal  {J.\wB Al\-ge\-bra}
\def\jopa  {J.\wb Phys.\ A}
\def\jpaa  {J.\wB Pure\wB Appl.\wb Alg.}
\def\jram  {J.\wB rei\-ne\wB an\-gew.\wb Math.}
\def\nupb  {Nucl.\wb Phys.\ B}
\def\phlb  {Phys.\wb Lett.\ B}
\def\phrl  {Phys.\wb Rev.\wb Lett.}
\def\pnas  {Proc.\wb Natl.\wb Acad.\wb Sci.\wb USA}
\def\sigm  {SIGMA}
\def\taac  {Theo\-ry\wB and\wB Appl.\wb Cat.}
\def\tams  {Trans.\wb Amer.\wb Math.\wb Soc.}
\def\toap  {Topology\wB Applic.}
\def\trgr  {Trans\-form.\wB Groups}

\end{document}